\documentclass[11pt]{article}

\usepackage[margin=1in]{geometry}

\usepackage{amsmath,amsthm,amssymb,mathtools}

\usepackage{graphicx}
\usepackage{caption}
\usepackage{subcaption}
\usepackage{xcolor}
\usepackage{float}
\usepackage{booktabs}
\usepackage{multirow}
\usepackage{hyperref}
\usepackage{stackengine}
\usepackage{appendix}

\usepackage{tikz}
\usetikzlibrary{positioning}

\usepackage[T1]{fontenc}
\usepackage[authoryear,sort]{natbib}
\graphicspath{{fig/}}

\mathtoolsset{showonlyrefs=true}

\newtheorem{theorem}{Theorem}[section]
\newtheorem{lemma}[theorem]{Lemma}
\newtheorem{prop}[theorem]{Proposition}
\newtheorem{corollary}[theorem]{Corollary}
\newtheorem{definition}[theorem]{Definition}

\counterwithin{figure}{section}
\counterwithin{equation}{section}
\counterwithin{table}{section}

\def\cG{{\mathcal G}}
\def\cL{{\mathcal L}}

\newcommand{\norm}[1]{\left\lVert#1\right\rVert}

\renewcommand{\bf}[1]{\textbf{#1}}

\newcommand{\change}[1]{#1}

\begin{document}

\vspace{-2em}

\title{Towards Consensus:\\ Reducing Polarization by Perturbing Social Networks}

\author{Miklos Z. R\'acz\thanks{ORFE, Princeton University; \url{mracz@princeton.edu}. Research supported in part by NSF grant DMS 1811724 and by a Princeton SEAS Innovation Award.} 
\and 
Daniel E. Rigobon\thanks{ORFE, Princeton University; \url{drigobon@princeton.edu}. Research supported in part by NSF grant DMS 1811724 and by a Princeton SEAS Innovation Award.}
}

\date{\today}

\maketitle

\vspace{-0.8cm}
\begin{center}
    \renewcommand{\UrlFont}{\normalsize}
    \begin{abstract}
        This paper studies how a centralized planner can modify the structure of a social or information network to reduce polarization. First, polarization is found to be highly dependent on degree and structural properties of the network \change{-- including the well-known isoperimetric number (i.e., Cheeger constant)}. We then formulate the planner's problem under full information, and motivate disagreement-seeking and coordinate descent heuristics. A novel setting for the planner in which the population's innate opinions are adversarially chosen is introduced, and shown to be equivalent to maximization of the Laplacian's spectral gap. We prove bounds for the effectiveness of a strategy that adds edges between vertices on opposite sides of the cut induced by the spectral gap's eigenvector. Finally, these strategies are evaluated on six real-world and synthetic networks. In several networks, we find that polarization can be significantly reduced through the addition of a small number of edges.
    \end{abstract}
\end{center}

\newpage
\section{Introduction}

In recent years there has been a substantial increase in sociopolitical polarization -- it is clear that our society does not agree on issues in politics, science, healthcare, and beyond. Counter-intuitively, this has been accompanied by the growth of social media platforms; individuals are connecting with others and sharing information more than ever before. How is it that ``bringing the world closer together''\footnote{The original mission statement of Facebook.} resulted in our opinions drifting further apart?

This phenomenon is a byproduct of the structure of our social networks; a greater number of connections does not necessarily reflect a closeness to consensus. It is possible for the proliferation of social media to reduce one's exposure to other opinions, and thereby entrench them in a community of like-minded users. This feature is known as an ``echo chamber,'' and has been found to emerge through the incentives of recommender systems rewiring the network~\citep{Chitra2020}. Furthermore, confirmation bias and structural similarity have been found to contribute to increases in polarization as the structure of the network evolves~\citep{Bhalla2021, santos2021link}. Therefore, \emph{how} the population is connected -- as opposed to how connected the population \emph{is} -- may be most important to the emergence of polarization.

In this paper, we seek an understanding of how a network planner can reduce polarization by changing the structure of a population's social or information network. To that end, we present a model of \emph{budgeted network perturbation}, where the planner is given a small budget with which to modify the structure of a given network. We study the planner's problem in two different settings, and evaluate simple heuristics on both real-world and synthetic networks.

There has been a significant research effort towards reducing polarization in networks~\citep{Chen2018a, garimella2017reducing, haddadan2021repbublik, matakos2017measuring,rahaman2021model}. In contrast to both~\citet{matakos2017measuring} and~\citet{rahaman2021model}, we hold fixed the population's opinions -- while allowing the network structure to be modified. This paper differs from~\citet{garimella2017reducing} and~\citet{haddadan2021repbublik} in both our use of a distinct measure of polarization and incorporating opinion dynamics. Finally, we improve upon the closely related work of~\citet{Chen2018a} through a more detailed theoretical analysis of edge effects, \change{consideration of weighted networks,} and study of larger datasets.

A very similar paper to our own is recent work by~\citet{zhu2021minimizing}, where the authors present a variation of the problem studied in~\citet{Musco2018}. Both these studies aim to minimize the sum of polarization and disagreement by changing the network structure, but~\citet{zhu2021minimizing} impose a budget that ensures only a small number of edges can be changed. These authors use a similar budget constraint to our own, but their polarization-disagreement index varies greatly with the edge density of the graph. Although it is convenient for analysis and computation, their index is inadequate for capturing the dynamics of polarization alone. Nonetheless, we believe the formulation in this paper and~\citet{zhu2021minimizing} to be practical. The network structure is not assumed to be completely malleable, but small changes are permitted. For instance, while social media platforms such as Facebook or Twitter cannot dictate who an individual chooses to `friend' or `follow', these platforms can curate an individual's feed to change one's relative exposure levels to certain content. This process perturbs the structure of external influence on an individual, so that it differs from their endogenously created network of `friends' or `follows'. If, instead, any of these platforms suddenly decided to completely rewire their social networks, users may be upset.

It is then natural to consider the questions in this paper: how does the network planner decide to allocate their budget? How much of an impact can be made? How large of a budget is needed to achieve a significant reduction in polarization?

We begin by first establishing a relationship between structural properties of a social network and its level of polarization. We find that \change{both the degree profiles and the strength of information bottlenecks -- quantified by the well-known Cheeger constant in spectral graph theory -- are closely tied to polarization}. This result naturally captures the intuition and dangers of echo chambers in real-world networks.

Next, we focus on the formulation and analysis of two settings for network optimization. In the first, the planner has full information of the population's opinions. We provide theoretical motivation for two heuristics: coordinate descent and a stepwise disagreement-seeking algorithm. The former is standard in optimization, while the latter is the antithesis of confirmation bias. Existing research has shown that addition of edges between like-minded individuals contributes to increasing polarization~\citep{Bhalla2021}. Moreover, according to~\citet{Bindel2015} it is `costly' for individuals to be connected to others who disagree with them, and recommender systems can be designed to minimize disagreement~\citep{Chitra2020}. Therefore, the incentives of both individuals and social media platforms may naturally lead to polarization growing over time. In contrast, we show that a simple disagreement-seeking approach taken by the planner leads to substantial reductions in polarization. \change{This result is closely tied to our choice of the opinion dynamics model. In this paper, interactions between individuals are always \emph{attractive} -- bringing opinions closer together. In reality, this is not the case (see \citet{bail2018exposure, balietti2021reducing}). This simplification, however, will facilitate theoretical results -- which we believe can be leveraged for partial understanding of polarization-reduction strategies in a richer class of models.}

This paper also presents a novel setting for the network planner, wherein the population's opinions are chosen adversarially. In several papers from the literature (see, for instance~\citet{Chen2022,Gaitonde2020,matakos2017measuring,rahaman2021model}), an adversary is able to change the individuals' opinions -- seeking to maximize polarization. The setting we study represents a planner whose network design must be \emph{robust} to the adversary's disruption. We show that this setting for the planner's problem is intimately related to maximizing the spectral gap of the graph's Laplacian, which is a well-studied problem~\citep{Donetti2006,Wang2008,Watanabe2010}. We provide theoretical guarantees for a heuristic that connects vertices on opposite sides of the cut corresponding to the spectral gap.

We then evaluate several natural heuristics on real-world and synthetic networks. There are significant reductions in polarization for networks with strong initial community structures. Furthermore, we study how the spectral gap and homophily are affected by the planner's modifications. We find that the largest reductions in polarization are accompanied by reductions in homophily. In many cases, however, one of our heuristics effectively reduces polarization with little effect on homophily. We also observe that two heuristics lead to vertices with extreme opinions becoming more central in the graph structure. In many of the networks studied, a small budget yields substantial reductions in polarization.

The paper is organized as follows. Section~\ref{sec:literature} provides a detailed review of recent and related work. Next, Section~\ref{sec:model} introduces relevant notation, definitions, and preliminaries. Section~\ref{sec:theoretical_results} provides theoretical ground for three heuristics, which are described and evaluated on several networks in Section~\ref{sec:emp_res}. Finally, Section~\ref{sec:conclusion} concludes and discusses potential directions for future work.

\section{Relevant Literature}
\label{sec:literature}

\change{
The papers most similar to our own are recent studies by~\citet{Chen2018a}, \citet{Gaitonde2020}, \cite{Chitra2020}, and~\citet{zhu2021minimizing}. \citet{Gaitonde2020} motivates our adversarial disruption of the population's opinions, while both \citet{Chen2018a} and \citet{zhu2021minimizing} aim to modify a social network's structure by adding a small number of edges. \citet{Chitra2020} impose a constraint on the edge weight modified -- but not the number of edges. In particular, they focus on changing a large number of edges by a small amount, whereas we seek to do the converse. Our work differs from \citet{Chen2018a} through greater emphasis on theory and generalization to weighted graphs. The objective function in~\citet{zhu2021minimizing} fundamentally differs from our own, and represents a different problem faced by the network planner. 

In addition, this paper is broadly tied to the literature on opinion dynamics, perturbation of network structures, and influencing polarization. Relevant studies in each of these areas are discussed in the following.
}

\subsubsection*{Opinion Dynamics}

The study of consensus-forming begins with the seminal work of~\citet{DeGroot1974}, where under weak conditions on the social network, the opinions eventually converge to a perfect consensus. This model was expanded by~\citet{Friedkin1990} \change{(and more recently by \citet{conjeaud2022degroot})}, so that the long-term opinions are heterogeneous. Because of this feature and its simplicity, the Friedkin-Johnsen (FJ) model has appeared in several recent studies on opinion polarization and disagreement -- see for instance,~\citet{matakos2017measuring,Musco2018,Chen2018a,Chitra2020,Gaitonde2020,zhu2021minimizing,Chen2022}. In this paper, we will also use the FJ model. Not only is it standard in the literature, but it is mathematically convenient for analysis. There are also rich areas of work which justify and extend the FJ model. For instance,~\citet{Bindel2015} show that the expressed opinions of this model correspond to the Nash equilibrium of a cost-minimizing game between individuals.

There are a few notable extensions to the FJ model, in which individuals have more complex behavior. For example, a recent survey by~\cite{biondi2022dynamics} presents several generalizations and (relevantly) assesses if polarization can occur in each. A fundamental feature of the FJ model is that individuals are always drawn toward the opinions of their neighbors -- but experimental evidence of this feature is inconclusive and contextual~\citep{bail2018exposure, balietti2021reducing}. Motivated by this observation, new models have been developed in which individuals have bounded confidence~\citep{hegselmann2002opinion} or even experience repulsion~\citep{rahaman2021model,cornacchia2020polarization}. \change{It is also possible to incorporate geometric structures into the dynamics, such as recent work by \citet{hazla2019geometric} and \citet{gaitonde2021polarization}. Finally, we note that there are several related studies within the controls literature, which focus on consensus dynamics on a network, for instance, when agents have antagonistic dynamics \citep{altafini2012consensus}, or are stubborn \citep{mao2018spread} -- see \citet{qin2016recent} for a more complete survey.}

\subsubsection*{Optimizing Network Structures}

This paper formulates an optimization problem over network structures, aiming to reduce a particular definition of polarization. There are several related works in the literature. For example, \citet{Musco2018} allows unconstrained rewiring of the social network to reduce the \textit{polarization-disagreement index}, which is defined as the sum of polarization and disagreement. A recent paper of~\citet{zhu2021minimizing} optimizes the same index via addition of a limited number of edges. This index is analytically and computationally convenient because of its monotonicity and convexity, but it is highly sensitive to the edge density of the graph.\footnote{The polarization-disagreement index consists of adding polarization, which is on the order of $n$ (the number of vertices), and disagreement, which is of order $m$ (the number of edges). Therefore this index is dominated by disagreement for dense graphs (specifically, if $m \gg n$).} We instead focus exclusively on minimizing polarization, which is shown to be neither convex nor monotone in Section~\ref{sec:given_ops}. However, this paper restricts edge modifications similarly to~\citet{zhu2021minimizing}.

A more closely related work by~\citet{Chen2018a} presents several definitions of `conflict' in social networks, and studies how they can be minimized through iterative perturbations to the graph. One such measure of conflict equals polarization. We expand on the authors' work by providing a detailed theoretical analysis of edge perturbations on polarization, \change{generalizing the analysis to weighted graphs,} and conducting simulations on larger real-world and synthetic networks.

The aforementioned papers share with ours a definition of polarization. However, it is possible to optimize for other notions of `cohesiveness' or `consensus'. For instance, \citet{garimella2017reducing} and ~\citet{haddadan2021repbublik} both present measures of polarization based on random walks, and propose algorithms for reducing it via edge addition. The greatest similarity between their work and ours lies in the use of a greedy, stepwise approach to a combinatorial optimization problem. However, the authors' definitions of polarization do not directly incorporate opinion dynamics.\footnote{We note that the Friedkin-Johnsen model has a random walk interpretation of the long-term opinions, see~\citet{Gionis2013}.} Moreover, in~\citet{haddadan2021repbublik}, nodes represent webpages, not individuals.

Another definition of cohesiveness, which does not depend on any node opinions or labels, is the spectral gap of a graph. The spectral gap controls the synchronizability of dynamical systems and mixing times of Markov chains~\citep{Donetti2006}, and therefore its maximization is of great interest. For instance,~\citet{Watanabe2010} seek to increase the spectral gap by removing nodes. Unlike these authors, we focus on changes to a graph's edges. More relevantly,~\citet{Wang2008} study how the algebraic connectivity (i.e., spectral gap) can be increased by adding edges. The authors present two strategies for doing so, one of which is derived from the eigenvector corresponding to the spectral gap. In this paper, we show that the adversarial setting of the planner's problem is closely related to their work, and provide bounds on polarization using this eigenvector-based strategy.

\subsubsection*{Natural Network Dynamics}

A different branch of research aims to understand how polarization is shaped by \emph{rewiring dynamics} in the network. For instance, a recent paper by~\citet{Bhalla2021} studies how individuals' local rewiring rules can lead to higher polarization. The authors conclude that confirmation bias and friend-of-friend behavior are critical for this result. However, their theoretical results focus on the polarization-disagreement index. Moreover, we derive an improved upper bound for polarization in Section~\ref{sec:contraction_and_pol}. A similar paper by~\citet{santos2021link} shows that allowing individuals to rewire according to structural similarity leads to polarization, although the authors use a distinct model of opinion dynamics.

It is also possible to study the dynamics driven by a network administrator. \citet{Chitra2020} present a setting in which a network administrator rewires the network over time by providing `recommendations' to users based on minimizing disagreement. They show that without a regularization term in the optimization problem, the administrator greatly increases polarization. The authors' result contrasts with one of the main findings of this paper, namely that connecting disagreeing individuals is effective for reducing polarization.

\subsubsection*{Optimizing Opinion Profiles}

While less relevant to this paper, a complimentary line of work assumes that the network structure remains fixed, but the innate opinions are subject to change. For instance, \citet{Gionis2013} establish NP-Hardness of an opinion maximization problem, in which an administrator takes over a small set of individuals and sets their opinions to the largest possible value. Papers by~\citet{matakos2017measuring} (resp.~\citet{matakos2020tell}) seek to minimize polarization (resp. maximize diversity, i.e., disagreement) by choosing a small subset of individuals to have neutral opinions. Finally, the work of~\citet{rahaman2021model} aims to minimize polarization in an extension of the FJ model, but by shifting each individuals opinion by a small amount.

These studies have generally taken the perspective of a benevolent network planner. It is also possible to consider the perspective of an adversary, who takes over a small number of individuals and seeks to maximize polarization or disagreement~\citep{Chen2022}. A more powerful adversary in~\citet{Gaitonde2020} chooses the opinions of the entire population to the same end. In particular,~\citet{Gaitonde2020} present a problem of \textit{defending} the network from this adversary by making some opinions more resistant to change. In this paper, we will consider a similar setting, but where the network is defended by altering its structure instead. Nonetheless, the adversary faced is modeled on their work.

\section{Model}
\label{sec:model}

\change{
An undirected graph $\cG(V, E, W)$ is defined by a set of vertices $V$ given by $[n] := \{ 1, \ldots, n \}$, a set of edges $E \subset V \times V$ consisting of unordered pairs of vertices, and weight matrix $W\in [0,\bar w]^{n\times n}$. $W$ is assumed to be a symmetric matrix of non-negative edge weights such that $w_{ij} > 0$ if and only if $(i,j) \in E$, and $\bar w < \infty$ indicates the maximum possible edge weight. For a graph $\cG$, its degree matrix $D$ is diagonal, and satisfies $D_{ii} = d_{i}$, where $d_{i} = \sum_{j} w_{ij}$ is the (weighted) degree of vertex $i$. Let $L = D-W$ denote the combinatorial graph Laplacian, and $\cL = D^{-1/2}LD^{-1/2}$ denote the normalized Laplacian. We write $N(i) := \{ j \in [n] : (i,j) \in E \}$ for the neighbors of vertex $i$.
}

Vertices are given \emph{innate opinions} $\mathbf{s} \in [0,1]^n$, which represent a continuum between two extreme positions on an issue. For instance, an individual who is totally in favor of strict firearm laws may have an opinion of $0$, whereas one extremely against any such regulations would have an opinion of $1$. The population's opinions evolve over time, beginning from the innate opinions $\mathbf{s}$. The evolution of opinions follows the dynamics of~\citet{Friedkin1990} (see below), and the opinions converge to a fixed point -- denoted $\mathbf{z}$ and called the \emph{expressed opinions} of the population. In this paper, we are interested in modifications to the underlying graph~$\cG$, and therefore take the innate opinions $\mathbf{s}$ to be fixed. Consequently, we write $\mathbf{z}$ and $\mathbf{z}'$ for the expressed opinions corresponding to the social networks $\cG$ and $\cG'$, respectively. Occasionally, to emphasize the underlying graph $\cG$, we will write $\mathbf{z}_{\cG}$.

\subsection{Opinion Dynamics}

In the seminal model of~\citet{DeGroot1974}, the population's expressed opinions converge to a perfect consensus under weak conditions. A notable extension of the DeGroot model is by~\citet{Friedkin1990}, whose model preserves long-term heterogeneity of opinions. In particular, $\mathbf{z} = c \vec{\mathbf{1}}$ if and only if $\mathbf{s} = c \vec{\mathbf{1}}$. This model is convenient for analysis because the expressed opinions can be written explicitly. 
Furthermore, several recent works in the literature have leveraged this opinion dynamics model -- see Section~\ref{sec:literature} for more detail.

The Friedkin-Johnsen (FJ) opinion dynamics model is specified by the discrete-time mapping $\mathbf{s}(t) \to \mathbf{s}(t+1)$ as follows. We initialize $\mathbf{s}(0) = \mathbf{s}$, and iterate

\change{
\begin{equation}
    \label{eq:op_dyn}
    s_{i}(t+1) = \frac{s_{i}(0) + \sum_{j \in N(i)}w_{ij}s_{j}(t)}{1+\sum_{j \in N(i)}w_{ij}},
\end{equation}
where $w_{ij}$ is the weight associated with edge $(i,j)$, and is non-zero if and only if $j\in N(i)$.} The expressed opinions $\mathbf{z}$ are the fixed point of this mapping, given by

\begin{equation}\label{eq:z}
    \mathbf{z} = (I+L)^{-1}\mathbf{s},
\end{equation}
\change{where $I$ denotes the $n\times n$ identity matrix.} Notice that $I+L \succcurlyeq I$ is necessarily invertible. Thus, there exist unique expressed opinions $\mathbf{z}$ for any given $\mathcal{G}$ and $s$. Moreover, since the eigenvalues of $(I+L)^{-1}$ are no greater than $1$, the expressed opinions of the FJ dynamics are a contraction of the innate opinions. \change{This observation also follows from the fact that the FJ model is purely \emph{attractive} -- opinions of connected individuals are always drawn to each other over time. One of the heuristics in this paper will depend on this feature of the dynamics. However, exposure to substantially differing opinions in the real-world may yield no effect, or even strengthen one's original position. In Section~\ref{sec:conclusion} we discuss how our results might be leveraged for such a class of richer opinion dynamics models, and relevant directions for future work.
}

\subsection{Polarization and Disagreement}
\label{sec:pol_and_dis}

In practice, a perfect consensus is rare; therefore, we seek to understand ``closeness'' to consensus. Accordingly, we define \emph{polarization} to be proportional to the variance of the expressed opinions. Large polarization indicates that the population is far from achieving a consensus, and vice-versa. Formally, we define:

\begin{definition}[Polarization]
    Given a vector of opinions $\mathbf{x} = \left( x_{1}, \ldots, x_{n} \right)$ and the mean of its entries $\overline{x} := \frac{1}{n} \sum_{i=1}^{n} x_{i}$, the \emph{polarization} of $\mathbf{x}$ is

    \begin{equation}
        P(\mathbf{x}) := \sum_{i=1}^n\left(x_i - \overline{x}\right)^2 = \norm{\widetilde{\mathbf{x}}}^2,
    \end{equation}
    where $\widetilde{\mathbf{x}} := \mathbf{x} - \overline{x} \vec{\mathbf{1}}$ are the mean-centered opinions.
\end{definition}

In particular, $P(\mathbf{z})$ is \emph{expressed} polarization, and $P(\mathbf{s})$ is \emph{innate} polarization.

It is also useful to define \emph{disagreement}, which captures distance from consensus on a local scale. Intuitively, if two vertices have very distinct opinions, then their disagreement is large.

\begin{definition}[Disagreement]
    For any vector of opinions $\mathbf{x} = \left( x_{1} \dots x_{n} \right)$, the disagreement between vertices $i$ and $j$ is given by:

    \begin{equation}
        D_{ij}(\mathbf{x}) := (x_i - x_j)^2.
    \end{equation}
\end{definition}

Again, between vertices $i$ and $j$, $D_{ij}(\mathbf{z})$ is the \emph{expressed} disagreement, while $D_{ij}(\mathbf{s})$ is the \emph{innate} disagreement. The two quantities above have been studied in several recent papers on social and information networks; see~\citet{matakos2017measuring, Musco2018, Chen2018a, Chitra2020, Gaitonde2020, zhu2021minimizing, Bhalla2021, rahaman2021model, santos2021link, Chen2022} and references therein.

\section{Theoretical Results}
\label{sec:theoretical_results}

We now present several theoretical results on polarization. We study how its magnitude depends on structural properties of the graph, and how it can vary as a planner modifies the edges.

\subsection{Opinion Contraction and Polarization}
\label{sec:contraction_and_pol}

This paper is primarily concerned with polarization of expressed opinions, $P(\mathbf{z})$. However, the relationship between expressed and innate polarization depends on $\cG$. Since the opinion dynamics model performs a contraction on the opinions, it follows that $P(\mathbf{z}) \le P(\mathbf{s})$. In fact, more is true: the contraction ratio is controlled by the degrees and structural properties of~$\cG$.

To present this result, we must introduce some notation. For any two disjoint subsets of vertices $B_1$ and $B_2$, let $E(B_1,B_2)$ denote the set of edges with one incident vertex in $B_1$ and the other in $B_2$. We define the conductance of a nonempty subset of vertices $X$ as

\change{
\begin{equation}
    h_{\cG}(X) := \frac{\sum_{(i,j) \in E(X,X^C)} w_{ij}}{\min\{ \sum_{v\in X} d_{v}, \sum_{u\in X^{C}} d_{u}\}}.
\end{equation}
}
The isoperimetric number \change{(also known as the Cheeger constant)} of a graph $\cG$ is given by

\begin{equation}
    \label{eq:isoperimetric}
    h_{\cG} := \min_{X\subset V, 0 < |X| < |V|} h_{\cG}(X),
\end{equation}
as in~\citet{chung1997spectral}, and will appear in the results. Note that $h_{\cG}\le 1$, since $h_{\cG}(X)= 1$ when $X$ consists of a single vertex. Furthermore, $h_{\cG} = 0$ if and only if $\cG$ is disconnected. The isoperimetric number of a graph is an indication of the presence of bottlenecks -- it is small when there exists a large set of vertices that is sparsely connected to the remainder of the graph.

We now arrive at a first result on the contraction properties of the FJ model on polarization.

\begin{prop}
	\label{prop:p_bounds}
	Let $d_{\min}$ and $d_{\max}$ be the minimum and maximum weighted degrees in $\cG$, and let $h_{\cG}$ be its isoperimetric number. 
    Then,
 
	\change{
    \begin{equation}
        \frac{P(\mathbf{s})}{\left(1+(2d_{\max})\wedge (\bar w n) \right)^{2}} \le P(\mathbf{z}) \le \frac{P(\mathbf{s})}{\left(1+\frac{1}{2}d_{\min}h_{\cG}^{2}\right)^{2}}.
    \end{equation}
	}
\end{prop}

Proposition~\ref{prop:p_bounds} quantifies the effects of the FJ model on polarization. In particular, if $\cG$ has strong expander properties (i.e., $h_{\cG}$ is large), then we expect the expressed polarization to be small, relative to the innate polarization. The proof of this result can be found in Appendix~\ref{sec:proofs}, and follows from simple eigenvalue bounds and a version of Cheeger's inequality.

This result provides a tighter upper bound on polarization than that of~\citet{Bhalla2021}. The tightening is achieved by observing that the mean-centered innate opinions $\widetilde{\mathbf{s}}$ are orthogonal to the eigenvector of $(I+L)^{-2}$ that has corresponding eigenvalue 1. In addition, we can use Proposition~\ref{prop:p_bounds} to show that the complete graph $K_n$, \change{with all edge weights equal to the maximal $\bar w$,} is a global minimum for polarization.

\begin{corollary}
    \label{cor:global_min}
    Fix innate opinions $\mathbf{s}$, and let $\cG$ be any graph on $n$ vertices with maximal edge weight $\bar w$. Let $\mathbf{z}_{K_n}$ and $\mathbf{z}_{\cG}$ denote the expressed opinions on $K_n$ and $\cG$, respectively. 
    Then, 
    
    \begin{equation}
        P(\mathbf{z}_{K_n}) \le P(\mathbf{z}_{\cG}).
    \end{equation}
    Moreover, \change{$P(\mathbf{z}_{K_n}) = \frac{P(\mathbf{s})}{(1+\bar w n)^2}$}.
\end{corollary}

\change{The key observation in the proof of Corollary~\ref{cor:global_min} is that all non-zero eigenvalues of $L_{K_n} = \bar w \left( n I - \vec{\mathbf{1}} \vec{\mathbf{1}}^T\right)$ (the Laplacian of the complete graph) are equal to $\bar w n$.} Therefore, for any~$\cG$, the value of polarization on $K_n$ achieves with equality the smallest lower bound from Proposition~\ref{prop:p_bounds}. This result also provides a useful reference point for indicating the planner's closeness to global optimality.

\subsection{Given Opinions}
\label{sec:given_ops}

We now turn to studying how the planner can decrease polarization by modifying the graph.

In a first setting, we assume that the innate opinions are known. If the planner can change (by adding or removing) the edge weight between at most $k$ pairs of vertices, what is the least polarization they can achieve? Formally, given a graph $\mathcal{G}$ with innate opinions $\mathbf{s}$, and integer budget $k > 0$, we wish to solve

\change{
\begin{equation}
	\label{opt:pol_min}
	\begin{split}
	\min_{\mathcal{G}'} \ & P(\mathbf{z}') \\
	 \mathrm{s.t.} \ & ||W-W'||_0 \le 2k,
	\end{split}
\end{equation}
where the expressed opinions $\mathbf{z}'$ correspond to $\cG'$, which must also be an undirected graph with maximal edge weight $\bar w$.} The factor of two in the constraint of~\eqref{opt:pol_min} follows from our assumption of undirected graphs. \change{The constraint naturally captures the assumption that it is costly for the planner to modify an edge, but upon committing to doing so, they may freely change the edge weight.}\footnote{\change{In principle, the constraint may bound the absolute difference in edge weights ($\ell_1$ norm). This is an entirely different problem, (more similar to \citet{Chitra2020}) but an interesting direction for future work. We believe that with an $\ell_1$ constraint, the planner would distribute its edge weight to maximize the minimal marginal return of polarization with respect to edge weight. We will also see that relaxing the $\ell_0$ constraint to $\ell_1$ is insufficient for obtaining a convex optimization problem.}}

\change{
Problem~\eqref{opt:pol_min} is challenging to solve efficiently since it is non-convex. Beyond the fact that the $\ell_0$ constraint gives a non-convex feasible set, the objective function (over valid Laplacian matrices) is also not convex -- see Figure~\ref{fig:obj_ncvx} for a small example. Therefore, relaxing the $\ell_0$ constraint to $\ell_1$ will still yield a non-convex optimization problem.} Instead of seeking an optimal set of $k$ edges to add, we propose a greedy stepwise approach where \change{the weight of $k$ edges are saturated} iteratively, one at a time. This simpler setting is tractable for analysis.

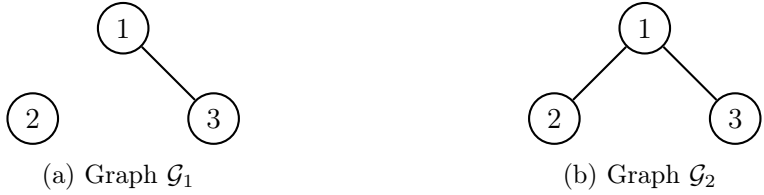
\begin{figure}
    \centering
    \begin{minipage}{0.6\linewidth}
        \begin{subfigure}{0.3\linewidth}
            \centering
            \begin{tikzpicture}[node distance=1cm]
              \node[circle, draw, thick] (1) {1};
              \node[circle, draw, thick] (2) [below left = of 1] {2};
              \node[circle, draw, thick] (3) [below right = of 1] {3};
        
              \path[thick]
                (1) edge (3);
            \end{tikzpicture}
            \caption{Graph $\cG_1$}
        \end{subfigure}
        \hfill
        \begin{subfigure}{0.3\linewidth}
            \centering
            \begin{tikzpicture}[node distance=1cm]
              \node[circle, draw, thick] (1) {1};
              \node[circle, draw, thick] (2) [below left = of 1] {2};
              \node[circle, draw, thick] (3) [below right = of 1] {3};
        
              \path[thick]
                (2) edge (1)
                (1) edge (3);
            \end{tikzpicture}
            \caption{Graph $\cG_2$}
        \end{subfigure}
    \end{minipage}
    \caption{\change{A simple example of the non-convex objective function. With innate opinions $\mathbf{s} = [0, 0.4, 1]$, it can be seen that $P\left( \frac{1}{2} \left[L_1 + L_2\right] \right) > \frac{1}{2} \left[ P(L_1) + P(L_2) \right]$. (Note the abuse of notation to illustrate $P(\cdot)$'s dependence on only the Laplacian.) In this particular example, the addition of any amount of weight to edge $(1,3)$ \textit{increases} polarization.}}
    \label{fig:obj_ncvx}
\end{figure}

It seems intuitive that adding \change{edge weight} to $\cG$ promotes the flow of information, and thereby reduces polarization. However, this is not the case in general. We will see that for most \change{non-saturated} edges, there exists a value of the innate opinions for which the addition of \change{weight to} that edge will increase polarization. The exact expression for the change in polarization when adding \change{edge weight} is given in the following.

\begin{lemma}
    \label{lem:p_diff_exact}
    Let $\cG(V,E)$ be an undirected graph yielding expressed opinions $\mathbf{z}$, \change{and $(i,j)$ be a pair of vertices with non-maximal weight, that is, $w_{ij} < \bar w$}. Let $\mathbf{v}_{ij} := \mathbf{e}_i - \mathbf{e}_j$. \change{For $\delta \in (0, \bar w - w_{ij}]$, we construct $\cG^+(V,E^+, W^+)$ according to $w_{ij}^+ = w_{ij}+\delta$, and $E^+ = \{(i,j): w_{ij}^+ > 0\}$.} If the expressed opinions on $\cG^+$ are given by $\mathbf{z}^+ := (I+L^+)^{-1}\mathbf{s}$, 
    then
    
    \change{
    \begin{equation}
        \label{eq:p_diff_exact}
        P(\mathbf{z}) - P(\mathbf{z}^+) 
        = D_{ij}(\mathbf{z}) \left[  \frac{2 \delta \widetilde{\mathbf{z}}^T(I+L)^{-1}\mathbf{v}_{ij}}{\widetilde{\mathbf{z}}^T\mathbf{v}_{ij}\left(1+\delta \mathbf{v}_{ij}^T(I+L)^{-1}\mathbf{v}_{ij}\right)}  - \frac{\delta^2 \mathbf{v}_{ij}^T(I+L)^{-2}\mathbf{v}_{ij}}{\left(1+\delta\mathbf{v}_{ij}^T(I+L)^{-1}\mathbf{v}_{ij}\right)^2}\right].
    \end{equation}
    }
\end{lemma}

The proof of this result can be found in Appendix~\ref{sec:proofs}. To discuss this result in more detail, it is useful to define the following.

\begin{definition}[$\partial_{w_{ij}} P(L)$]
    \label{def:dP_dw}
    Fix some innate opinions $\mathbf{s}$. Let $\mathbf{z}_{L}$ denote the resulting expressed opinions when the underlying graph $\cG$ has Laplacian $L$. We write:

    \begin{equation}
        \label{eq:pol_deriv_def}
        \partial_{w_{ij}} P(L) = \lim_{t\to0^+} \frac{P(\mathbf{z}_{L+tL_{ij}}) - P(\mathbf{z}_L)}{t}
    \end{equation}
    where $L_{ij} = \mathbf{v}_{ij}\mathbf{v}_{ij}^T$.
\end{definition}

This definition allows us to analyze the first-order effects of edge modifications on polarization. Notice that even if a graph were unweighted, we can define this derivative for its equivalent \emph{weighted} graph, where the weight of each existing edge equals one. In the following proposition, we derive a closed form expression for these partial derivatives.

\begin{prop}
    \label{prop:dP_dw}
    For fixed innate opinions $\mathbf{s}$, we have

    \begin{equation}
        \begin{split}
            \partial_{w_{ij}} P(L) &= -2\widetilde{\mathbf{s}}^T \left(I+L\right)^{-2}L_{ij}\left(I+L\right)^{-1}\widetilde{\mathbf{s}}\\
            &= -2\widetilde{\mathbf{z}}^T \left(I+L\right)^{-1}L_{ij}\widetilde{\mathbf{z}}.
        \end{split}
    \end{equation}
\end{prop}

This result allows us to re-write~\eqref{eq:p_diff_exact} in Lemma~\ref{lem:p_diff_exact} as:

\change{
\begin{equation}
    \label{eq:p_diff_exact_dP_dw}
    P(\mathbf{z}) - P(\mathbf{z}^+) = \frac{-\delta \partial_{w_{ij}} P(L)}{1+\delta \mathbf{v}_{ij}^T(I+L)^{-1}\mathbf{v}_{ij}} - \frac{\delta^2 \mathbf{v}_{ij}^T(I+L)^{-2}\mathbf{v}_{ij}}{\left(1+\delta \mathbf{v}_{ij}^T(I+L)^{-1}\mathbf{v}_{ij}\right)^2}(z_i - z_j)^2.
\end{equation}
}
Therefore, the necessary and sufficient condition for a reduction in polarization \change{due to adding weight $\delta$ to edge $(i,j)$} is:

\change{
\begin{equation}
    \label{eq:pol_grad_suff}
    -\partial_{w_{ij}} P(L) > (z_i - z_j)^2 \frac{\mathbf{v}_{ij}^T(I+L)^{-2}\mathbf{v}_{ij}}{\delta^{-1}+\mathbf{v}_{ij}^T(I+L)^{-1}\mathbf{v}_{ij}},
\end{equation}
}
which amounts to a steep enough first derivative.

Lemma~\ref{lem:p_diff_exact} also allows us to study when polarization \emph{increases} after adding weight to edge~$(i,j)$. In particular, if $\frac{\widetilde{\mathbf{z}}^T(I+L)^{-1}\mathbf{v}_{ij}}{\widetilde{\mathbf{z}}^T\mathbf{v}_{ij}} = 0$, then $P(\mathbf{z}^+) \ge P(\mathbf{z})$. Notice that if $\mathbf{v}_{ij}$ is not an eigenvector of $L$, then the addition of $(i,j)$ can increase polarization when the mean-centered innate opinions $\widetilde{\mathbf{s}}$ lie on the $(n-1)$-dimensional subspace orthogonal to $(I+L)^{-2}\mathbf{v}_{ij}$. \change{This condition is sufficient -- but not necessary -- the example in Figure~\ref{fig:obj_ncvx} illustrates this point.} Therefore, the planner cannot add edge weight arbitrarily and expect polarization to be reduced -- the innate opinions can determine the sign of the effect.

However, there are special cases in which polarization is always reduced, such as the following.

\begin{corollary}
	\label{cor:p_diff_special}
    If $\cG$, $i$, and $j$ satisfy $N(i) = N(j)$, 
    then polarization is always reduced by adding \change{weight $\delta$ to} the edge~$(i,j)$, and the difference is 
 
	\change{
    \begin{equation}
	   P(\mathbf{z}) - P(\mathbf{z}^+) = (z_i - z_j)^2 \frac{2\delta(1+\delta+d_i - w_{ij})}{(1+2\delta+d_i - w_{ij})^2}. 
	\end{equation}
    }
\end{corollary}

This result follows from proving that \change{$L\mathbf{v}_{ij} = (d_i - w_{ij}) \mathbf{v}_{ij}$} under the assumptions; see Appendix~\ref{sec:proofs} for full details.

Corollary~\ref{cor:p_diff_special} is somewhat counter-intuitive -- if we strengthen connections between individuals who share the same set of neighbors, we may expect to form an `echo chamber'. However, the opinion dynamics show that the addition of \change{weight to} such an edge $(i,j)$ will \textit{only} affect the expressed opinions of vertices $i$ and~$j$. While this edge fails to have any global effect, it does indeed bring the opinions of its incident vertices closer together -- hence reducing polarization. \change{The limitation of these effects to only its incident vertices suggests that in practice, the return on polarization may be small.}

Lemma~\ref{lem:p_diff_exact} is also used for arriving at the main result of this Section.

\begin{theorem}
    \label{thm:p_diff_bounds}
    Let $\mathbf{z},$ $\mathbf{z}^+,$ \change{$\delta,$} and $\mathbf{v}_{ij}$ be as before. 
    Then, 
    \change{
    $$P(\mathbf{z}) - P(\mathbf{z}^+) \le \frac{1+\lambda_n(L)}{1+2\delta+\lambda_n(L)} \left( - \delta \partial_{w_{ij}} P(L) \right).$$
    }
    
    Furthermore, if there exists $\epsilon > 0$ for which
    \change{
    $$\frac{\widetilde{\mathbf{z}}^T(I+L)^{-1}\mathbf{v}_{ij}}{\widetilde{\mathbf{z}}^T\mathbf{v}_{ij}} \ge \epsilon + \frac{\delta}{2\delta+(1+\lambda_{2}(L))^2},$$
    }
    then we also have
    \change{
    $$P(\mathbf{z}) - P(\mathbf{z}^+)  \ge \frac{2\delta \epsilon (z_i - z_j)^2}{1+2\delta}.$$
    }
\end{theorem}

Theorem~\ref{thm:p_diff_bounds} directly motivates two heuristics for the planner. First, we see that the largest possible reduction in polarization is proportional to the first order effect \change{$-\delta \partial_{w_{ij}} P(L)$}. Therefore, it is natural for the planner to iteratively add \change{maximal edge weight} along the direction of steepest descent -- a heuristic well-known as a \emph{coordinate descent}. Additionally, for fixed~$\epsilon$, the lower bound grows with the expressed disagreement. Therefore, edges with large $(z_i-z_j)^2$ are also good candidates for the planner to \change{add weight to}; we name this strategy \emph{disagreement-seeking}.

The upper bound in Theorem~\ref{thm:p_diff_bounds} implies that there is a diminishing return in adding more weight to a single edge, as \change{$\frac{{1+\lambda_n(L)}}{{1+2\delta+\lambda_n(L)}} < 1$.} In particular, this shows that although $P(\mathbf{z})$ is not globally convex, it is indeed convex along the direction of $w_{ij}$.

\subsection{Adversarial Opinions}
\label{sec:adv_ops}

In some cases, the planner may not reliably use the innate or expressed opinions. For instance, they may be difficult (even impossible) to measure, or vertices may be susceptible to takeovers; see~\citet{Gionis2013, matakos2017measuring, matakos2020tell, Gaitonde2020, rahaman2021model, Chen2022} for examples of the latter. \change{Moreover, individuals' opinions may be multidimensional -- capturing many distinct issues (e.g., firearm regulation, universal basic income, healthcare, etc.), all of which are shaped by the network's structure.} Such cases may require the planner to take a robust approach: they seek to design a network structure that minimizes polarization for \emph{any} possible vector of innate opinions.\footnote{There is one other possible justification for this formulation -- a robust (or minimax) optimization problem arises when the decision-maker is ambiguity averse, as is shown axiomatically by~\citet{Gilboa1989}.} \change{Formally, they aim to solve the following:}

\change{
\begin{equation}
	\label{opt:robust_pol_min}
	\begin{split}
	\min_{\cG'} \max_{\mathbf{s} \in \mathbb{R}^n: \norm{\mathbf{s}}_2^2 \le R} \ & \mathbf{\widetilde{s}}^T\left( I+L'\right)^{-2}\mathbf{\widetilde{s}} \\
	\mathrm{s.t.} \ & ||W-W'||_0 \le 2k.
	\end{split}
\end{equation}
Polarization in the resulting graph $\cG'$ will be robust to the choice of innate opinions}, and this optimization problem yields different graph structures than problem~\eqref{opt:pol_min}. As before, the factor of two in the constraint captures all graphs being undirected.

\change{
This optimization problem can be interpreted as a game -- an adversary selects $\mathbf{s}$ from the $n$-dimensional sphere of radius $R$, and the planner evaluates polarization on this choice of $\mathbf{s}$. A similar problem appears in \citet{Gaitonde2020}, who studies a `network defender' that decreases vertices' susceptibility to the adversary. In contrast, we consider defending the network through modification of its structure. However, we note that both our defender and theirs face the same adversary. This choice allows us to directly compare the effectiveness of these two defensive strategies. Although such an adversary may not be realistic, we believe this setting has numerous other justifications.
}

Note that the innate opinions now lie in the $n$-dimensional sphere, as opposed to the hypercube. This formulation allows us to relate the adversary's problem to spectral properties of the resultant graph~$\cG'$. In fact, the planner's problem~\eqref{opt:robust_pol_min} is equivalent to maximization of $\lambda_2$, the spectral gap of the Laplacian.

\begin{prop}
    \label{prop:robust_pol_min}
    The optimal solution $\cG'$ to~\eqref{opt:robust_pol_min} is the same as that of
	
    \change{
    \begin{equation}
	\label{opt:spectral_gap_max}
		\begin{split}
		\max_{\cG'} \ & \lambda_{2}(L') \\
		\mathrm{s.t.} \ & ||W-W'||_0 \le 2k,
		\end{split}
	\end{equation}
    }
	If the optimal solution to~\eqref{opt:spectral_gap_max} is $L^*$, 
    then the optimal value of~\eqref{opt:robust_pol_min} is $\frac{R}{\left(1+\lambda_2(L^*)\right)^2}$.
\end{prop}

For two graphs $\cG$ and $\cG'$, if \change{$W \le W'$ elementwise}, then $L \preccurlyeq L'$, and therefore $\lambda_2(L) \le \lambda_2(L')$. Therefore, the planner must only \change{add edge weight to $\cG$, as reducing weights} cannot increase the spectral gap.\footnote{We remark that this monotonicity of the spectral gap in the edge set does not hold for the normalized Laplacian~$\cL$, see for instance~\citet{Eldan2017}.} The spectral gap of the Laplacian is intimately tied to the synchronizability of various types of dynamical systems and the mixing time of Markov chains \citep{Donetti2006} and hence several studies seek to maximize it~\citep{Watanabe2010,Wang2008}. In this adversarial setting, where perfect synchronization is impossible, the spectral gap controls the best achievable consensus.

The proof of Proposition~\ref{prop:robust_pol_min} follows from solving the inner maximization problem, for which the optimal solution is the eigenvector of $L'$ corresponding to the second-smallest eigenvalue. This eigenvector is called the \emph{Fiedler} vector of $\cG'$, and describes a partition of vertices that approximates the normalized sparsest cut of~$\cG$~\citep{chung1997spectral}.

For a graph with Laplacian $L$, Proposition~\ref{prop:robust_pol_min} indicates that the \emph{worst-case} polarization is equal to $P(L) = \frac{R}{(1+\lambda_2(L))^2}$. The adversary achieves this by choosing $\widetilde{\mathbf{s}}$ along the span of the Fiedler vector. The planner's effectiveness in problem~\eqref{opt:robust_pol_min} is controlled by $P(L) - P(L')$, the difference in \emph{worst-case} polarization.

As in the previous setting, we approach this problem by iteratively choosing \change{edges to saturate -- starting from} the initial graph until no further budget remains. Therefore, the principal results address how \change{increasing an edge's weight} affects the spectral gap and thereby polarization. This is quantified in Theorem~\ref{thm:spectral_gap_bds}, which relates changes in the spectral gap to elementwise differences in the Fiedler vector.

\begin{theorem}
	\label{thm:spectral_gap_bds}
    \change{
    Let $\cG$ be an undirected graph, and $(i,j)$ be an edge with non-maximal weight, that is, $w_{ij} < \bar w$. Let also $\mathbf{v}$ be the Fiedler vector of $\cG$ of unit magnitude with corresponding eigenvalue $\lambda_{2}(L)$. Recall that $\lambda_{3}(L)$ is the third smallest eigenvalue of $L$, and define $\beta = \lambda_{3}(L) - \lambda_{2}(L)$. 
 
    For some $\delta \in (0,\bar w - w_{ij}]$, let $\cG^+$ be constructed by adding weight $\delta$ to edge $(i,j)$. If $\alpha = |v_i - v_j|$, 
    then we have that
    
	\begin{equation}
	    \label{eq:spectral_gap_bds}
	    \max \left\{1 - \frac{2\delta }{\beta}, 0 \right\}\delta \alpha^2 \le \lambda_{2}(L^+) - \lambda_{2}(L) \le \delta \alpha^2.
	\end{equation}
	}
\end{theorem}

The proof follows from adapting the result of~\citet{Maas1987}. The bounds are tightest when $\beta$ is largest, equivalently when $\lambda_2(L)$ is the sole small eigenvalue of $L$.

This result motivates a simple heuristic for maximizing the spectral gap, which appears in~\citet{Wang2008}. The planner can iteratively compute the Fiedler vector and \change{add weight to non-saturated edges whose incident vertices have large absolute difference in $\mathbf{v}$.}

Corollary~\ref{cor:spectral_gap_red_bds} quantifies the effects on polarization induced by the perturbation in Theorem~\ref{thm:spectral_gap_bds}.

\begin{corollary}
    \label{cor:spectral_gap_red_bds}
    \change{
    Let $P(L) = \frac{R}{\left(1+\lambda_2(L)\right)^2}$ be the worst-case polarization on a graph with Laplacian~$L$. 
    In the setting of Theorem~\ref{thm:spectral_gap_bds}, we have
    
    \begin{equation}
        \label{eq:adv_pol_bounds}
        \frac{2R\delta}{\left(1+2\delta+\lambda_2(L)\right)^3}\max\left\{1-\frac{2\delta}{\beta},0\right\}\alpha^2 \le P(L) - P(L^+) \le \frac{4R \left(\delta \vee \delta^2\right)}{\left(1+\lambda_2(L)\right)^3}\alpha^2
    \end{equation}
    }
\end{corollary}

In contrast to the setting with full information, the worst-case polarization $P(L)$ cannot increase when the planner increases an edge's weight. Recall that this follows from the monotonicity of the spectral gap in \change{$W$}. However, it is possible that \change{the resulting graph $\cG^+$ has greater} polarization for \emph{some} particular innate opinions. The settings in~\eqref{opt:pol_min} and~\eqref{opt:robust_pol_min} are distinct, and therefore the quantities compared before and after edge-weight addition are fundamentally different.

\section{Empirical Simulations}
\label{sec:emp_res}

If we solved problems~\eqref{opt:pol_min} or~\eqref{opt:spectral_gap_max} naively, it would be necessary to test all $\sum_{i = 1}^k \binom{\binom{n}{2}}{i}$ possibilities. Given that computing polarization (or the spectral gap) requires $O(n^3)$ time, we obtain a crude upper bound of $O(k n^{2k+3})$. Note that for fixed $k$, this rate is polynomial in $n$ -- albeit still not scalable. However, in subsequent experiments we choose $k$ to grow linearly with $n$, which results in superexponential runtime. It is therefore extremely impractical to compute the optimal solution, and we resort to theoretically motivated heuristics.

In Sections~\ref{sec:given_ops} and~\ref{sec:adv_ops}, we briefly discussed three heuristics for solving the planner's problem in a greedy, iterative fashion. Our theoretical results studied how polarization is reduced by \textit{addition} of \change{weight to} a single edge. Therefore, all of the following heuristics are based on \change{increasing edge weights}. These are presented below -- detailing the edge to \change{be saturated (i.e. setting edge weight to $\bar w$)} at every step and briefly discussing the time complexity of each iteration. We will compare these approaches with a random baseline.

\begin{itemize}
    \item \textbf{Random:} Add an edge from $E^C$ chosen uniformly at random; this has runtime of $O(\log(n))$.
    
    \item \textbf{Disagreement Seeking (DS):} \change{$\underset{(i,j) \in E^C}{\mathrm{argmax}} \ (\bar w - w_{ij})(z_i - z_j)^2$}.
    
    Computing the expressed opinions requires $O(n^3)$ time, and it takes $O(|E^C|)$ time to check all candidate nonedges.
    \item \textbf{Coordinate Descent (CD):} \change{$\underset{(i,j) \in E^C}{\mathrm{argmax}} \   -(\bar w - w_{ij}) \partial_{w_{ij}} P(\mathbf{z})$}
    
    Requires $O(n^3)$ runtime for computing a matrix inverse and multiplication, and $O(|E^C|)$ to find the optimal edge.\footnote{Naively, one might think we need $O(n^3 |E^C|)$ time to find the optimum, as we perform a matrix multiplication to compute the gradient of every candidate edge. However, the matrix multiplication is extremely sparse, and can be reduced to operating on four entries of a fixed, pre-computed matrix.}
    
    \item \textbf{Fiedler Difference (FD):} \change{$\underset{(i,j) \in E^C}{\mathrm{argmax}} \ (\bar w - w_{ij})|v_i - v_j|$, where $\lambda_2\mathbf{v} = L\mathbf{v}$}
    
    Takes $O(n^3)$ time to compute the eigendecomposition of $L$, and $O(|E^C|)$ to find the argmax.
\end{itemize}

\change{Notice two effects at play: the maximal weight that can be added ($\bar w - w_{ij}$), and some measure of effectiveness per unit weight (disagreement, partial derivative, or absolute difference in Fiedler vector). Naturally, each heuristic attempts to maximize the two's product.}

In addition, note that the three non-random heuristics have total runtime of $O(k(n^3+|E^C|))$. The random baseline has runtime of only $O(k \log(n))$, but computing polarization at each step (for purposes of comparison) comes with an additional cost of $O(kn^3)$. \change{We believe that the random heuristic is useful for two reasons. First, it captures a totally naive recommendation system, which does not curate a user's content exposure based on their opinions. Second, in two of the random graph models we study -- Erd\H{o}s-R\'enyi and stochastic block -- the result of the random heuristic is another graph from the same model, but with slightly higher edge density. Therefore, this heuristic allows us to study how much \textit{additional} polarization is reduced by adding edges in an informed, targeted, manner.}

We now study the performance of these heuristics on \change{six unweighted graphs.} First we look at three real-world networks -- sourced from Twitter, Reddit, and political blogs -- and then three synthetic networks with different characteristics: the Erd\H{o}s-R\'enyi, stochastic block, and preferential attachment models. Table~\ref{tab:graph_params} provides basic information about the graphs studied. In what follows, the planner's budget is given by $k = \lfloor \frac{n}{2} \rfloor $, such that on average each vertex receives one new edge. We plot the value of polarization with the planner's budget, along with the reference point $P(\mathbf{z}_{K_n})$, which represents the global minimum of polarization.

\begin{table}[t]
    \centering
    \begin{tabular}{c|c c}
        \toprule
        Network & Vertices $n$ & Edges $m$ \\
        \midrule
        Twitter & 548 & 3638  \\
        Reddit & 556 & 8969  \\
        Blogs & 1222 & 16717 \\
        Erd\H{o}s-R\'enyi & 1000 & 9990  \\
        SBM & 1000 & 13726 \\
        PA & 1000 & 4883  \\
        \bottomrule
    \end{tabular}
    \caption{Initial networks for evaluation of polarization-reducing heuristics.}
    \label{tab:graph_params}
\end{table}

\begin{table*}[t]
    \centering
    \begin{tabular}{c c | c c c | c c c}
        \toprule
        \multirow{2}{*}{\textbf{Quantity}} 
        & \multirow{2}{*}{\textbf{Heuristic}} 
        & \multicolumn{3}{c}{\textbf{Real-World Networks}} 
        & \multicolumn{3}{c}{\textbf{Synthetic Networks}}
        \\
        & & Twitter & Reddit & Blogs & Erd\H{o}s-R\'enyi & SBM & PA \\
        \midrule
        \multirow{5}{*}{\stackanchor{Expressed}{Polarization}} 
        & Initial   & 0.1664     & 0.0053      & 36.6      & 0.2422      & 3.526      & 1.71      \\
        & Random    & 0.1011     & 0.0035      & 22.1      & 0.2185      & 2.579      & 1.35      \\
        & DS        & 0.0221     & \bf{0.0006} & \bf{8.2}  & 0.1432      & 1.774      & 0.62      \\
        & CD        & \bf{0.0200}& \bf{0.0006} & \bf{8.2}  & \bf{0.1420} & \bf{1.767} & \bf{0.61} \\
        & FD        & 0.0754     & 0.0013      & 15.1      & 0.2011      & 1.814      & 1.23      \\
        \midrule
        \multirow{5}{*}{\stackanchor{Spectral}{Gap}} 
        & Initial   & 0.439      & 0.960      & 0.169      & 7.381       & 4.58      & 2.85         \\
        & Random    & 0.689      & 0.977      & 0.300      & 8.183       & 5.51      & 3.19         \\
        & DS        & 0.792      & 0.972      & 1.393      & 7.439       & 6.67      & 3.10         \\
        & CD        & 0.799      & 2.820      & 1.258      & 7.435       & 6.79      & 3.21         \\
        & FD        & \bf{2.052} & \bf{9.170} & \bf{2.327} & \bf{12.046} & \bf{6.87} & \bf{4.02}    \\
        \midrule
        \multirow{5}{*}{\stackanchor{Assortativity of}{Innate Opinions}} 
        & Initial   & 0.023     & -0.007     & 0.811      & -0.016     & 0.687     & 0.025  \\
        & Random    & 0.018     & -0.005     & 0.779      & -0.015     & 0.661     & 0.029  \\
        & DS        & -0.143    & -0.142     & 0.747      & -0.114     & 0.606     & -0.138 \\
        & CD        & -0.090    & -0.093     & 0.747      & -0.102     & 0.618     & -0.114 \\
        & FD        & 0.029     & -0.007     & 0.780      & -0.013     & 0.635     & 0.026  \\
        \bottomrule
    \end{tabular}
    \caption{Values for expressed polarization, spectral gap, and innate assortativity computed before and after the planner applies each heuristic to six networks. The planner adds $k=\lfloor \frac{n}{2} \rfloor$ edges -- an average of one new edge per vertex. The best-performing heuristics are highlighted in bold. Appendix~\ref{app:figures} contains additional figures showing changes in the spectral gap and assortativity with the planner's budget.}
    \label{tab:comp_results}
\end{table*}

Table~\ref{tab:comp_results} shows three quantities: expressed polarization, spectral gap, and assortativity of innate opinions. Expressed polarization is the principal concern of this study, and through Proposition~\ref{prop:robust_pol_min} is closely related to the spectral gap. Assortativity is introduced by~\citet{newman2003mixing}, and captures the degree of homophily in a network \change{-- which has been shown to control the speed of consensus-forming \citep{golub2012homophily}.} In particular, assortativity lies in $[-1,1]$, and measures the correlation of an attribute across edges. In this paper, this metric is evaluated for the innate opinions.

In general, the random baseline decreases polarization the least, and both the DS and CD heuristics outperform the Fiedler vector-based strategy. This is expected, as the FD heuristic is blind to the innate opinions, and uses strictly less information. However, we observe that DS and CD tend to result in negative values of homophily, while the FD heuristic does not share this tendency. As an interesting implication, it does not appear that a reduction in polarization requires negative values of homophily. Namely, it may not be necessary to directly connect the most polarized individuals in a society to reduce its level of polarization.

In the figures that follow, vertices are colored according to their innate opinions. Graphs are plotted using the python module \verb!networkx!~\citep{Hagberg2008}. Vertices are placed in two-dimensional space using force-directed algorithms, in which vertices repel each other and edges behave like springs in tension. Therefore, the vertex layout reflects their relative attraction. The same random seed for initial node placement is used for every graph type studied. All code and data used in this paper is publicly available \href{https://github.com/drigobon/minimizing-polarization}{{\color{blue} here}}.

\begin{figure}[ht]
    \centering
    \begin{subfigure}{0.53\linewidth}
        \centering
	    \includegraphics[width=\linewidth]{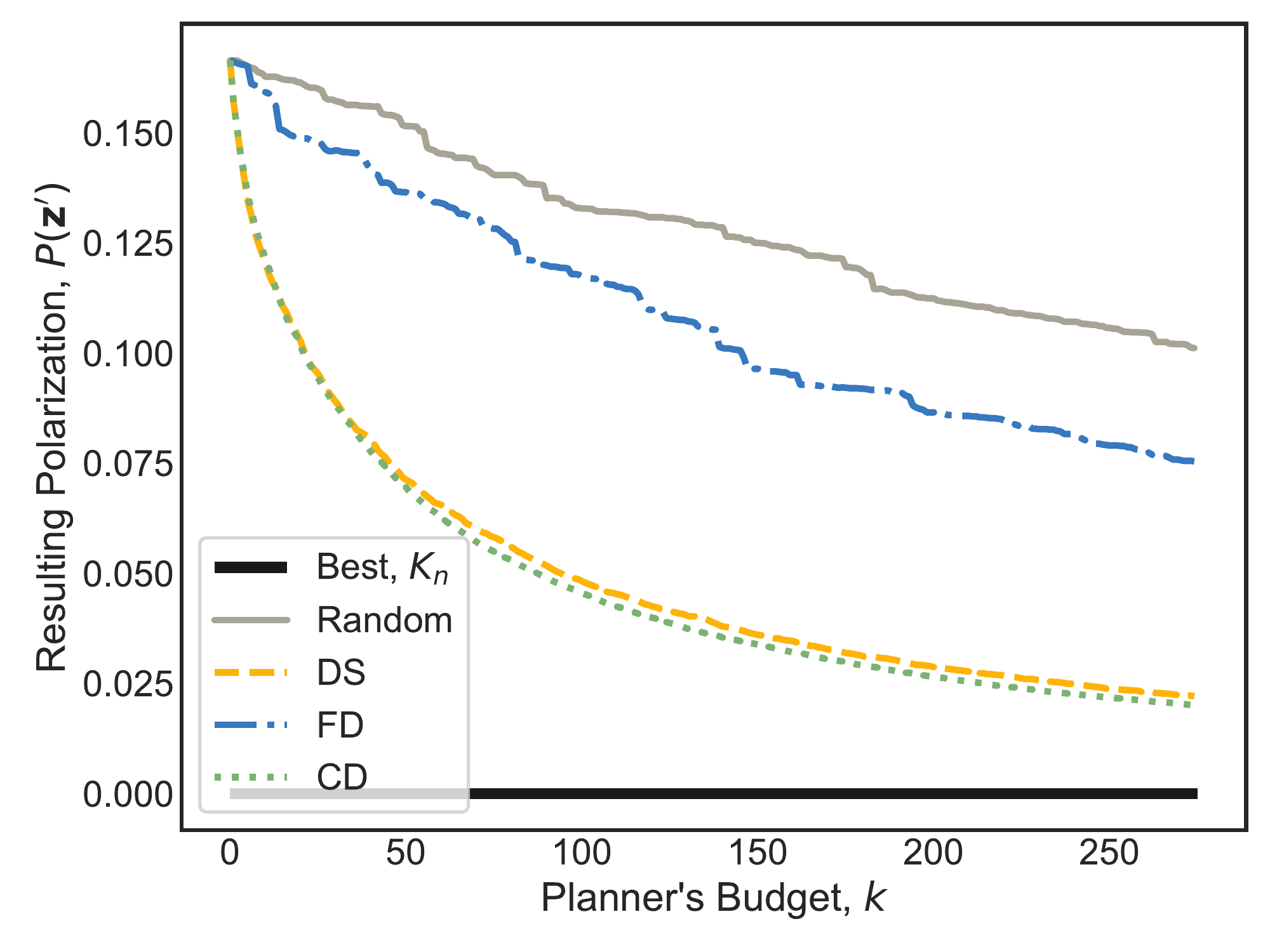}
        \caption{Reduction of Polarization}
        \label{fig:pol_tw}
    \end{subfigure}
    \begin{subfigure}{0.05\linewidth}
    	\raisebox{0.2in}{\includegraphics[width=\linewidth]{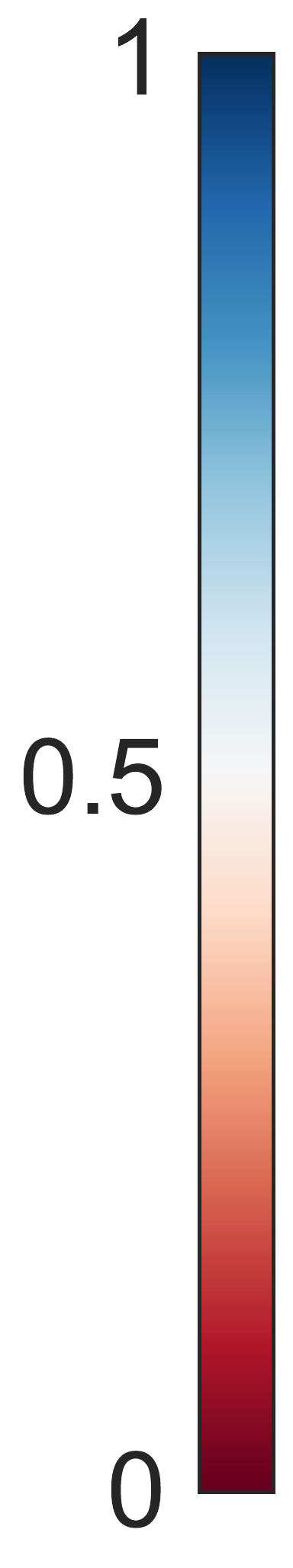}}
    \end{subfigure}
    \begin{subfigure}{0.4\linewidth}
    	\includegraphics[trim=0 10 0 0, clip, width=\linewidth]{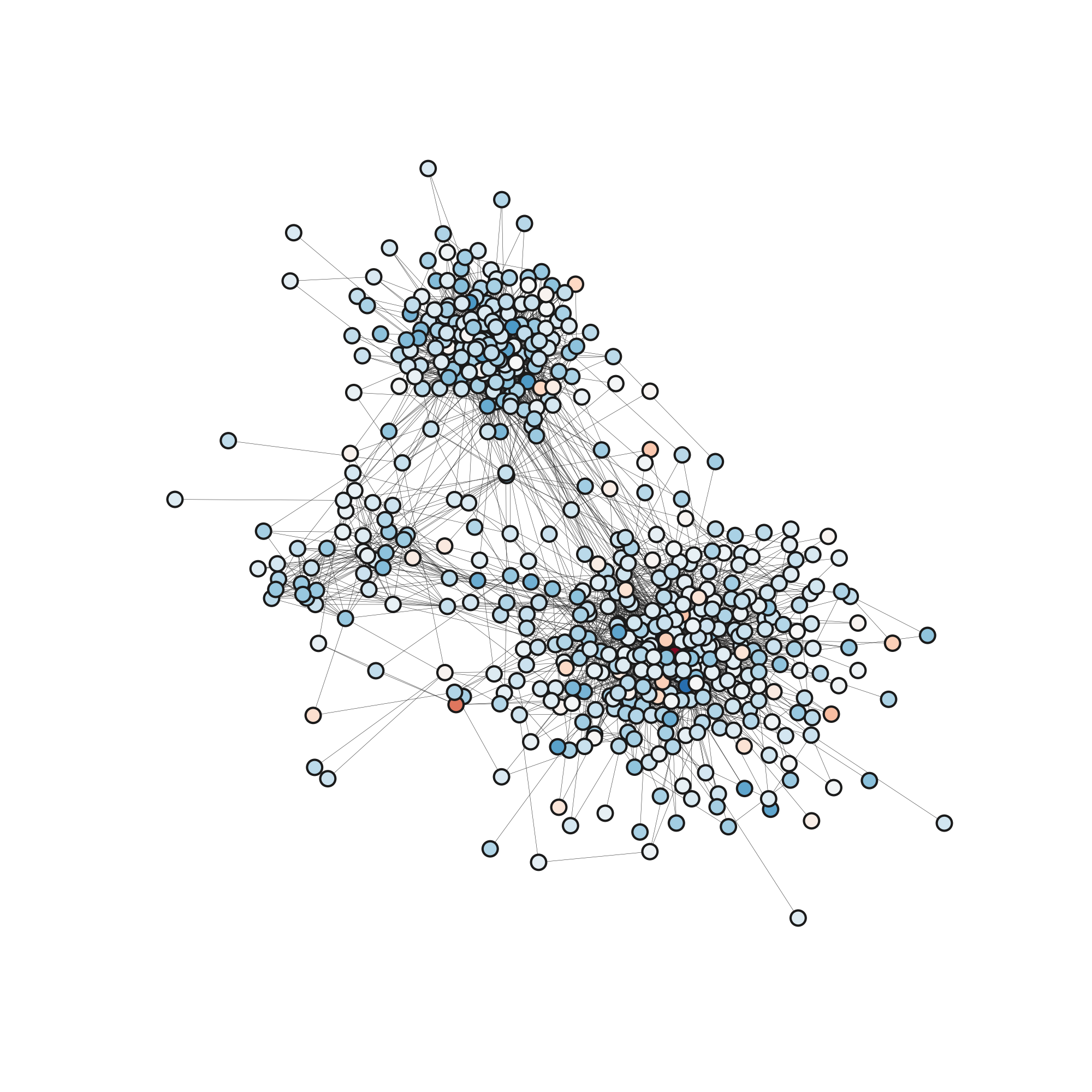}
    	\caption{Initial Graph Structure} \label{fig:tw_pre}	
    \end{subfigure}
    \\
    \begin{subfigure}{0.03\linewidth}
    	\raisebox{0.2in}{\includegraphics[width=\linewidth]{colormap}}
    \end{subfigure}
    \begin{subfigure}{0.23\linewidth}
        \centering
    	\includegraphics[width=\linewidth]{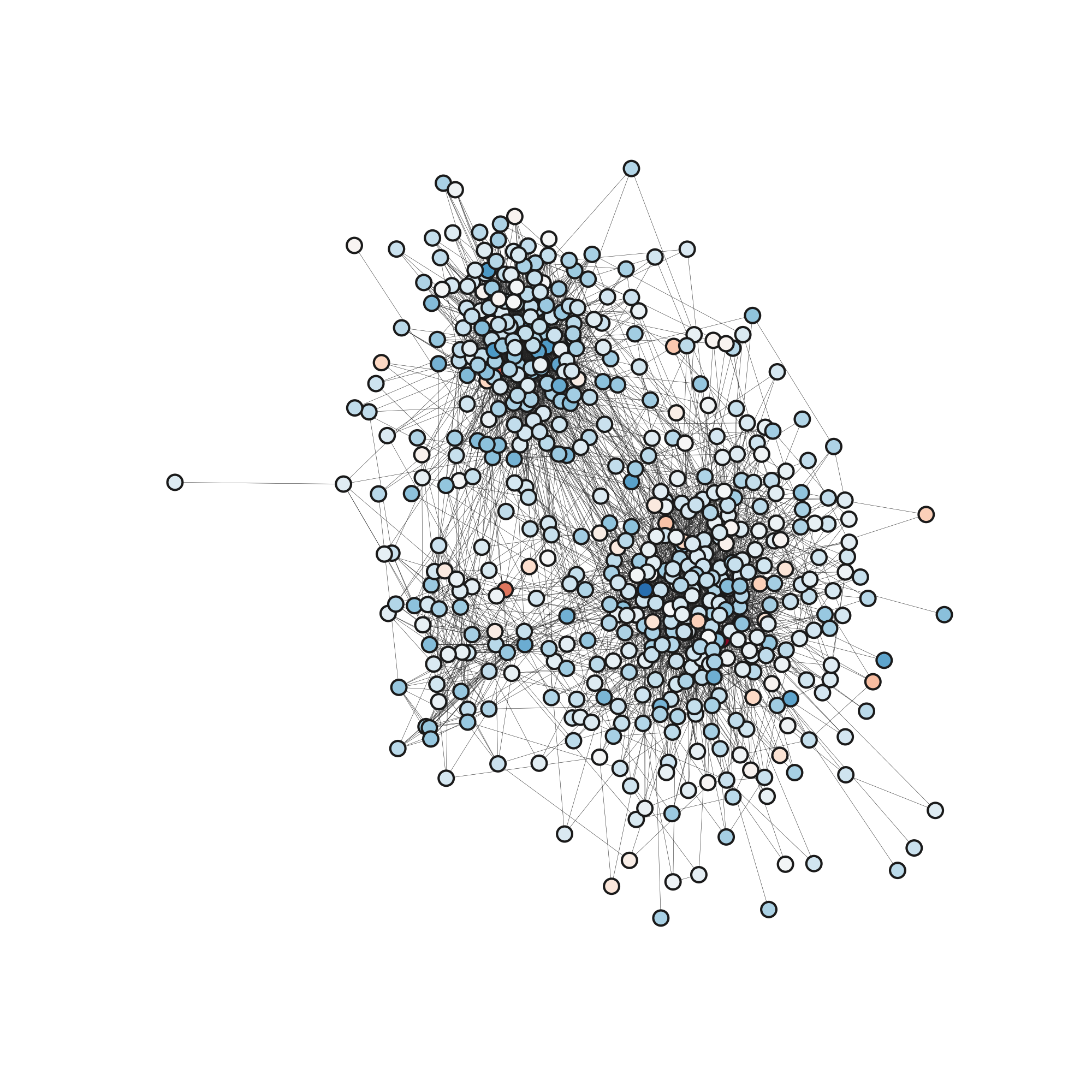}
    	\caption{Random} \label{fig:tw_post_random_add}	
    \end{subfigure}
    \begin{subfigure}{0.23\linewidth}
    	\includegraphics[width=\linewidth]{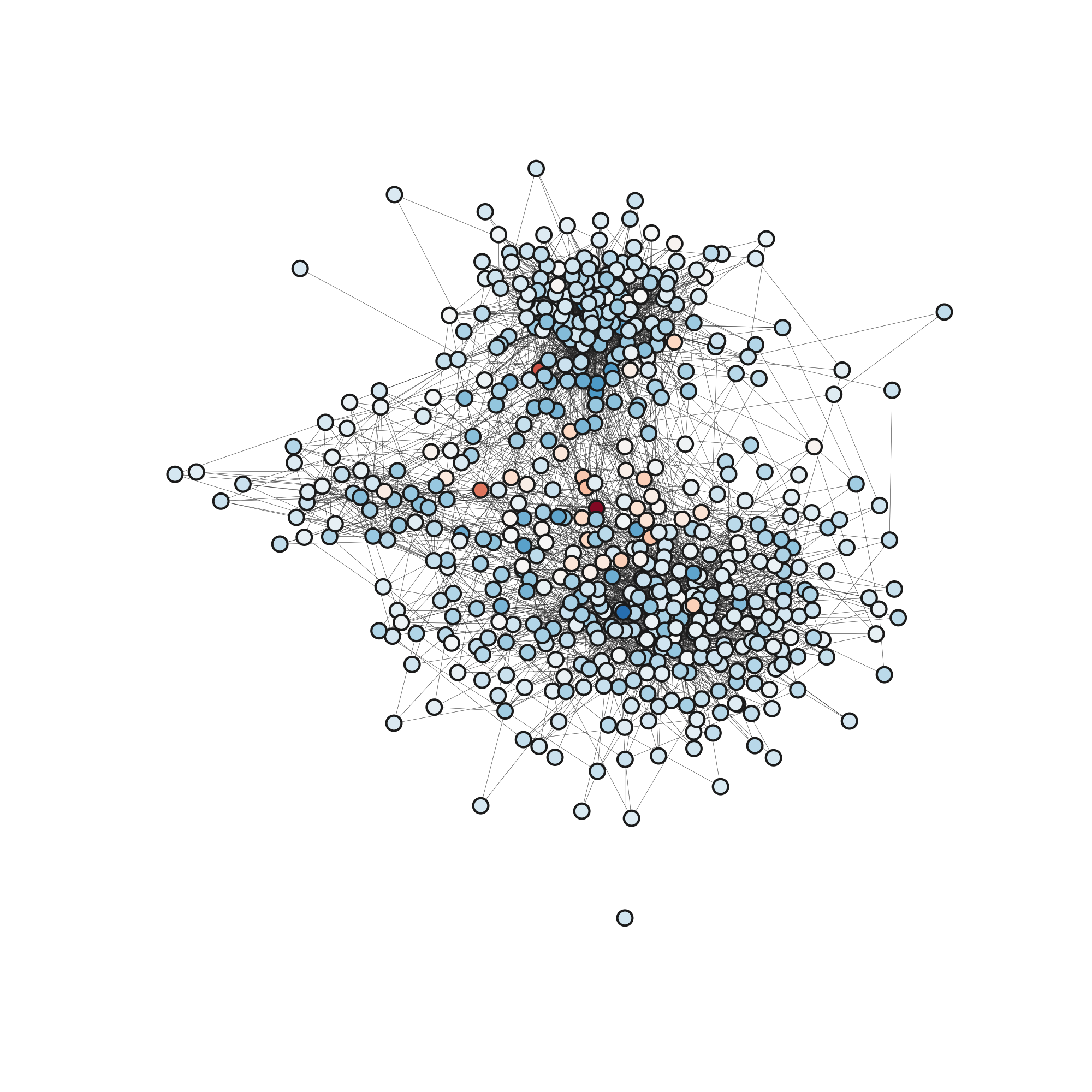}
    	\caption{DS} \label{fig:tw_post_max_dis}	
    \end{subfigure} 
    \begin{subfigure}{0.23\linewidth}
    	\includegraphics[width=\linewidth]{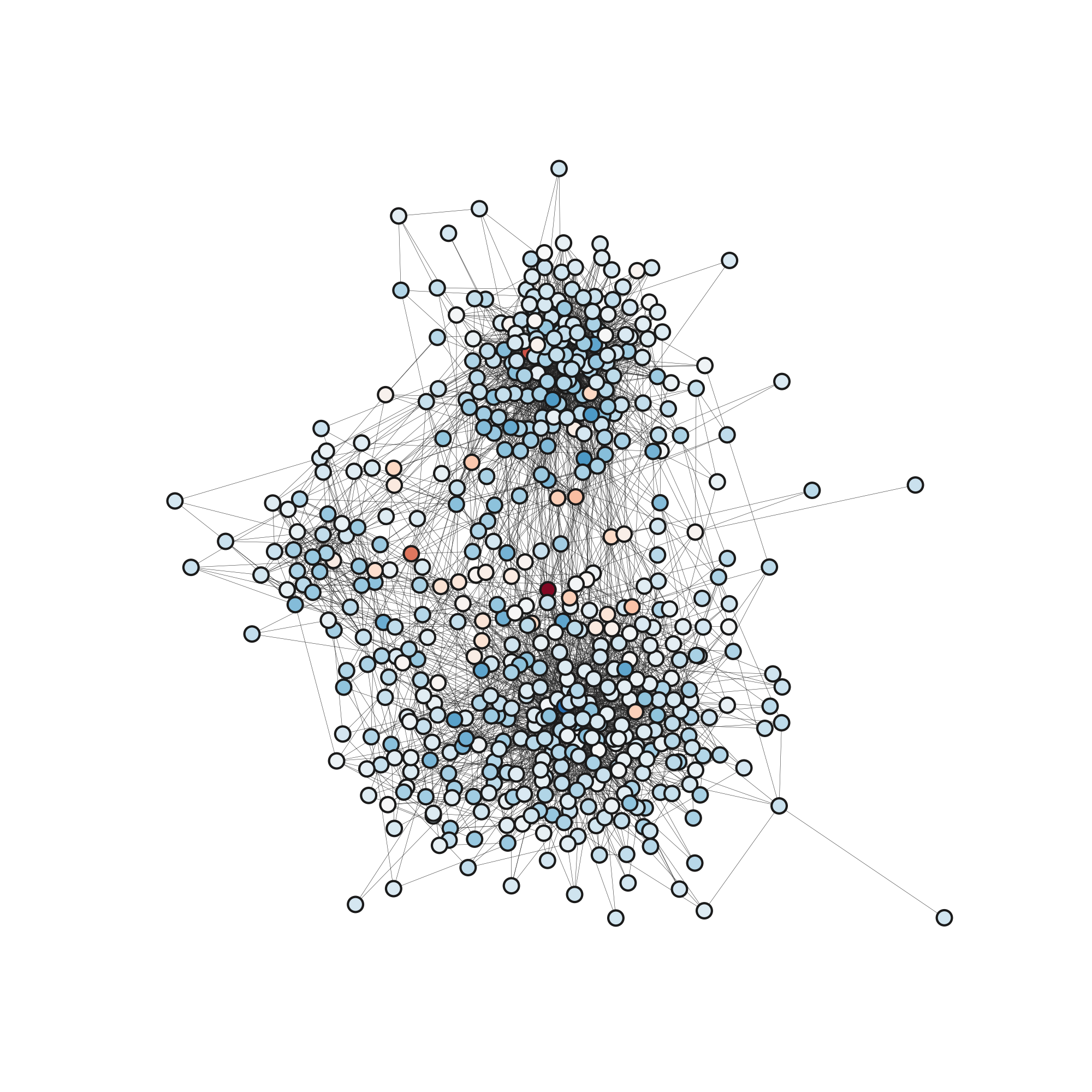}
    	\caption{CD} \label{fig:tw_post_max_grad}	
    \end{subfigure}   
    \begin{subfigure}{0.23\linewidth}
    	\includegraphics[width=\linewidth]{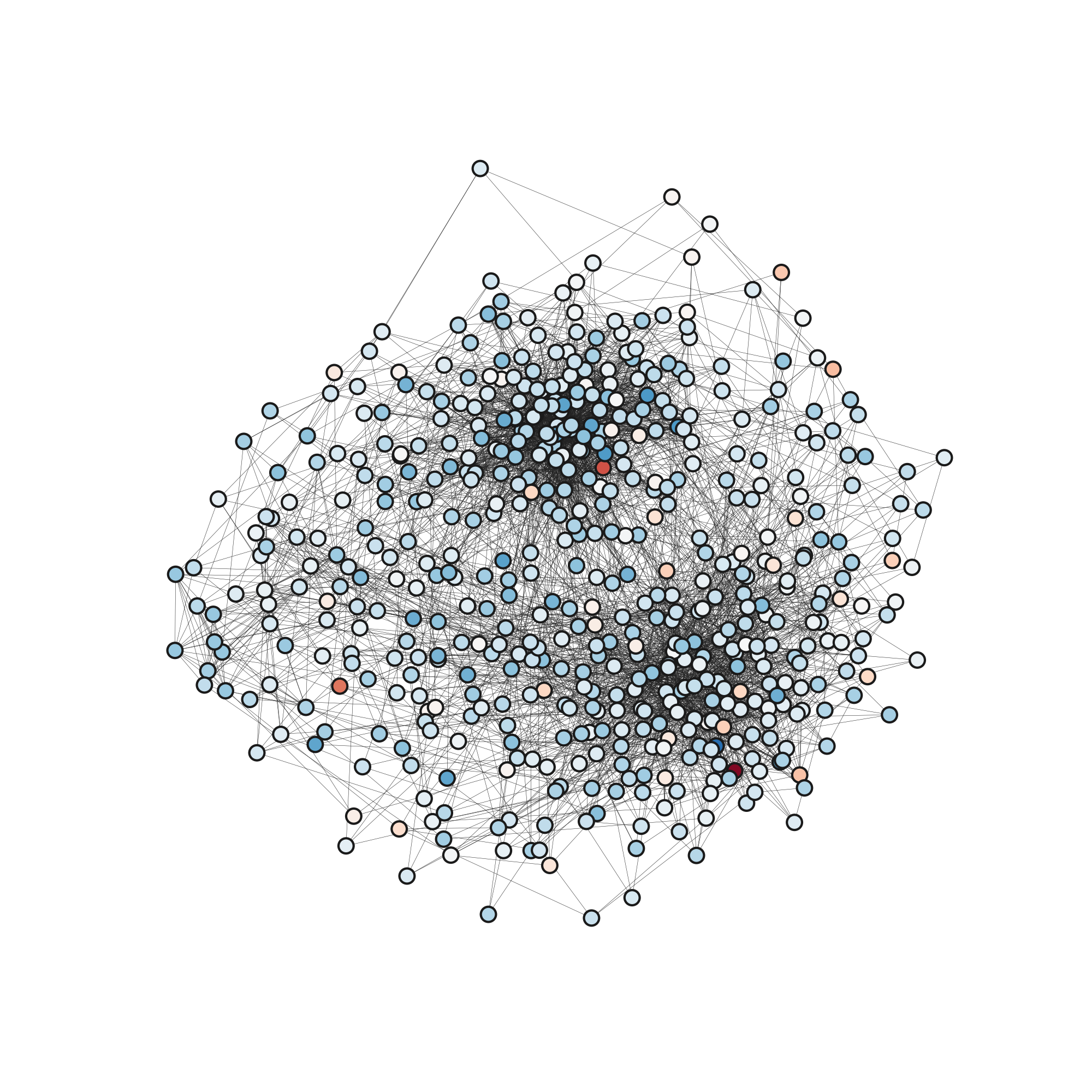}
    	\caption{FD} \label{fig:tw_post_max_fiedler_diff}	
    \end{subfigure} 
    \caption{Evaluation of the planner's heuristics on the Twitter network. Panel (a) shows the reduction achieved as the planner gradually adds edges. Panel (b) shows the initial network, while (c)-(f) visualize the network after the planner has exhausted their budget according to each heuristic. Vertices are colored according to their innate opinions.}
\end{figure}

\begin{figure}[ht]
    \centering
    \begin{subfigure}{0.53\linewidth}
        \centering
	    \includegraphics[width=\linewidth]{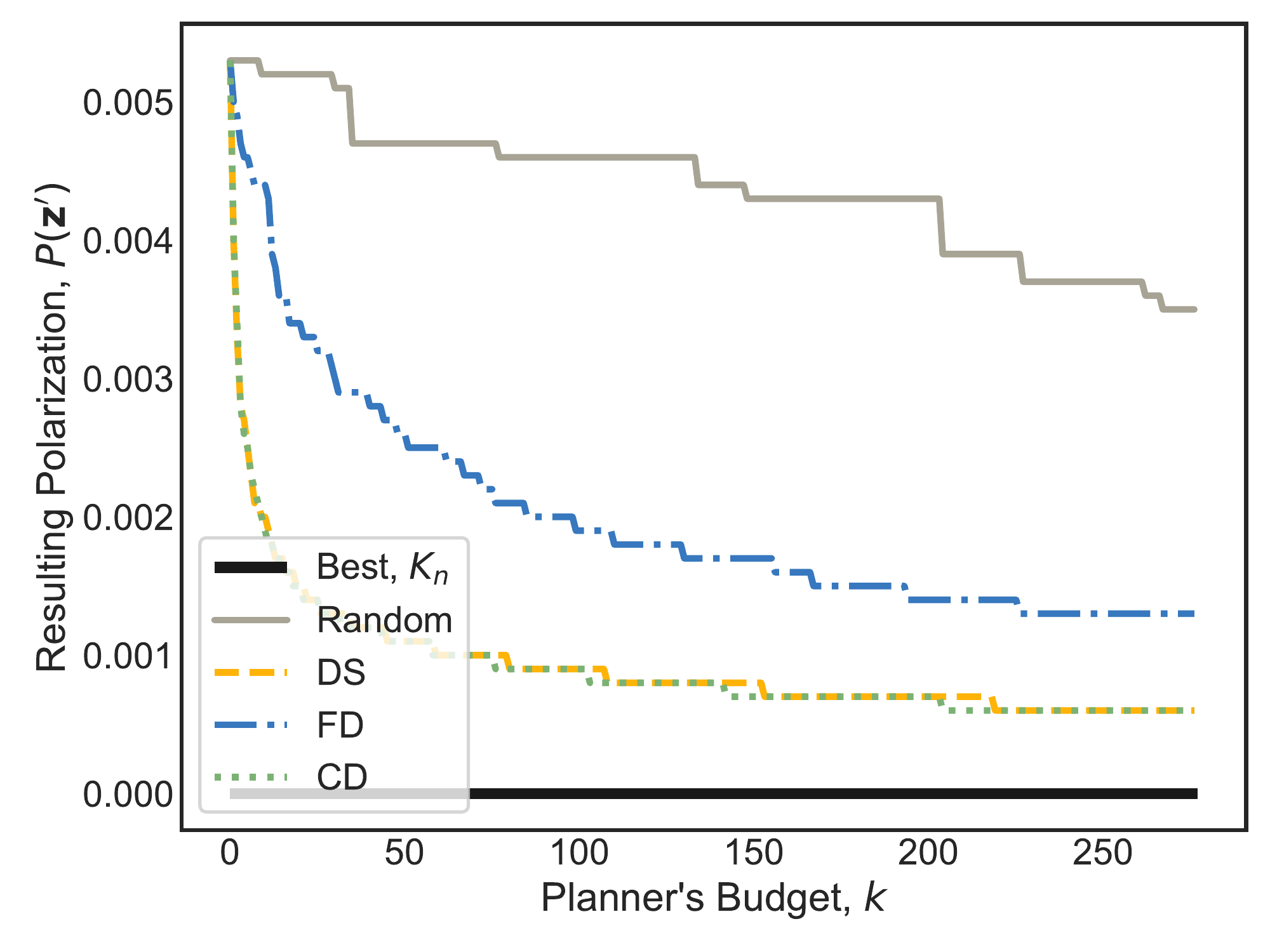}
        \caption{Reduction of Polarization}
        \label{fig:pol_rd}
    \end{subfigure}
    \begin{subfigure}{0.05\linewidth}
    	\raisebox{0.2in}{\includegraphics[width=\linewidth]{colormap}}
    \end{subfigure}
    \begin{subfigure}{0.4\linewidth}
    	\includegraphics[trim=0 10 0 0, clip, width=\linewidth]{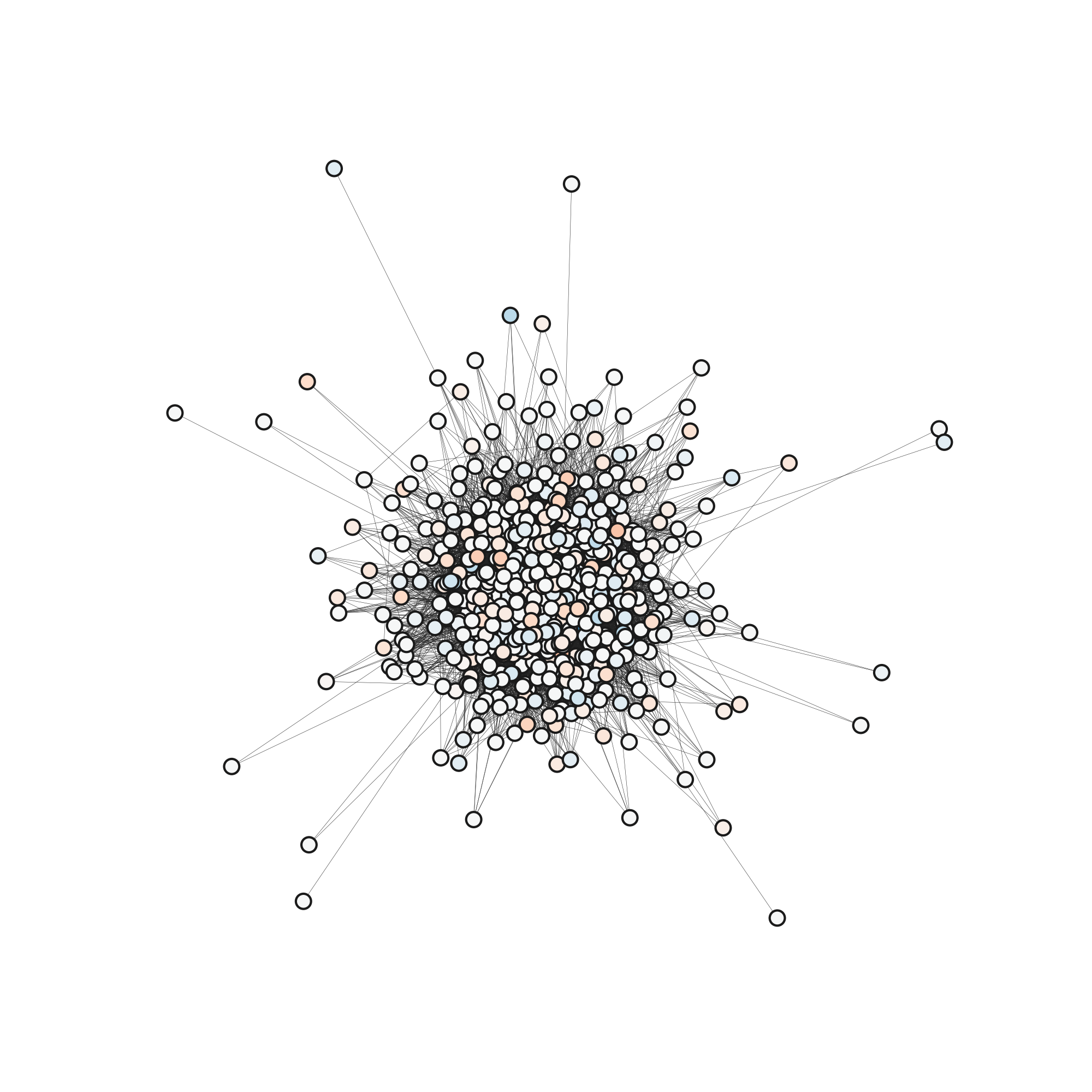}
    	\caption{Initial Graph Structure} \label{fig:rd_pre}	
    \end{subfigure}
    \\
    \begin{subfigure}{0.03\linewidth}
    	\raisebox{0.2in}{\includegraphics[width=\linewidth]{colormap}}
    \end{subfigure}
    \begin{subfigure}{0.23\linewidth}
        \centering
    	\includegraphics[width=\linewidth]{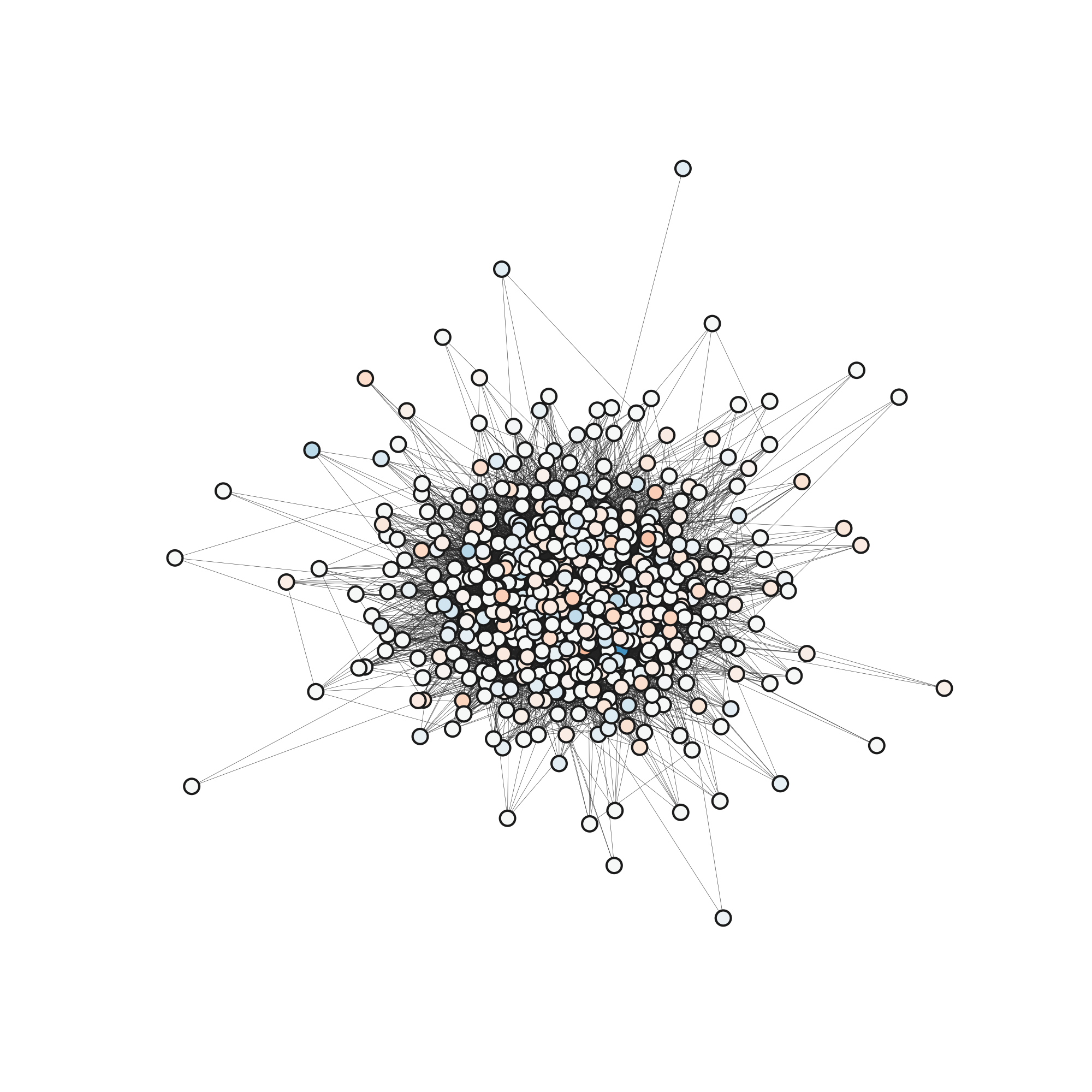}
    	\caption{Random} \label{fig:rd_post_random_add}	
    \end{subfigure}
    \begin{subfigure}{0.23\linewidth}
    	\includegraphics[width=\linewidth]{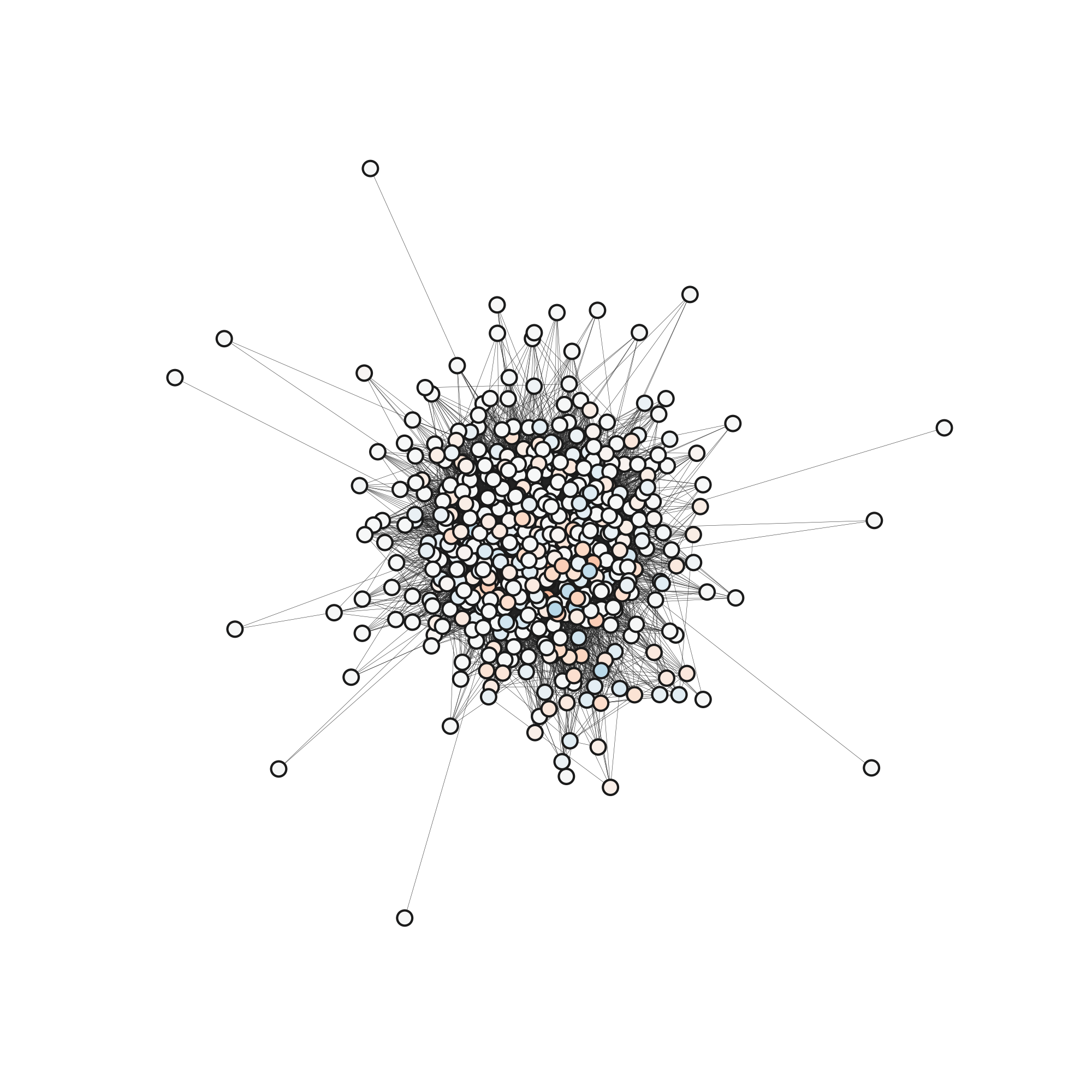}
    	\caption{DS} \label{fig:rd_post_max_dis}	
    \end{subfigure} 
    \begin{subfigure}{0.23\linewidth}
    	\includegraphics[width=\linewidth]{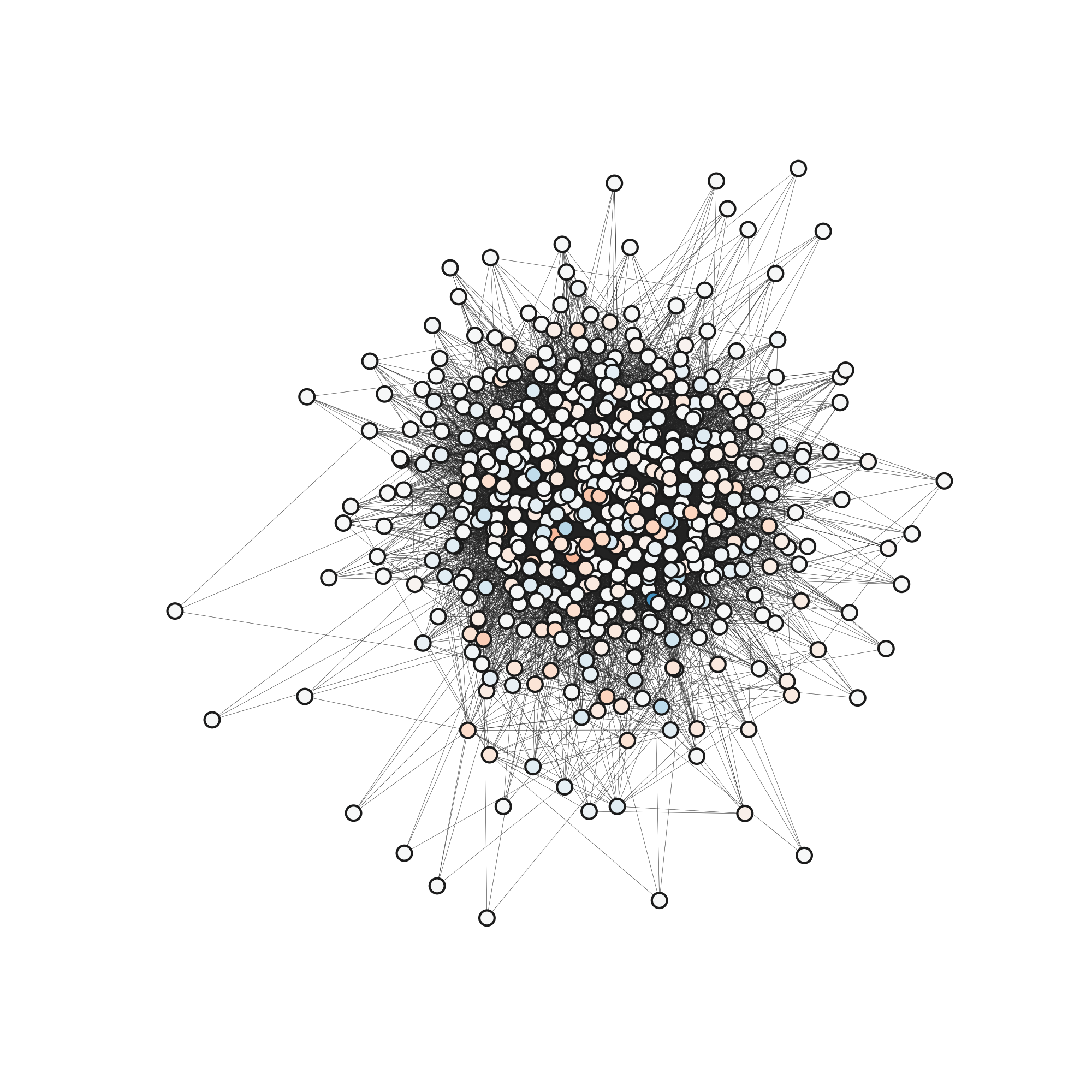}
    	\caption{CD} \label{fig:rd_post_max_grad}	
    \end{subfigure}   
    \begin{subfigure}{0.23\linewidth}
    	\includegraphics[width=\linewidth]{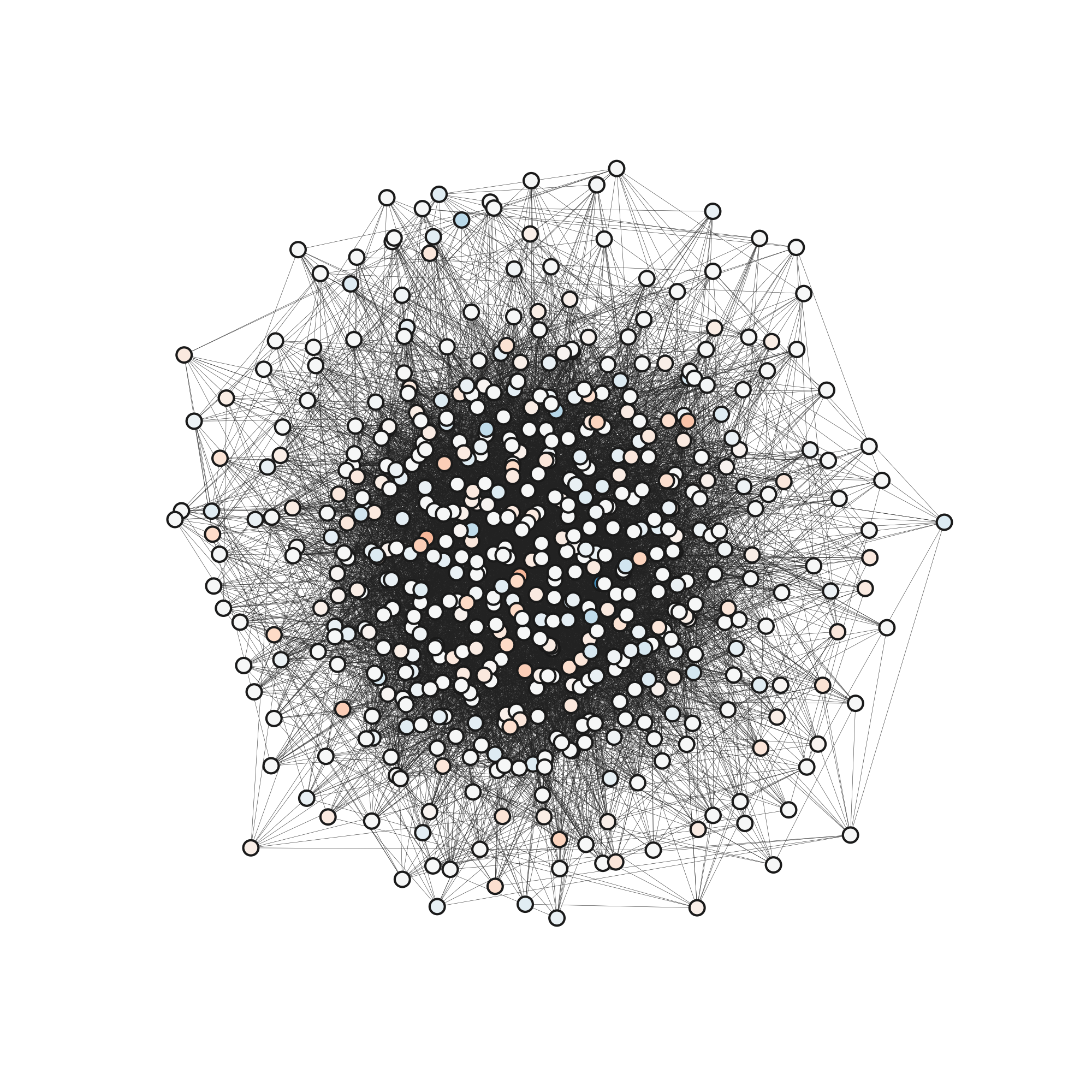}
    	\caption{FD} \label{fig:rd_post_max_fiedler_diff}	
    \end{subfigure} 
    \caption{Evaluation of the planner's heuristics on the Reddit network. Panel (a) shows the reduction achieved as the planner gradually adds edges. Panel (b) shows the initial network, while (c)-(f) visualize the network after the planner has exhausted their budget according to each heuristic. Vertices are colored according to their innate opinions.}
\end{figure}

\begin{figure}[ht]
    \centering
    \begin{subfigure}{0.53\linewidth}
        \centering
	    \includegraphics[width=\linewidth]{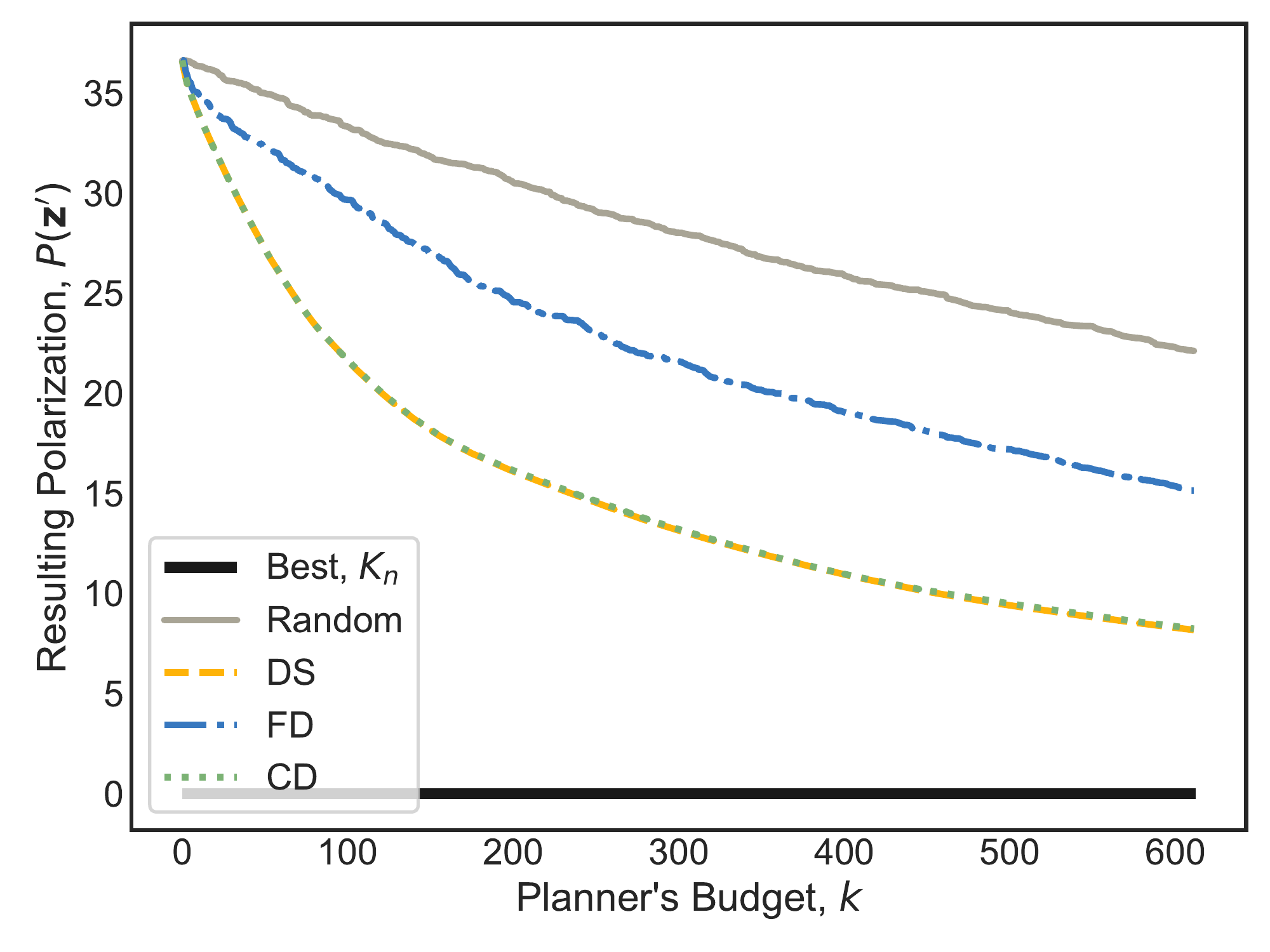}
        \caption{Reduction of Polarization}
        \label{fig:pol_bg}
    \end{subfigure}
    \begin{subfigure}{0.05\linewidth}
    	\raisebox{0.2in}{\includegraphics[width=\linewidth]{colormap}}
    \end{subfigure}
    \begin{subfigure}{0.4\linewidth}
    	\includegraphics[trim=0 10 0 0, clip, width=\linewidth]{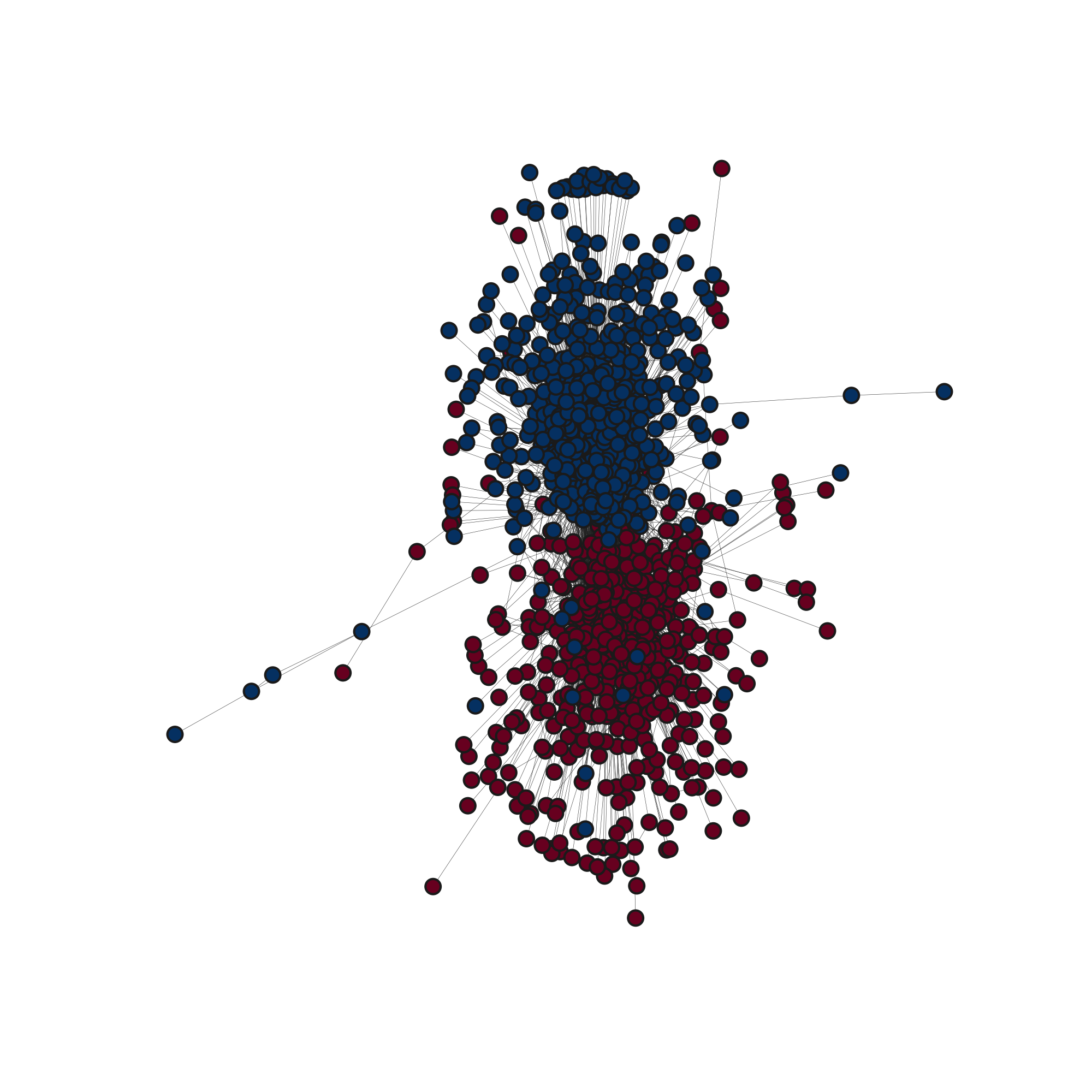}
    	\caption{Initial Graph Structure} \label{fig:bg_pre}	
    \end{subfigure}
    \\
    \begin{subfigure}{0.03\linewidth}
    	\raisebox{0.2in}{\includegraphics[width=\linewidth]{colormap}}
    \end{subfigure}
    \begin{subfigure}{0.23\linewidth}
        \centering
    	\includegraphics[width=\linewidth]{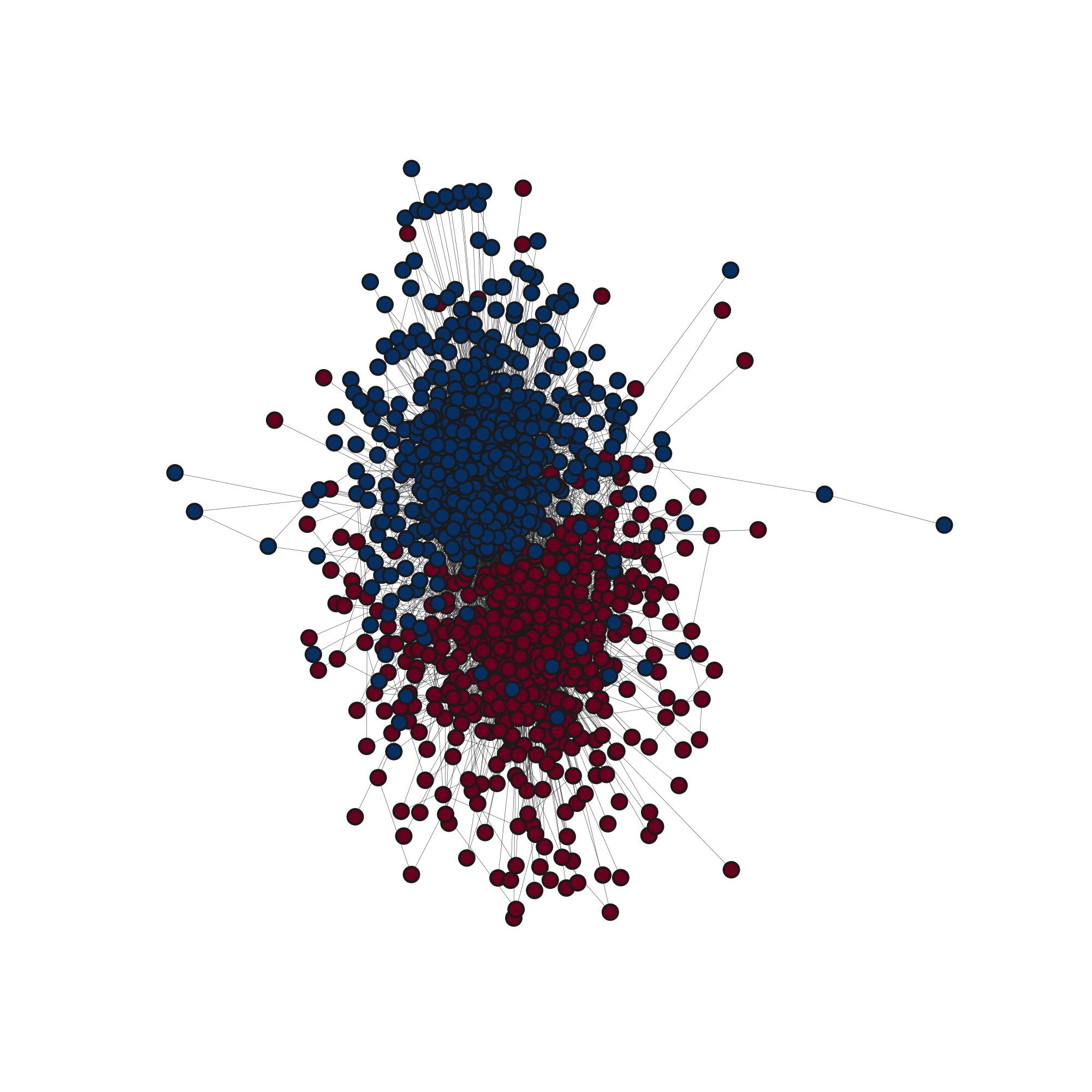}
    	\caption{Random} \label{fig:bg_post_random_add}	
    \end{subfigure}
    \begin{subfigure}{0.23\linewidth}
    	\includegraphics[width=\linewidth]{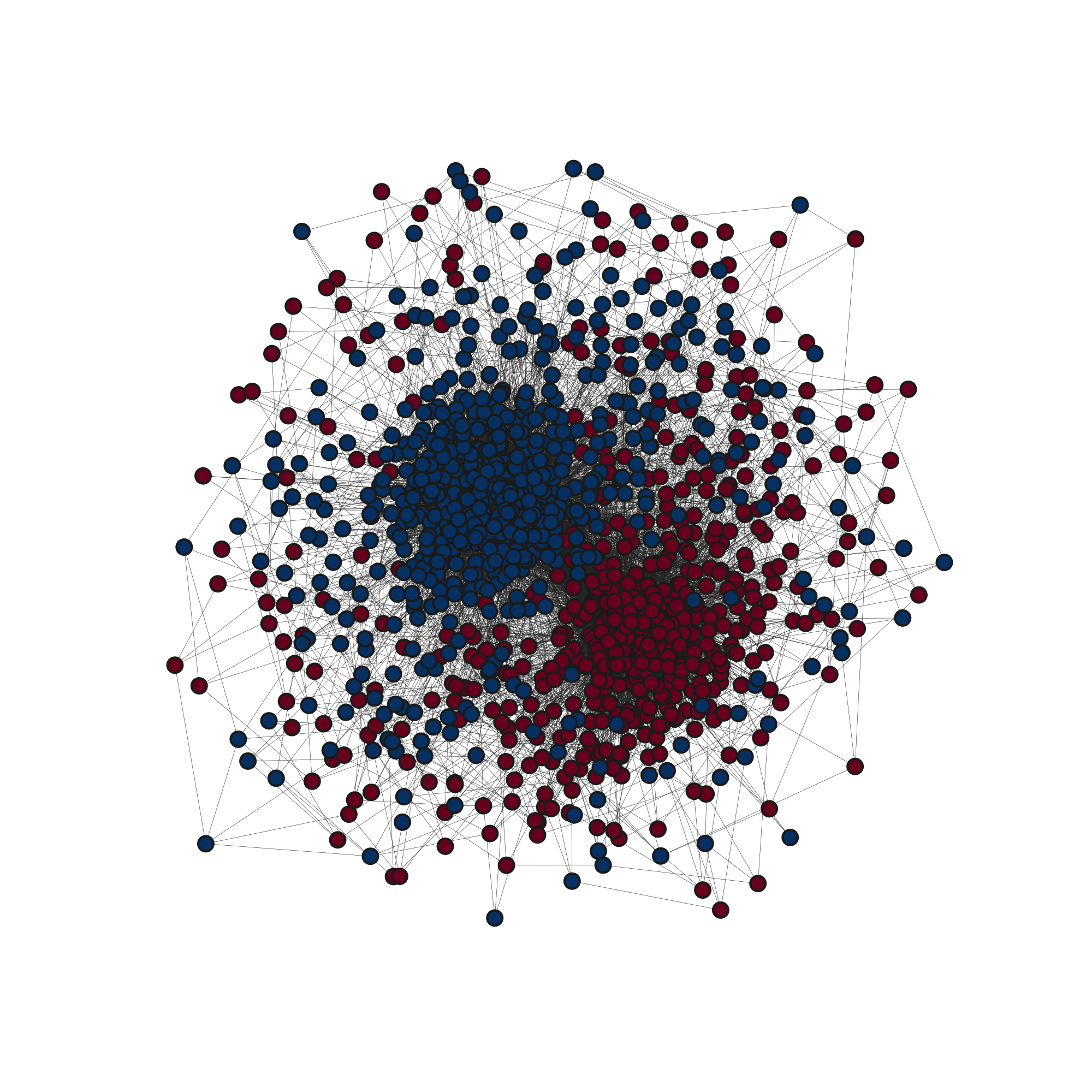}
    	\caption{DS} \label{fig:bg_post_max_dis}	
    \end{subfigure} 
    \begin{subfigure}{0.23\linewidth}
    	\includegraphics[width=\linewidth]{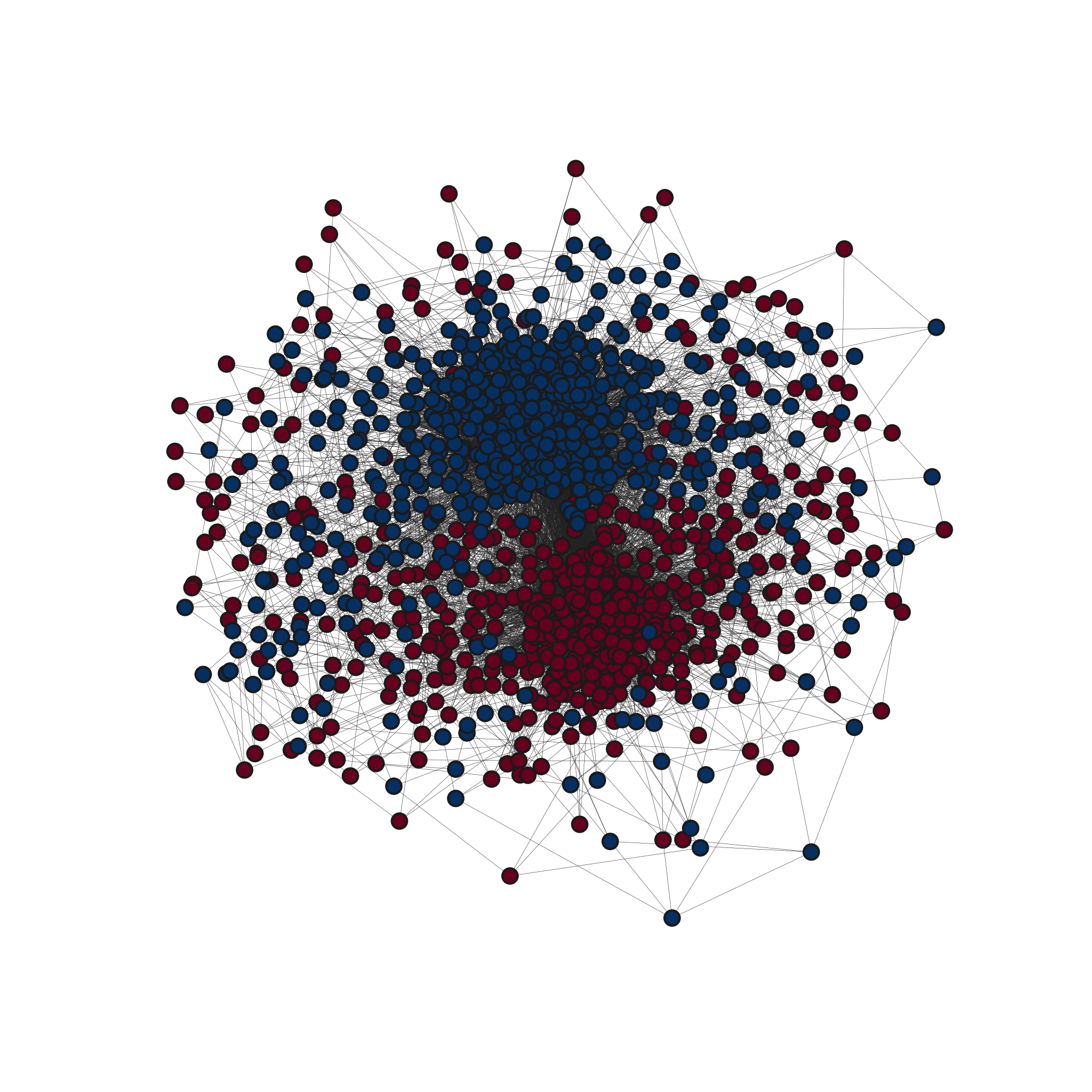}
    	\caption{CD} \label{fig:bg_post_max_grad}	
    \end{subfigure}   
    \begin{subfigure}{0.23\linewidth}
    	\includegraphics[width=\linewidth]{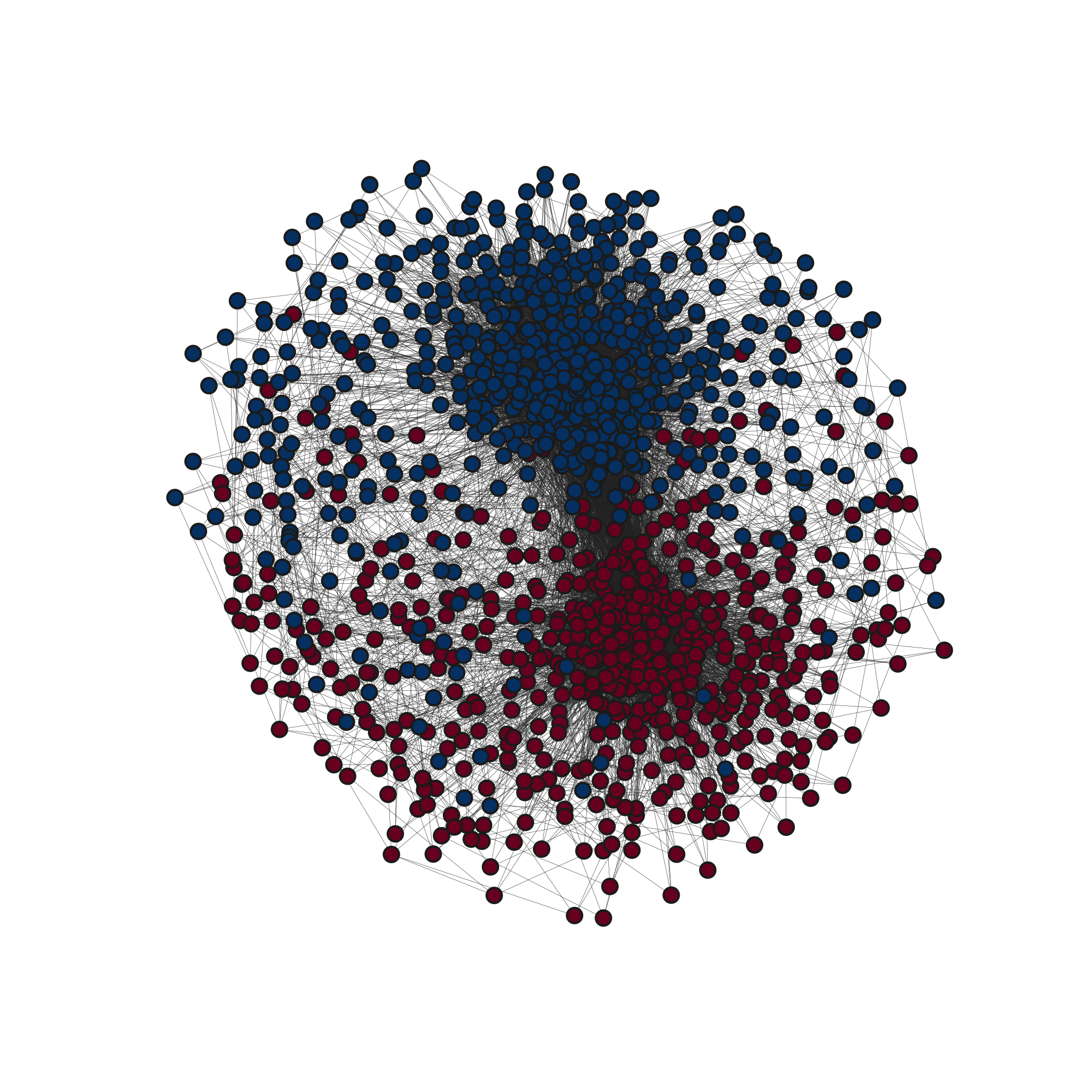}
    	\caption{FD} \label{fig:bg_post_max_fiedler_diff}	
    \end{subfigure} 
    \caption{Evaluation of the planner's heuristics on the political blogs network. Panel (a) shows the reduction achieved as the planner gradually adds edges. Panel (b) shows the initial network, while (c)-(f) visualize the network after the planner has exhausted their budget according to each heuristic. Vertices are colored according to their innate opinions.}
\end{figure}

\begin{figure}[ht]
    \centering
    \begin{subfigure}{0.53\linewidth}
        \centering
	    \includegraphics[width=\linewidth]{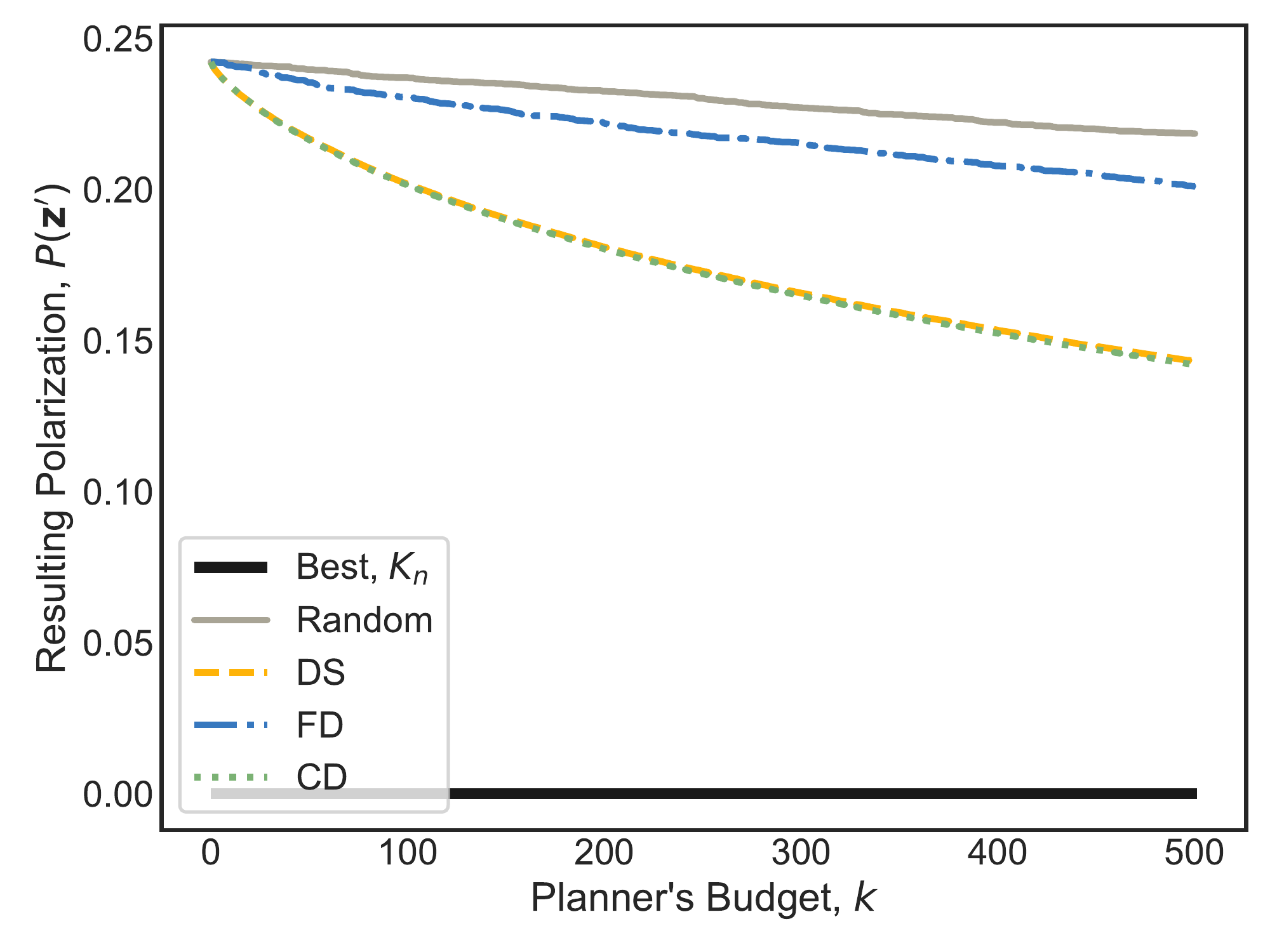}
        \caption{Reduction of Polarization}
        \label{fig:pol_er}
    \end{subfigure}
    \begin{subfigure}{0.05\linewidth}
    	\raisebox{0.2in}{\includegraphics[width=\linewidth]{colormap}}
    \end{subfigure}
    \begin{subfigure}{0.4\linewidth}
    	\includegraphics[trim=0 10 0 0, clip, width=\linewidth]{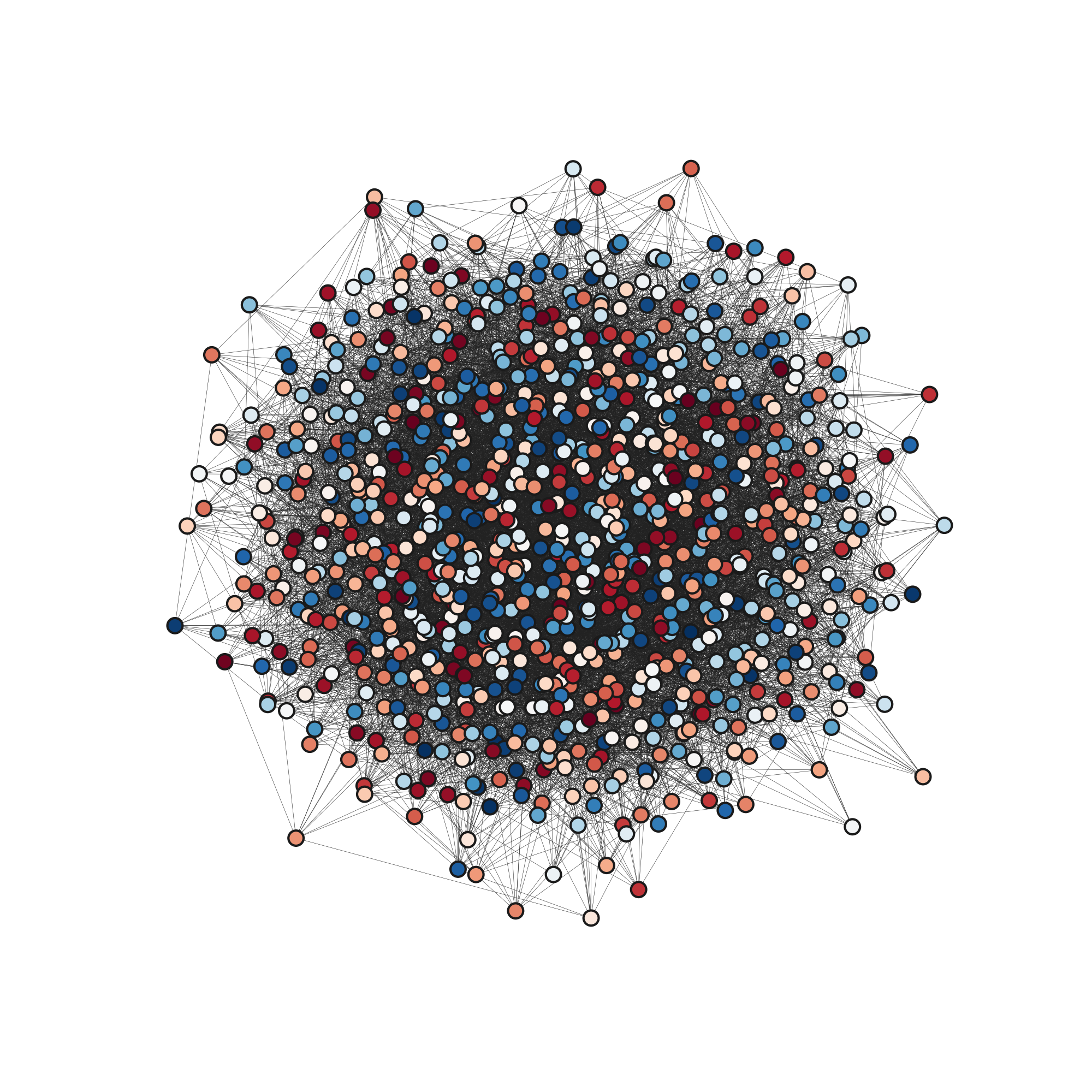}
    	\caption{Initial Graph Structure} \label{fig:er_pre}	
    \end{subfigure}
    \\
    \begin{subfigure}{0.03\linewidth}
    	\raisebox{0.2in}{\includegraphics[width=\linewidth]{colormap}}
    \end{subfigure}
    \begin{subfigure}{0.23\linewidth}
        \centering
    	\includegraphics[width=\linewidth]{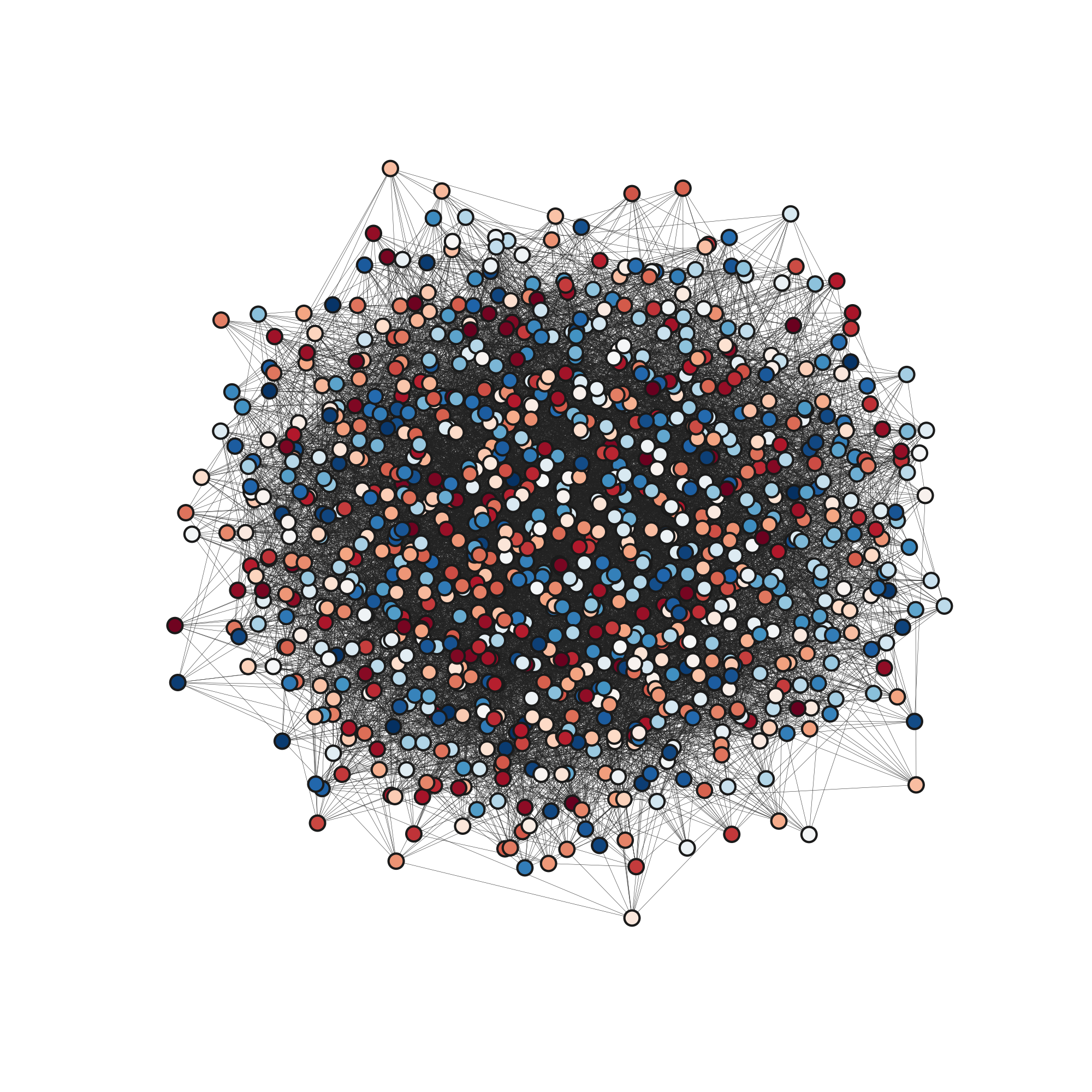}
    	\caption{Random} \label{fig:er_post_random_add}	
    \end{subfigure}
    \begin{subfigure}{0.23\linewidth}
    	\includegraphics[width=\linewidth]{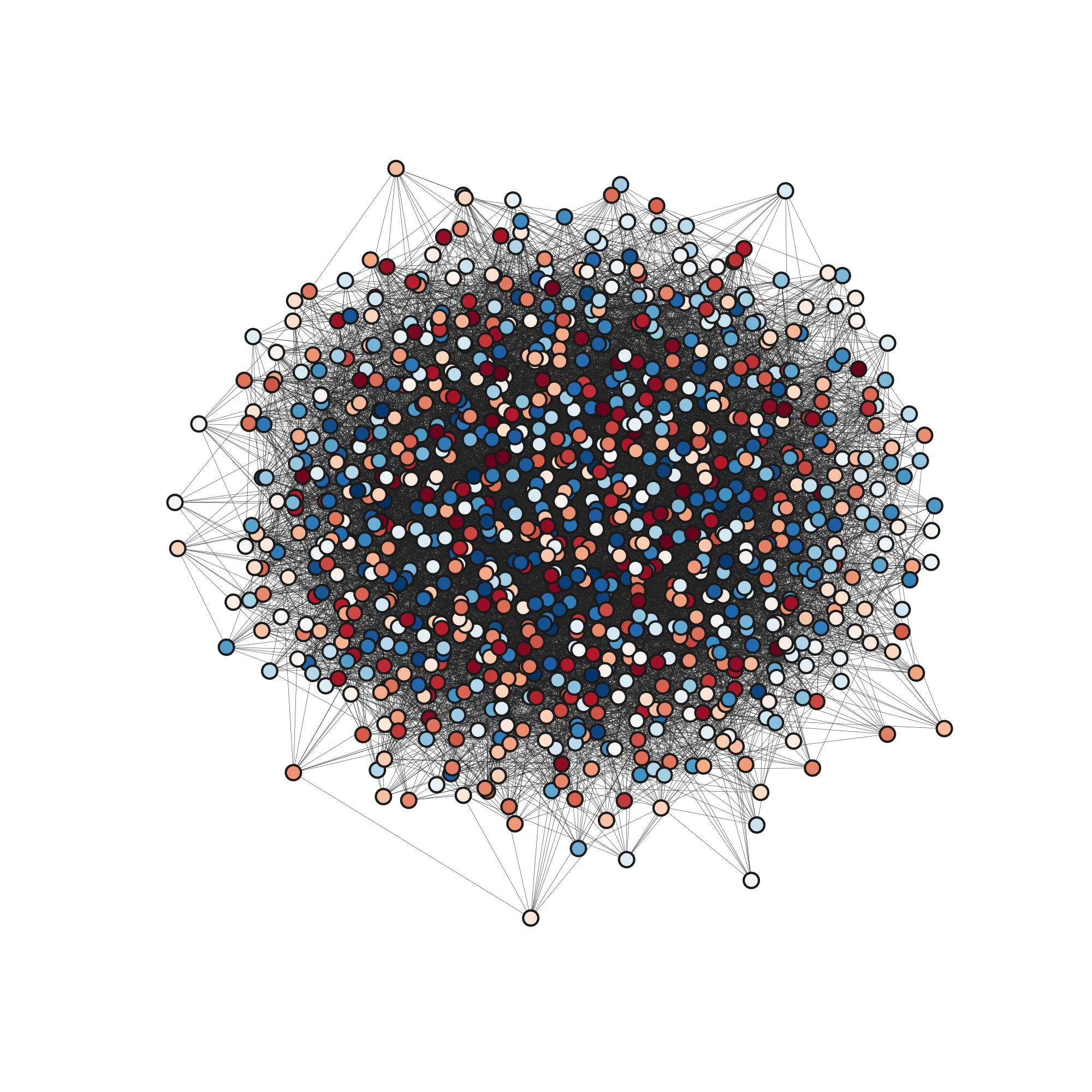}
    	\caption{DS} \label{fig:er_post_max_dis}	
    \end{subfigure} 
    \begin{subfigure}{0.23\linewidth}
    	\includegraphics[width=\linewidth]{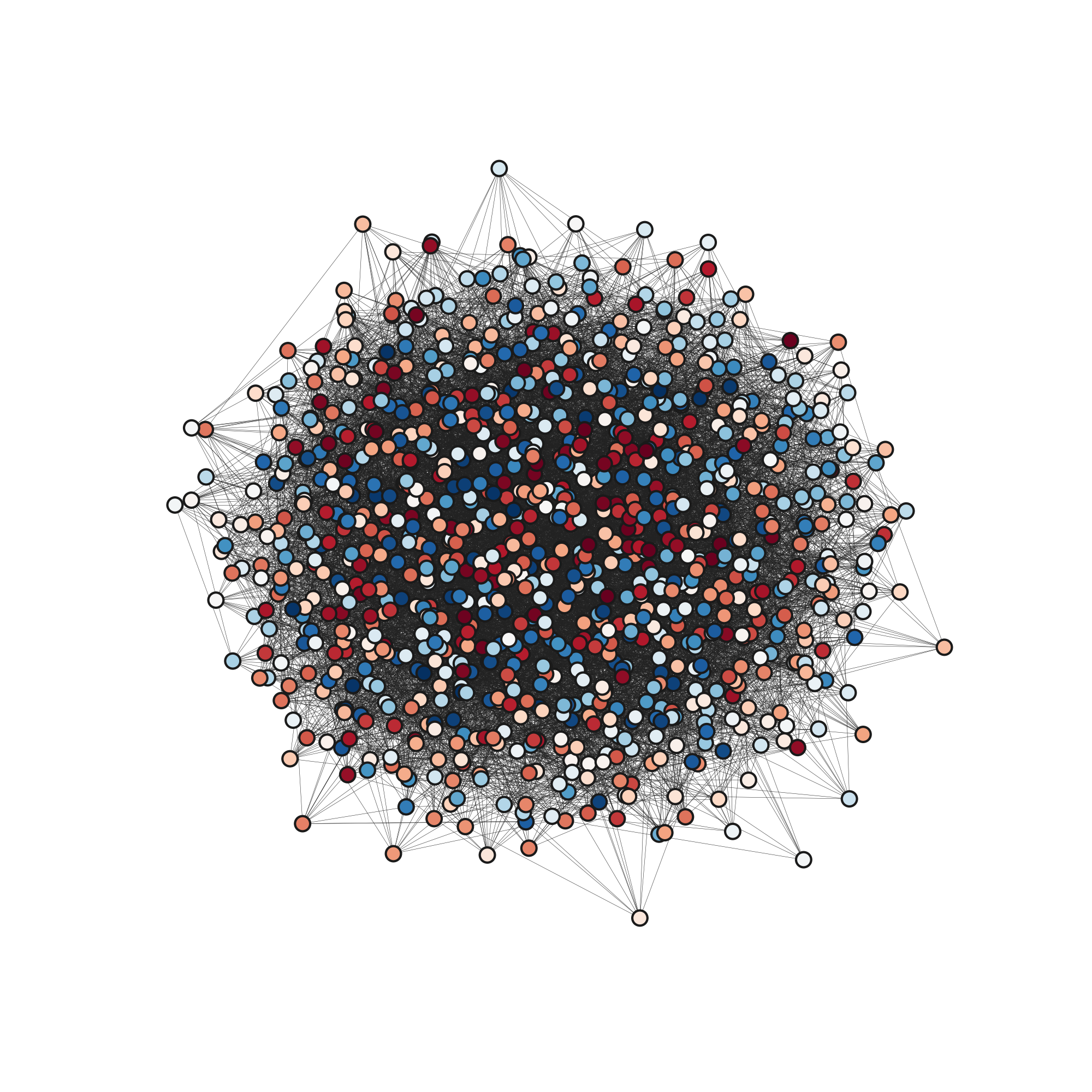}
    	\caption{CD} \label{fig:er_post_max_grad}	
    \end{subfigure}   
    \begin{subfigure}{0.23\linewidth}
    	\includegraphics[width=\linewidth]{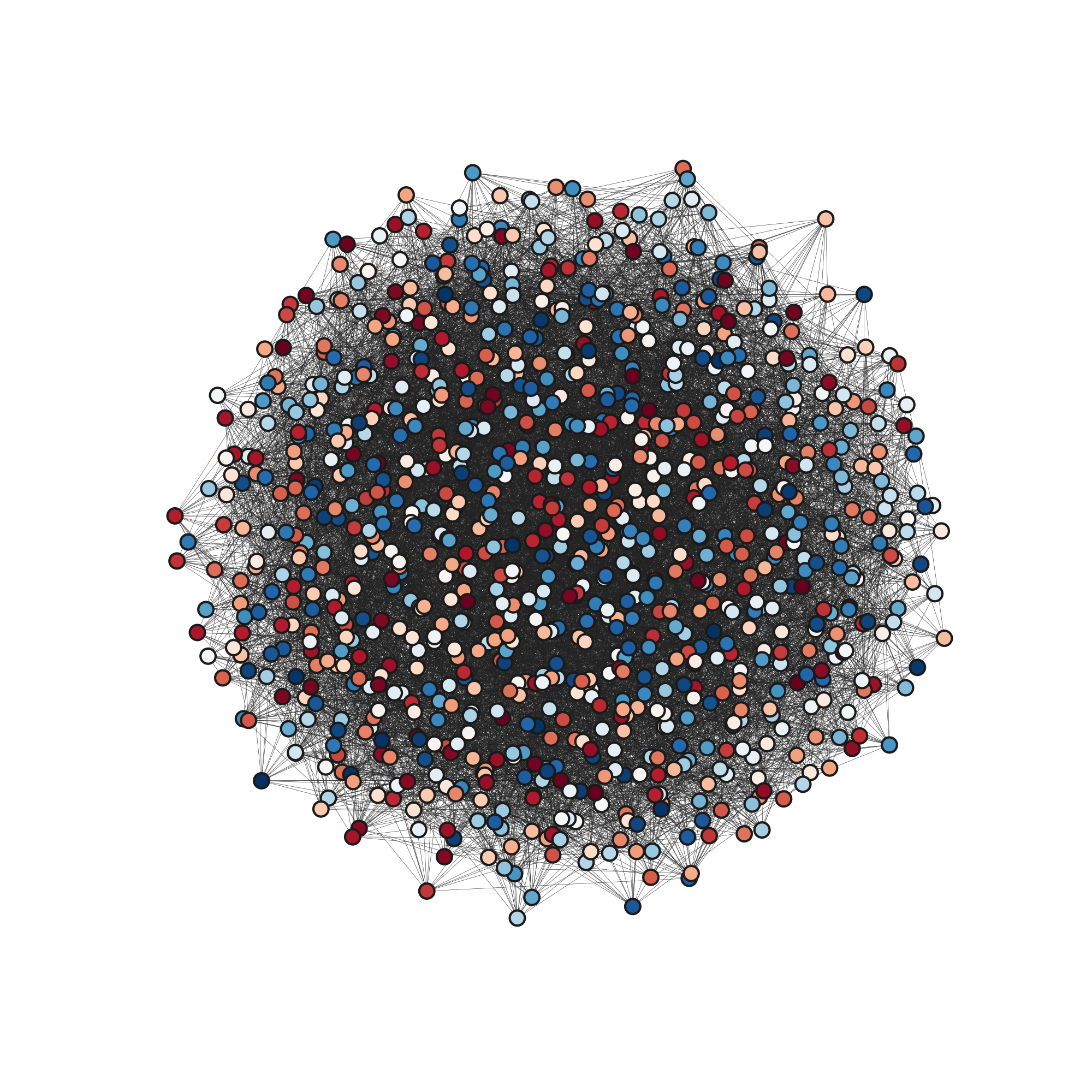}
    	\caption{FD} \label{fig:er_post_max_fiedler_diff}	
    \end{subfigure} 
    \caption{Evaluation of the planner's heuristics on the Erd\H{o}s-R\'enyi graph. Panel (a) shows the reduction achieved as the planner gradually adds edges. Panel (b) shows the initial network, while (c)-(f) visualize the network after the planner has exhausted their budget according to each heuristic. Vertices are colored according to their innate opinions.}
\end{figure}

\begin{figure}[ht]
    \centering
    \begin{subfigure}{0.53\linewidth}
        \centering
	    \includegraphics[width=\linewidth]{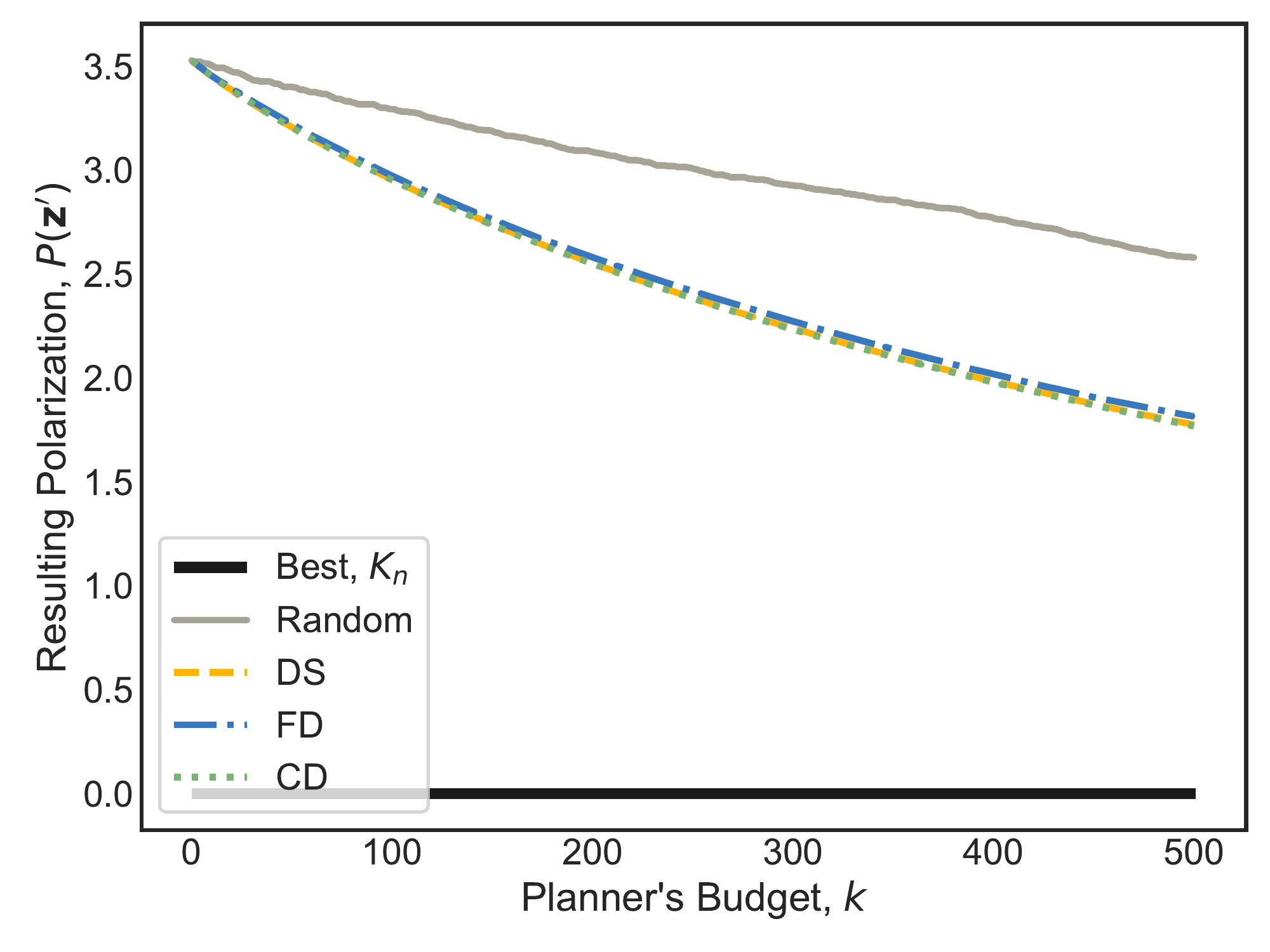}
        \caption{Reduction of Polarization}
        \label{fig:pol_sbm}
    \end{subfigure}
    \begin{subfigure}{0.05\linewidth}
    	\raisebox{0.2in}{\includegraphics[width=\linewidth]{colormap}}
    \end{subfigure}
    \begin{subfigure}{0.4\linewidth}
    	\includegraphics[trim=0 10 0 0, clip, width=\linewidth]{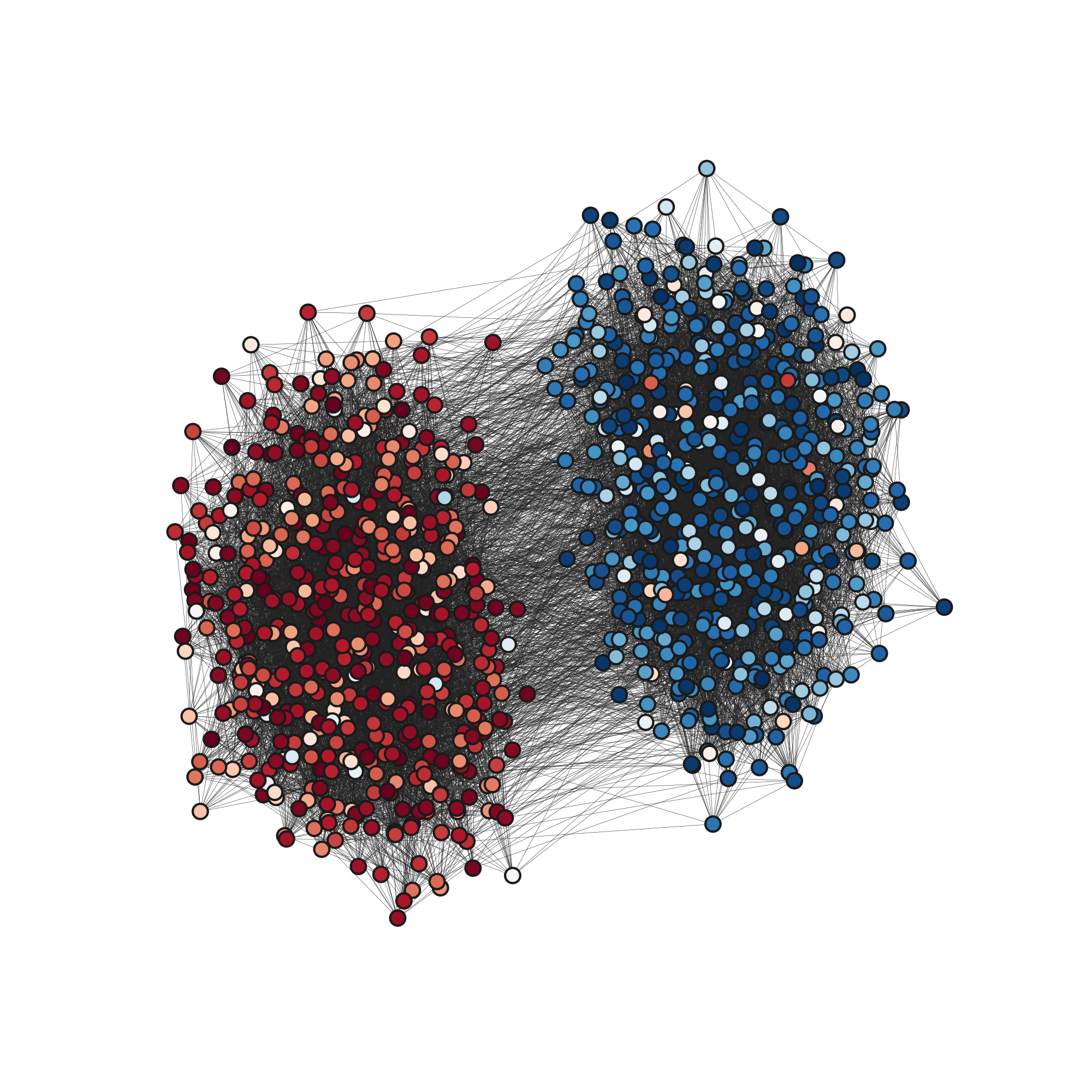}
    	\caption{Initial Graph Structure} \label{fig:sbm_pre}	
    \end{subfigure}
    \\
    \begin{subfigure}{0.03\linewidth}
    	\raisebox{0.2in}{\includegraphics[width=\linewidth]{colormap}}
    \end{subfigure}
    \begin{subfigure}{0.23\linewidth}
        \centering
    	\includegraphics[width=\linewidth]{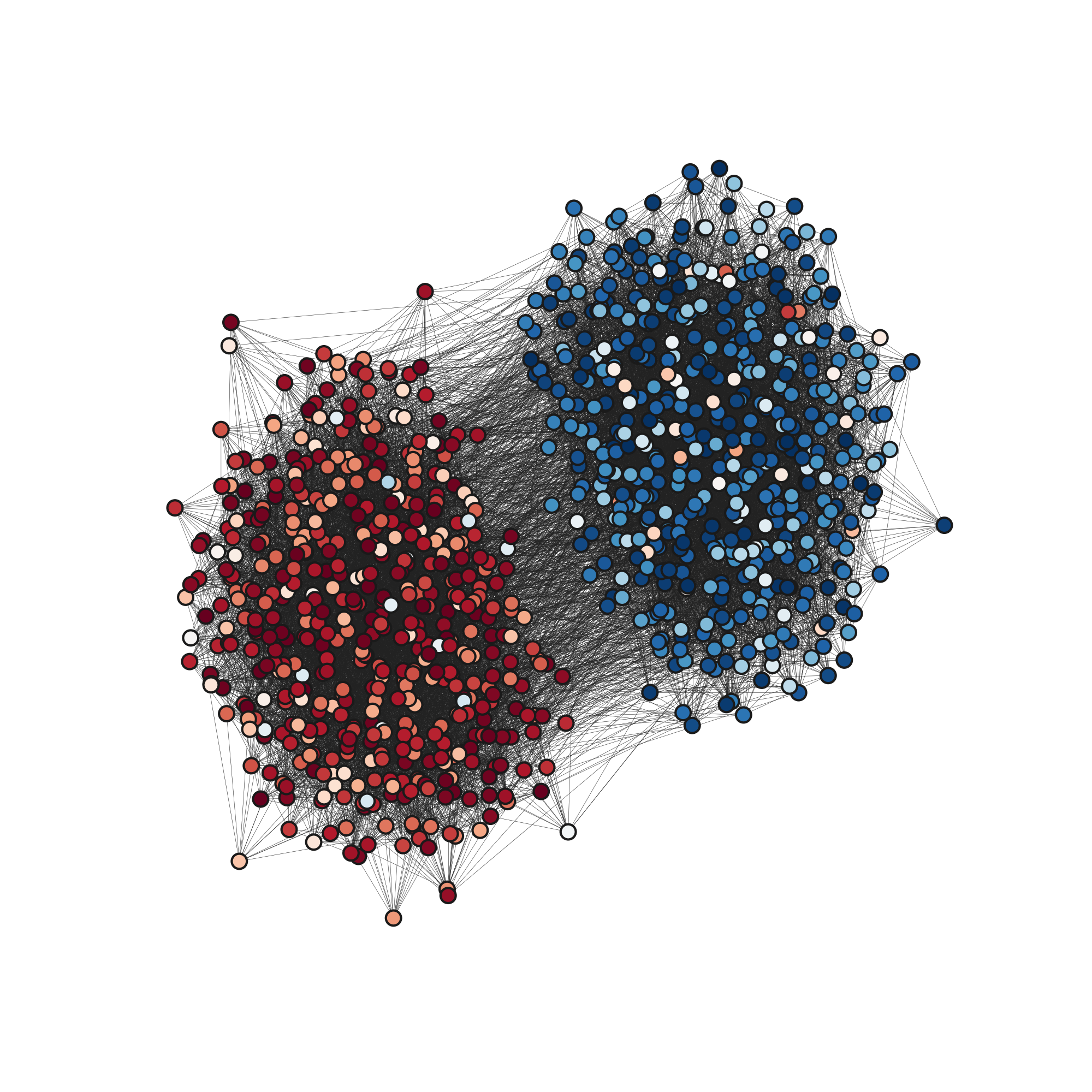}
    	\caption{Random} \label{fig:sbm_post_random_add}	
    \end{subfigure}
    \begin{subfigure}{0.23\linewidth}
    	\includegraphics[width=\linewidth]{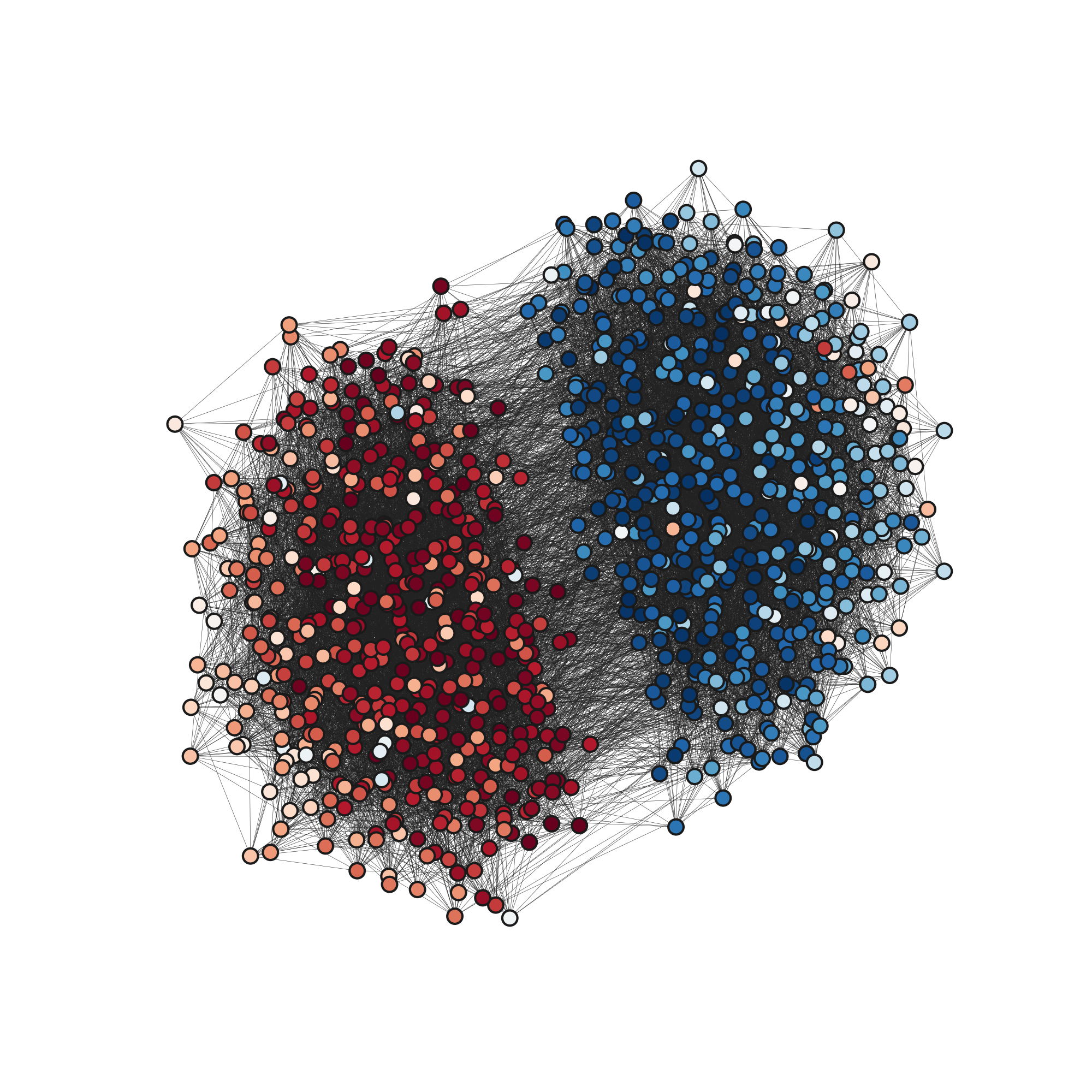}
    	\caption{DS} \label{fig:sbm_post_max_dis}	
    \end{subfigure} 
    \begin{subfigure}{0.23\linewidth}
    	\includegraphics[width=\linewidth]{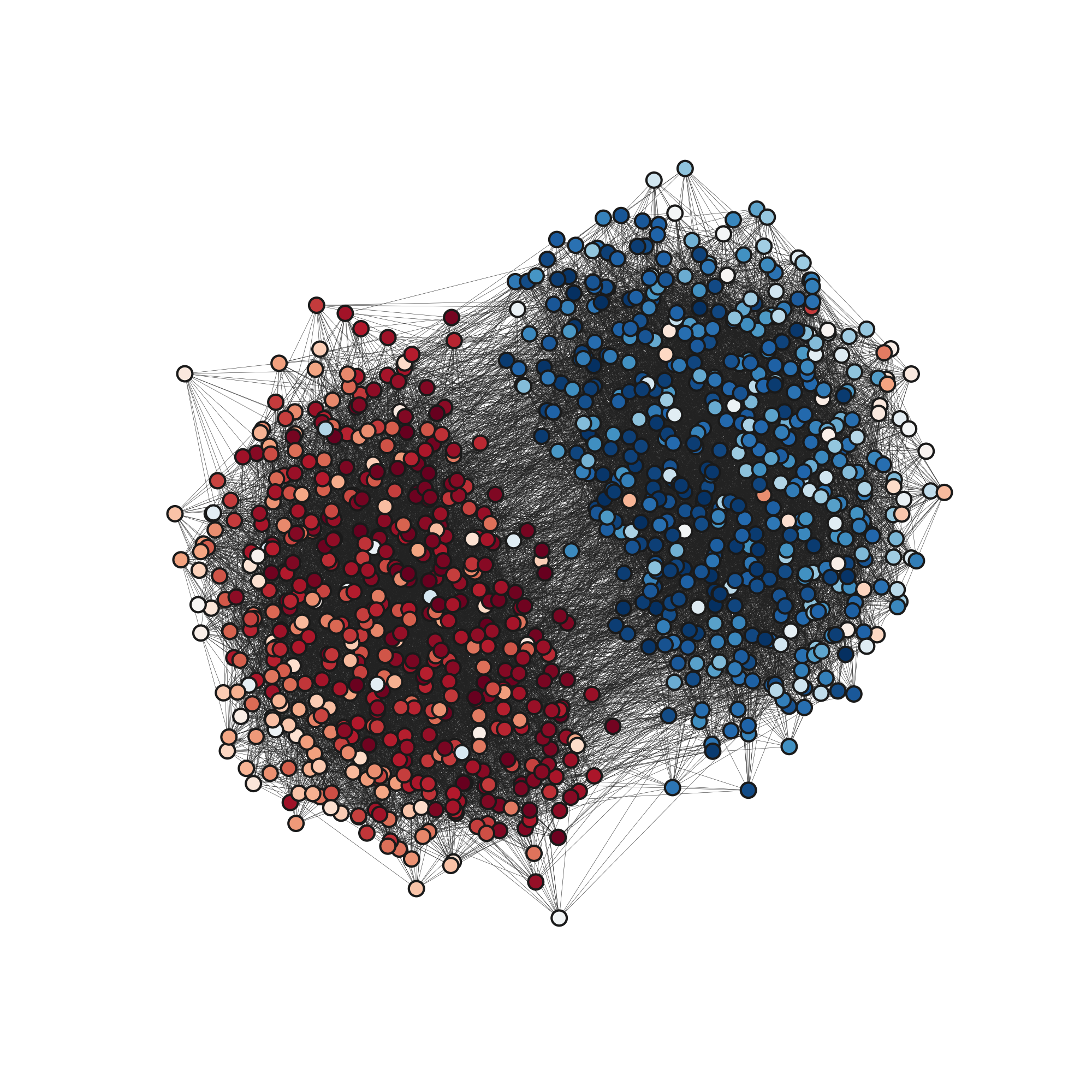}
    	\caption{CD} \label{fig:sbm_post_max_grad}	
    \end{subfigure}   
    \begin{subfigure}{0.23\linewidth}
    	\includegraphics[width=\linewidth]{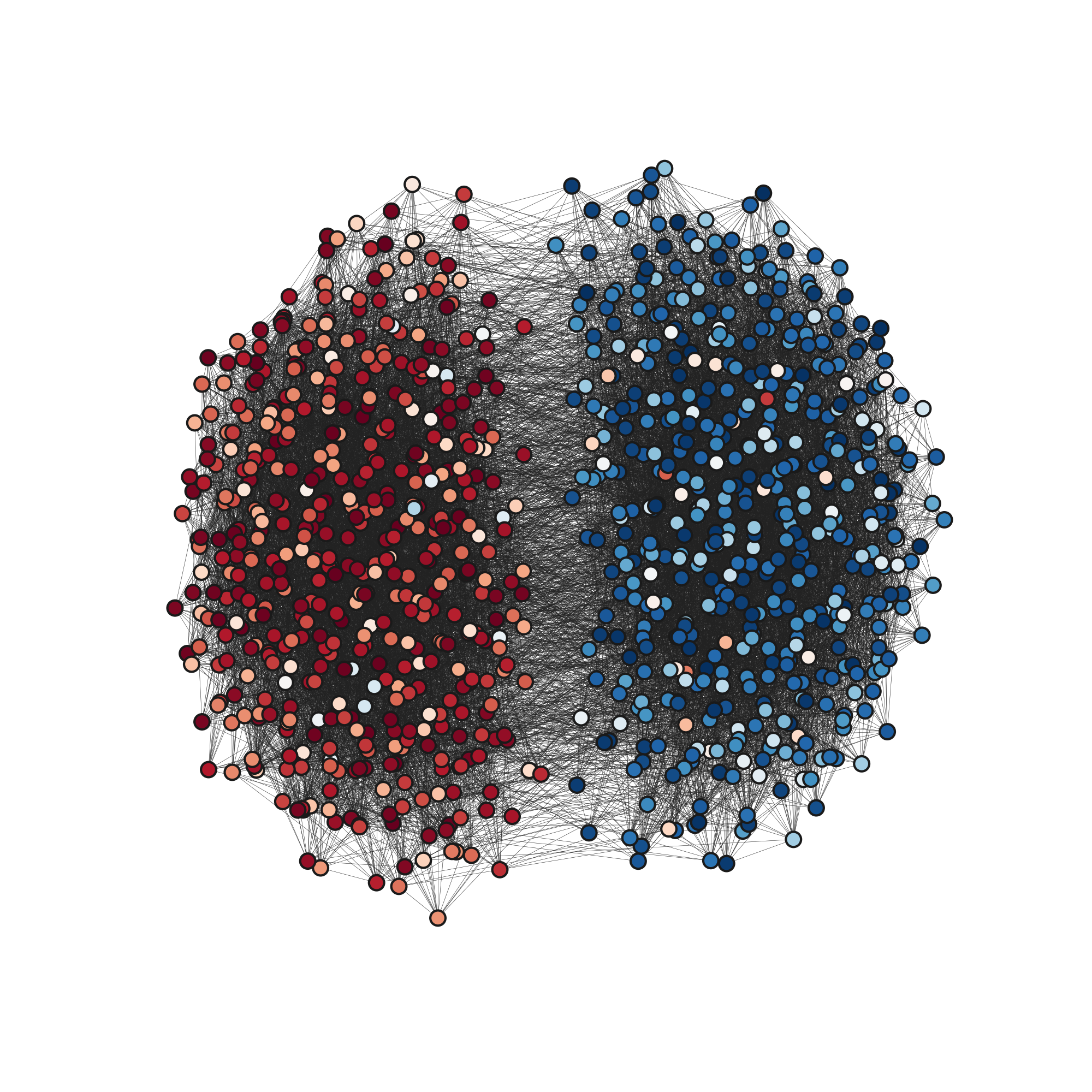}
    	\caption{FD} \label{fig:sbm_post_max_fiedler_diff}	
    \end{subfigure} 
    \caption{Evaluation of the planner's heuristics on the stochastic block model graph. Panel (a) shows the reduction achieved as the planner gradually adds edges. Panel (b) shows the initial network, while (c)-(f) visualize the network after the planner has exhausted their budget according to each heuristic. Vertices are colored according to their innate opinions.}
\end{figure}

\begin{figure}[ht]
    \centering
    \begin{subfigure}{0.53\linewidth}
        \centering
	    \includegraphics[width=\linewidth]{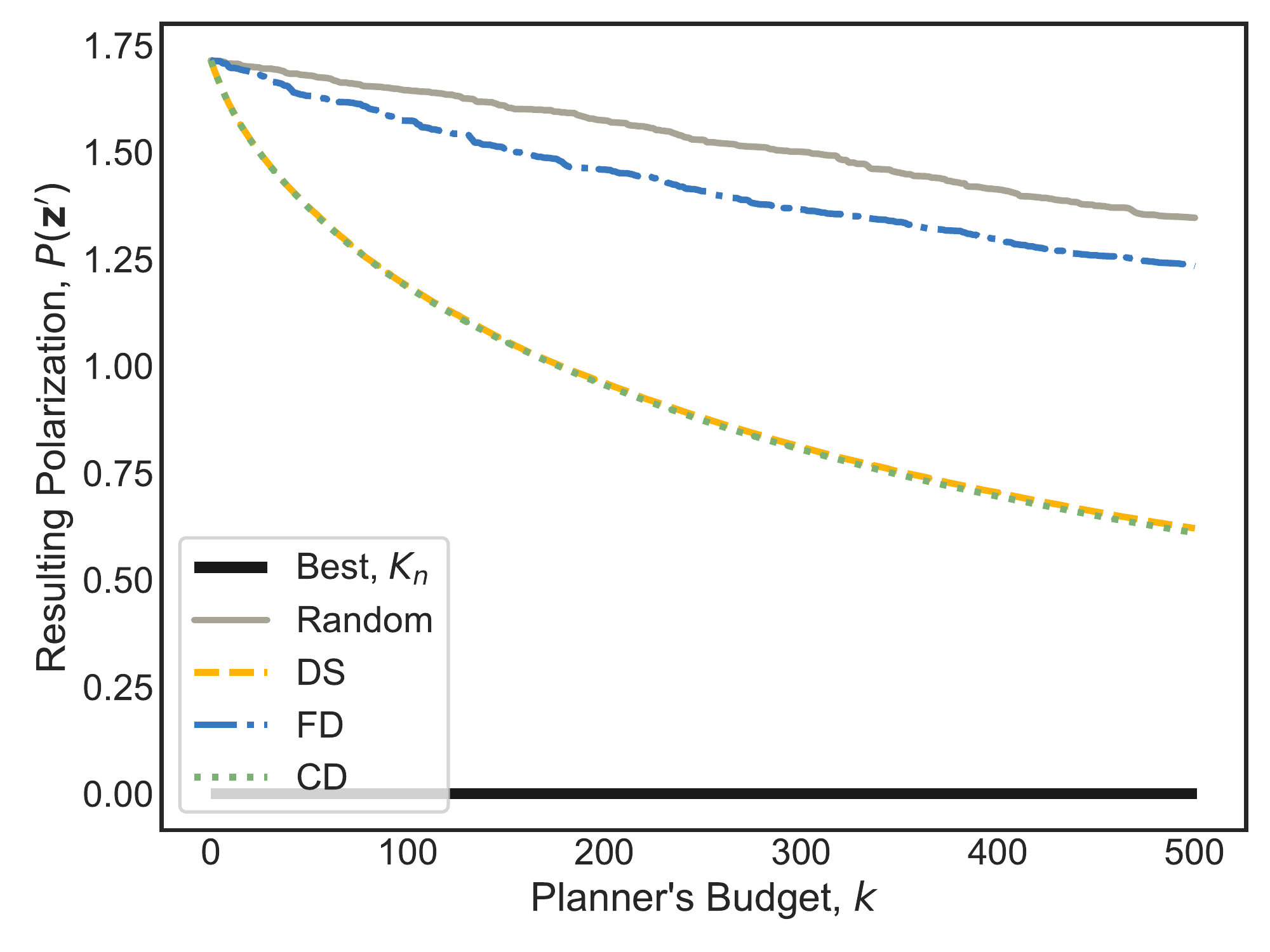}
        \caption{Reduction of Polarization}
        \label{fig:pol_pa}
    \end{subfigure}
    \begin{subfigure}{0.05\linewidth}
    	\raisebox{0.2in}{\includegraphics[width=\linewidth]{colormap}}
    \end{subfigure}
    \begin{subfigure}{0.4\linewidth}
    	\includegraphics[trim=0 10 0 0, clip, width=\linewidth]{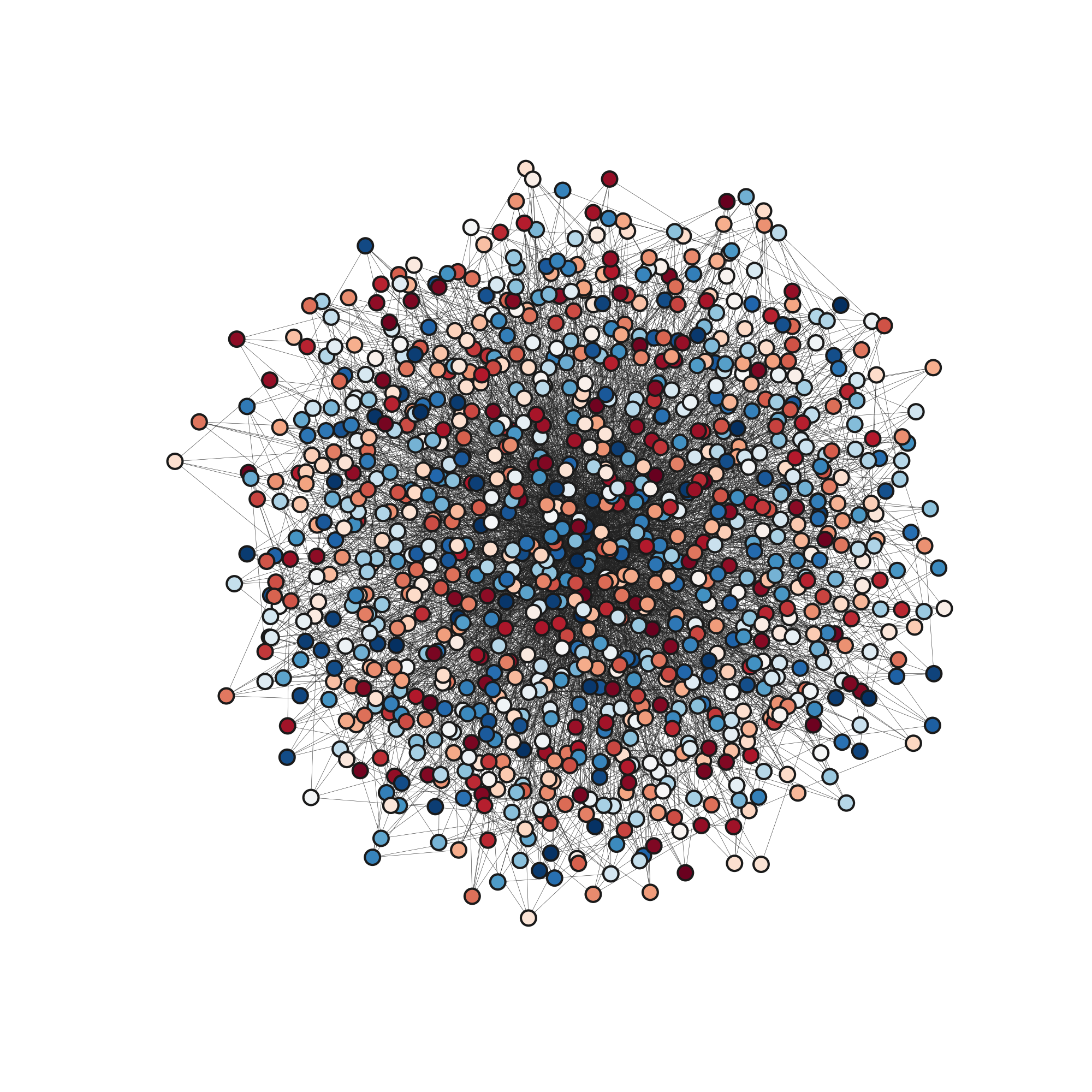}
    	\caption{Initial Graph Structure} \label{fig:pa_pre}	
    \end{subfigure}
    \\
    \begin{subfigure}{0.03\linewidth}
    	\raisebox{0.2in}{\includegraphics[width=\linewidth]{colormap}}
    \end{subfigure}
    \begin{subfigure}{0.23\linewidth}
        \centering
    	\includegraphics[width=\linewidth]{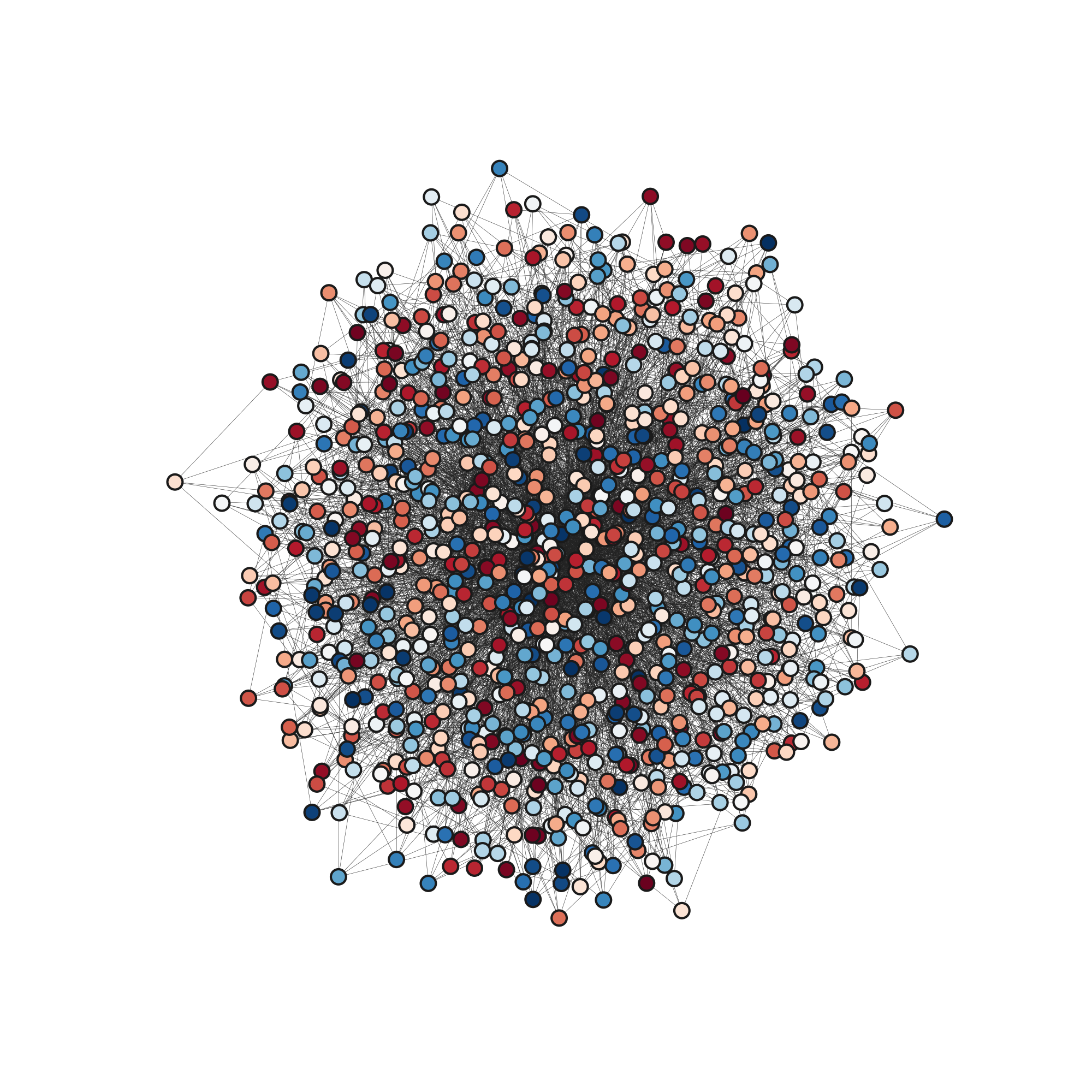}
    	\caption{Random} \label{fig:pa_post_random_add}	
    \end{subfigure}
    \begin{subfigure}{0.23\linewidth}
    	\includegraphics[width=\linewidth]{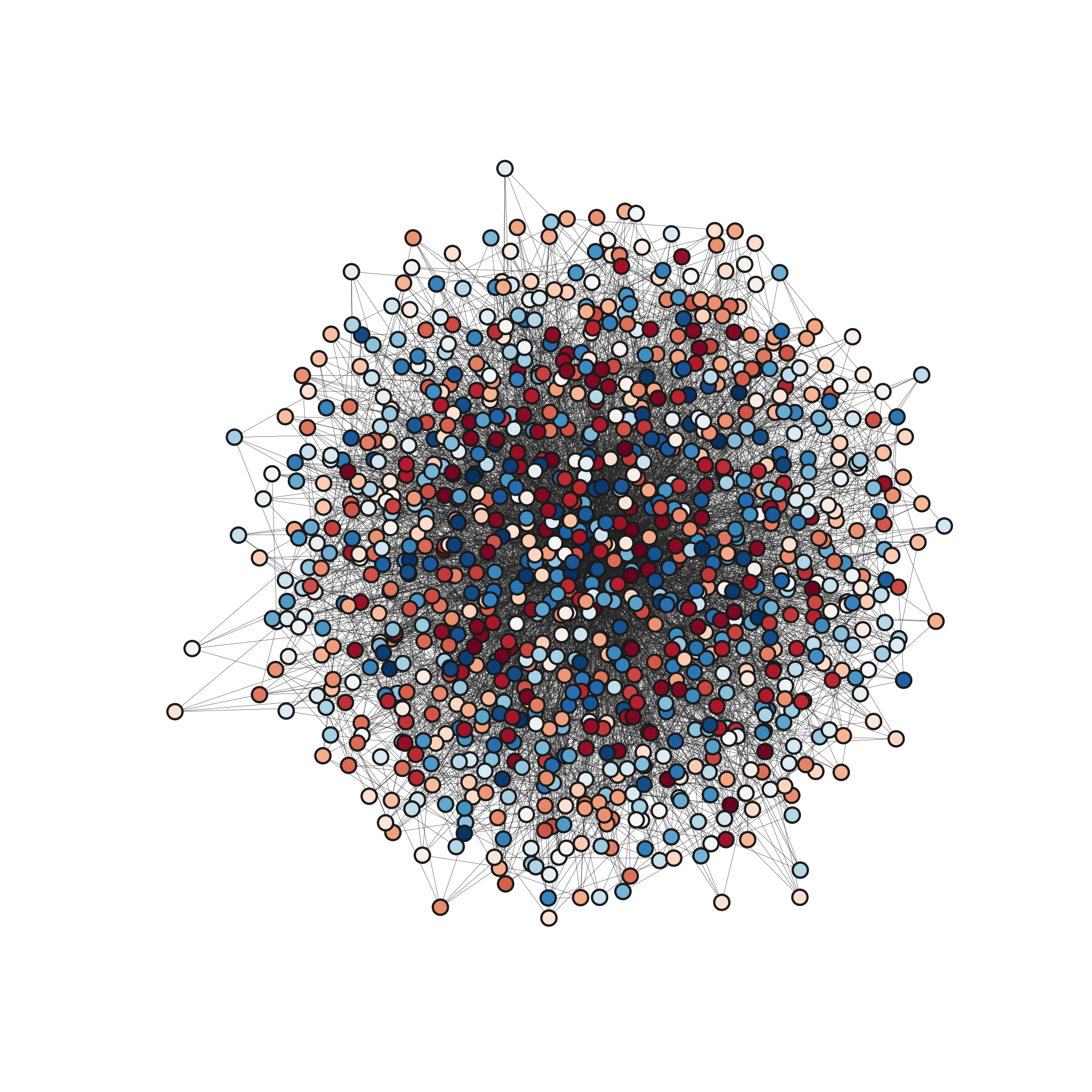}
    	\caption{DS} \label{fig:pa_post_max_dis}	
    \end{subfigure} 
    \begin{subfigure}{0.23\linewidth}
    	\includegraphics[width=\linewidth]{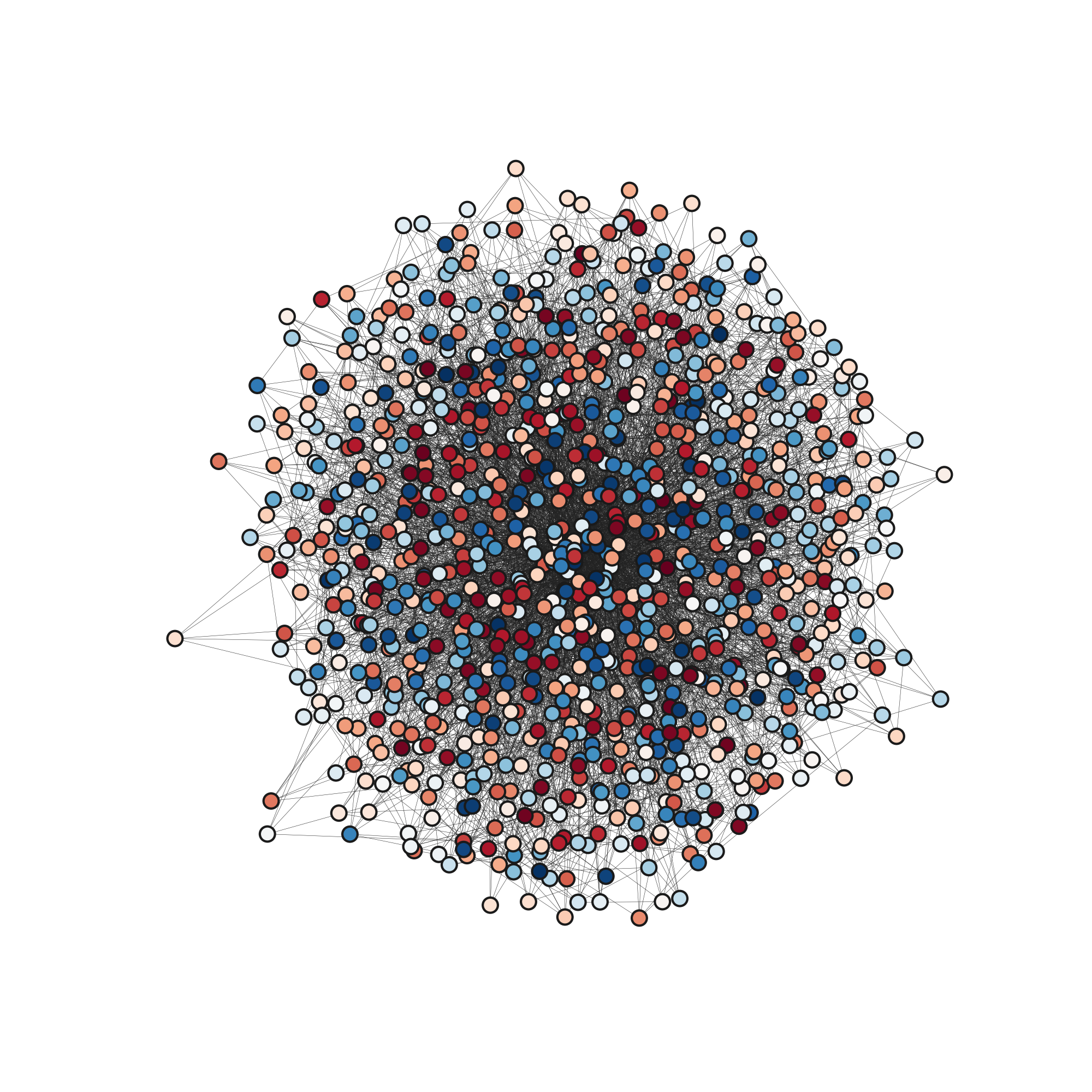}
    	\caption{CD} \label{fig:pa_post_max_grad}	
    \end{subfigure}   
    \begin{subfigure}{0.23\linewidth}
    	\includegraphics[width=\linewidth]{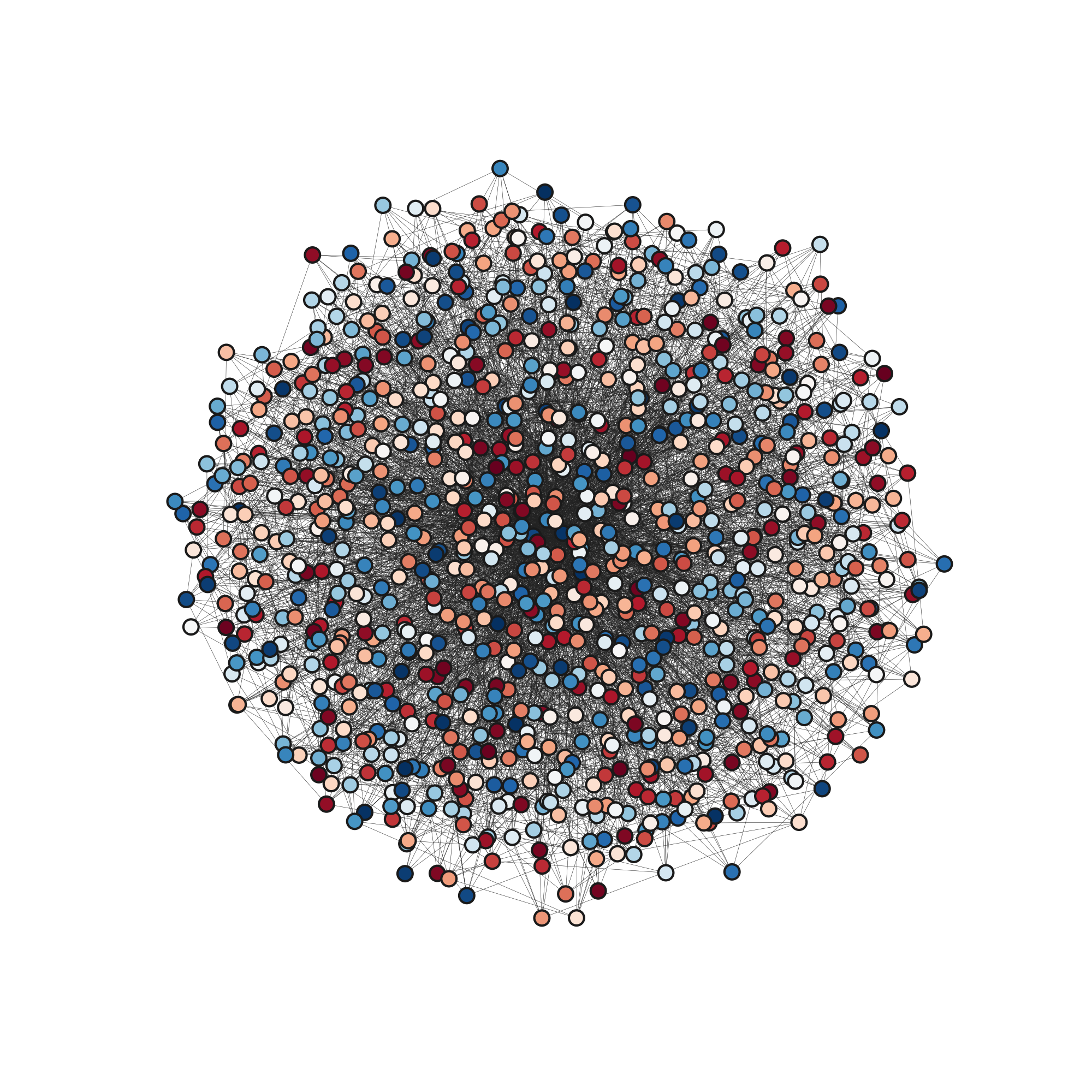}
    	\caption{FD} \label{fig:pa_post_max_fiedler_diff}	
    \end{subfigure} 
    \caption{Evaluation of the planner's heuristics on the preferential attachment graph. Panel (a) shows the reduction achieved as the planner gradually adds edges. Panel (b) shows the initial network, while (c)-(f) visualize the network after the planner has exhausted their budget according to each heuristic. Vertices are colored according to their innate opinions.}
\end{figure}

\subsection{Real-World Networks}
\label{sec:emp_res:real}

The Twitter and Reddit datasets used in this section were first collected by~\citet{De2014}, and used by both~\citet{Chen2022} and~\citet{Musco2018} in recent work. An additional dataset comprised of political blogs was collected by~\citet{adamic2005political} and used in~\citet{matakos2017measuring, matakos2020tell}.

\paragraph{Twitter:} 
This network reflects individuals who tweeted about a Delhi assembly debate in 2013. The network is shown in Fig.~\ref{fig:tw_pre}, and mainly consists of two communities.

Fig.~\ref{fig:pol_tw} shows the reduction in polarization achieved by the planner when applying each of the heuristics. Notably, all heuristics outperform our simple baseline. For the two best-performing heuristics, the first 50 edges modified reduce polarization by about a factor of two, and the subsequent 50 achieve a similar fractional reduction. This highlights both the substantial effect that the planner can have with minimal modifications to the graph, along with the diminishing returns of their budget.

The networks resulting from the planner's heuristics are shown in Figs.~\ref{fig:tw_post_random_add}-f. There are notable reductions in the strength of community structures. While less effective in reducing polarization, the Fiedler vector-based heuristic (FD) appears to smooth out communities the most.

\paragraph{Reddit:} 
This network was generated by following Reddit users who posted in a politics forum. Three isolated vertices are removed in preprocessing. Fig.~\ref{fig:rd_pre} shows that the initial network appears to be tightly clustered, and Table~\ref{tab:comp_results} indicates that it exhibits an extremely small level of polarization.

For any non-baseline heuristic, the full budget reduces polarization by almost a factor of four. For the best-performing heuristics, this reduction is by nearly an order of magnitude. We observe greatly diminishing returns, with the most significant reduction achieved with the first few edges modified. Moreover, the best-performing heuristics come close to achieving the globally optimal solution after fully exhausting the budget.

Only minor changes are observed in the resulting graph structures. Figs.~\ref{fig:rd_post_random_add}, \ref{fig:rd_post_max_dis}, and~\ref{fig:rd_post_max_grad} look almost identical to the initial network. In contrast, the graph in Fig.~\ref{fig:rd_post_max_fiedler_diff} does not have as dense a core, and appears to be more evenly connected. Since maximizing the spectral gap results in the graph behaving similarly to an expander, which (informally) is equally well-connected across all cuts, this is to be expected.

\paragraph{Blogs:} 
This network was collected by aggregating online directories of political blogs around the 2004 US elections. Note that vertices in this network represent blogs -- not individuals as in the previous datasets. Each blog was identified as either `conservative' or `liberal', which we encode by innate opinions of $0$ or $1$, respectively. Observe in Table~\ref{tab:comp_results} that this network exhibits extremely large values of polarization and homophily, and a small spectral gap.

We find consistent reductions in polarization with all heuristics -- including the baseline. This network is unique in that the community structure largely mirrors the innate opinions. That is, the mean-centered innate opinions vector is highly collinear with the Fiedler vector. Hence, both the DS and FD heuristics will choose to modify edges between the two communities. Furthermore, a large fraction of the non-edges span the two communities -- a randomly chosen edge is therefore likely to bridge the two.

Fig.~\ref{fig:bg_post_random_add} shares with Fig.~\ref{fig:bg_pre} a tightly-knit core, with a few vertices at the extremities. In contrast, Figures~\ref{fig:bg_post_max_dis}-f depict networks that are more uniformly connected. As before, we find feature this to be most observable with use of the FD heuristic.

\subsection{Synthetic Datasets}
\label{sec:emp_res:synth}

These heuristics are also applied to three canonical models of random graphs. \change{In the first two models, the number of edges grows quadratically in the number of vertices (for fixed parameters). However, in our final model, the number of edges is always linear in the number of vertices. Therefore, one may expect that the impact of the planner's $O(n)$ edges is greatest in the sparser model -- but we will see that this is not the case.}

\paragraph{Erd\H{o}s-R\'enyi:} 
A graph from this model connects each pair of vertices independently with a fixed probability $p\in[0,1]$. We take $n=1000$ and $p = 0.02$, although the results are qualitatively similar for different values. The innate opinions are independent uniform random variables in $[0,1]$.

This model produces homogeneous, well-connected networks, which are good spectral approximations of the complete graph $K_n$~\citep{hoffman2021spectral}. This can be seen through the large initial spectral gap in Table~\ref{tab:comp_results}. Therefore, according to Proposition~\ref{prop:p_bounds} it is natural to expect polarization to be small. Nonetheless, all heuristics fail to significantly reduce polarization. A comparison with the random baseline is particularly interesting in this model, as it results in another Erd\H{o}s-R\'enyi graph, but with a slightly larger value of $p$.

Few changes can be seen among Figs.~\ref{fig:er_post_random_add}-f. However, there are few vertices with extreme opinions in the fringes of Fig.~\ref{fig:er_post_max_dis}. Instead, these vertices tend to be concentrated in the center of the graph. This aligns with the most negative assortativity seen in Table~\ref{tab:comp_results}. Notably, this feature is not visible in Figures~\ref{fig:er_post_max_fiedler_diff} or~\ref{fig:er_post_random_add}.

\paragraph{Stochastic Block Model:} 
A graph drawn from a stochastic block model can replicate community structures, and is shown in Fig.~\ref{fig:sbm_pre}. This random graph on $n=1000$ vertices with two equal-sized communities is generated by mirroring the methodology in~\citet{Chen2022}. Specifically, the probability of inter-community edges is given by $q = 0.005$, and the probability of intra-community edges is $p=0.05$. Since $p>q$, we expect to see strong communities. The innate opinions of vertices in each community are drawn independently from either $\mathrm{Beta}(1,5)$ or $\mathrm{Beta}(5,1)$, such that the distribution of opinions mirrors the graph's community structure.

In Fig.~\ref{fig:pol_sbm}, a nearly identical reduction in polarization can be seen for all non-random heuristics. This occurs because the mean-centered innate opinions are highly collinear with the Fiedler vector, which partitions the graph into its two communities. Therefore, both the DS and FD strategies will add edges between the two communities. Similarly to the Erd\H{o}s-R\'enyi setting, the random baseline yields another stochastic block graph, but with slightly larger parameters $p$ and $q$.

Qualitatively, all three heuristics can be seen to bring the two communities closer together. However, in Fig.~\ref{fig:sbm_post_max_dis} and~\ref{fig:sbm_post_max_grad}, the vertices with extreme opinions are brought closer to the center. As before, this is not observed in Fig.~\ref{fig:sbm_post_max_fiedler_diff}.

\paragraph{Preferential Attachment Model:} 
This model generates graphs with power-law degree distribution, often known as scale-free or Barab\'asi-Albert networks~\citep{barabasi1999emergence}. Again, we follow a similar procedure to~\citet{Chen2022}, with $n=1000$ vertices added sequentially. Each incoming vertex connects to at most $m=5$ vertices, where the likelihood of connecting to a particular vertex is proportional to its degree.

This graph tends to exhibit a small, highly interconnected core, and many vertices with low degree. This structure is not conducive to low polarization, as we see in Fig.~\ref{fig:pol_pa}. The best-performing heuristics manage to reduce polarization by just over a factor of two, whereas the others see only negligible fractional reductions. Notably, the FD heuristic only slightly outperforms the baseline. \change{These observations are a result of the Friedkin-Johnsen model -- higher-degree nodes experience the smallest marginal effects of increased edge weight. Since the preferential attachment model yields a heavy-tailed degree distribution, a larger fraction of nodes are resilient to the planner's modifications. These nodes will also exert large amounts of influence on their neighbors due to their high degree. We therefore believe that the structural properties of the preferential attachment graph dampen the planner's effectiveness.}

Qualitatively, Figs.~\ref{fig:pa_post_max_dis}-f appear similar to the original network in Fig.~\ref{fig:pa_pre}. We do not see strong changes in the structure, which aligns with the minor differences in homophily and spectral gap in Table~\ref{tab:comp_results}.

\section{Discussion and Conclusion}
\label{sec:conclusion}

In this paper, we analyze the connection between the structure of social and information networks and opinion polarization.

First, we establish a relationship between the ratio of \emph{expressed} to \emph{innate} polarization. This ratio is controlled by structural properties of the graph, such as the degree profile and isoperimetric number (i.e. Cheeger constant). In particular, the worst-case polarization depends directly on the spectral gap of the Laplacian. Consequentially, we show that the complete graph achieves the global minimum for polarization. This result aligns with one's intuition -- bottlenecks in the graph are liabilities to a consensus.

Next, this paper presents two variations of the planner's problem -- one in which the innate opinions of the population are known, and another in which they are chosen adversarially. In the first, an expression is derived for the exact difference in polarization when weight is added to a single edge. We find that strengthening the connections between vertices with large \emph{expressed} disagreement reduces polarization. However, it is seen as costly for individuals to interact with differently-minded others~\citep{Bindel2015}. Therefore, reaching a consensus, while arguably beneficial for the population, may prove costly to individuals. We also present a second setting wherein the planner defends the network against adversarially-controlled innate opinions. Here, we prove that the planner aims to maximize the spectral gap. We then show the effectiveness of a strategy based on the Fiedler vector $\mathbf{v}$ -- the eigenvector corresponding to the spectral gap. Intuitively, this vector partitions the graph based on the signs of its elements, and the planner should strengthen edges across the cut.

Finally, we evaluate the performance of four heuristics on several real-world and synthetic networks. We find that all strategies may smooth out community structures \change{-- often referred to as `echo-chambers'}. Furthermore, when there are no strong communities present, the Fiedler vector-based strategy is able to reduce polarization without simultaneously reducing homophily. With this approach, a reduction in polarization did not necessitate direct connections between opposite-minded individuals. However, this strategy performed significantly worse in several networks. We believe that the difference reflects how much of the polarization is driven by the particular values of opinions. For instance, all three heuristics perform similarly when \change{the profile of opinions mirrors the graph structure, and therefore both contribute similarly to the level of polarization. Specifically, all heuristics behave similarly when the mean-centered innate opinions $\widetilde{\mathbf{s}}$ are highly aligned with the Fiedler vector} -- this observation can be seen most easily in the blogs and stochastic block model networks.

There are several interesting directions for future theoretical work. First, this paper has only derived bounds for single-edge modifications. It is an open problem to characterize the effects on polarization of more substantial perturbations to the graph structure. Furthermore, it may be possible to study the planner's effectiveness within various classes of random graph models. For instance, what fraction of non-edges in an Erd\H{o}s-R\'enyi graph reduce polarization when added?

\change{
At first glance, the results paper are severely limited by the model of opinion dynamics. Experimental research has shown that exposure to differing opinions may increase polarization~\citep{bail2018exposure}. Motivated by these observations, many models incorporate non-attractive forces between opinions -- see \citet{cornacchia2020polarization,rahaman2021model} for extensions of the FJ model, and \citet{hegselmann2002opinion,hazla2019geometric,gaitonde2021polarization} for geometrically-inspired approaches. Within the broader problem of understanding how social network structures relate to polarization, this paper provides only a first step -- the analytical tractability of the FJ model comes at the expense of expressibility. Nonetheless, we believe the results in this paper may be generalizable to a wider class of opinion dynamics models that exhibit attraction -- which includes all of the above examples. For instance, one could modify a `disagreement-seeking' heuristic to only consider non-saturated edges between individuals within each others' radius of attraction. The study of polarization-reducing strategies in these more complex models of opinion interaction is a rich and fruitful area for future work.
}

Several networks showed large reductions in polarization with \change{a small number of edge modifications}. However, in the Erd\H{o}s-R\'enyi and preferential attachment networks, this paper's heuristics did not have as strong of a performance. Beyond our speculation, it remains to be understood what properties of these networks may limit the planner's effectiveness, or what minimal budget is necessary for a fixed fractional reduction in polarization. Moreover, it is not yet clear if this observation is a feature of the heuristics or the graph itself -- are we closely approximating the true optimal solution?

In this study, we have shown that strengthening ties between disagreeing individuals is an effective strategy for reducing social polarization. Therefore, if polarization is instead increasing as society becomes increasingly connected, then both individuals and social media platforms may be failing to contribute to discourse between opposing perspectives.

\bibliography{refs.bib}

\newpage
\appendix

\section{Proofs}
\label{sec:proofs}

First, we specify notation. Let $I$ denote the identity matrix, $\vec{\mathbf{1}}$ the all-ones vector, and $\mathbf{e}_{i}$ the $i$-th standard basis vector -- all of appropriate dimension. Additionally, for $\mathbf{x}\in \mathbb{R}^{n}$, we write $\overline{x} := \frac{1}{n}\sum_{j=1}^n x_j$ to denote the mean of its entries and $\widetilde{\mathbf{x}} := \mathbf{x} - \overline{x} \vec{\mathbf{1}}$ for the mean-centered version of~$\mathbf{x}$. For a square matrix $A\in \mathbb{R}^{n \times n}$, we write $A_{i}$ for the $i$-th column of $A$. The eigenvalues of $A$ in descending order are given by $\lambda_n(A) \ge \lambda_{n-1}(A) \ge \ldots \ge \lambda_1(A)$). We frequently use the notation $\lambda_{\max}(A) = \lambda_{n}(A)$ and $\lambda_{\min}(A) = \lambda_{1}(A)$ to denote the largest and smallest eigenvalue of~$A$, respectively.

Given an initial graph $\mathcal{G}$ and any other graph $\mathcal{G}'$, define $T \equiv T(\mathcal{G}';\mathcal{G}) \in \mathbb{R}^{n \times n}$ to be

\begin{equation}
	\label{eq:T_def}
	T := (I+L)^{-1}(I+L'),
\end{equation}
where $L$ and $L'$ denote the combinatorial Laplacians of $\mathcal{G}$ and $\mathcal{G}'$, respectively. The dependence of~$T$ on $\mathcal{G}$ and $\mathcal{G}'$ will be clear from context and hence omitted. The expressed opinions $\mathbf{z}'$ can be computed in terms of $T$ and the original expressed opinions as follows:

\begin{equation}
    \mathbf{z}'  = T^{-1}\mathbf{z}.
\end{equation}
This matrix is also useful in allowing us to express the new value of polarization in terms of the expressed opinions on the initial graph. After some algebra, we have that

\begin{equation}
    P(\mathbf{z}') = \widetilde{\mathbf{z}}^{T}(T^{-1})^{T}T^{-1}\widetilde{\mathbf{z}},
\end{equation}
where we used~\eqref{eq:T_def}. The spectrum of $T$ will be critical for theoretical results.

Recall the definition of the isoperimetric number (also known as the Cheeger constant) of a graph from~\eqref{eq:isoperimetric}. The following simple Lemma is useful in many of the subsequent proofs.

\begin{lemma}
	\label{lem:L_eigenvalues}
	Let $d_{\max}$ and $d_{\min}$ denote the maximum and minimum weighted degrees of $\cG$. Additionally, let $L$ be the combinatorial Laplacian of $\cG$, and let $\lambda_{n} \geq \lambda_{n-1} \geq \ldots \geq \lambda_{2} \geq \lambda_{1} = 0$ denote its eigenvalues in decreasing order. 
    Then, we have that 
        
	\begin{equation}
        \label{eq:spectral_gap_bounds}
	    \frac{1}{2} d_{\min} h_{\cG}^{2} \le \lambda_{2} \le 2 d_{\max} h_{\cG},
	\end{equation}
	and also 
 
	\change{
    \begin{equation}
        \label{eq:largest_eigenvalue_bound}
        \lambda_{n} \le (2 d_{\max}) \wedge (\bar w n).
	\end{equation}
    }
\end{lemma}

\begin{proof}
	For the normalized Laplacian $\cL$, the well-known Cheeger inequality (see, e.g.,~\citet{chung1997spectral}) gives that
	
    \begin{equation}
	   \frac{h_{\cG}^{2}}{2} \le \lambda_{2}(\cL) \le 2h_{\cG}.
    \end{equation}
    
	Notice that the eigenvalues of $L = D^{1/2}\cL D^{1/2}$ are equal to those of $\cL D$. Additionally, the ordered eigenvalues of $D$ are simply the degrees of $\cG$ in descending order. Since both $\cL$ and $D$ are positive semidefinite and Hermitian, we can apply a  Weyl multiplicative inequality from~\citet{horn1994topics} to establish that 
 
	\begin{equation}
        \label{eq:weyl1}
        \lambda_{i+j-n}(\cL D) \le \lambda_i(\cL)\lambda_j(D), \qquad \text{ if } \quad  i+j-n\ge 1
	\end{equation}
	and
 
	\begin{equation}
        \label{eq:weyl2}
	    \lambda_i(\cL)\lambda_j(D) \le \lambda_{i+j-1}(\cL D), \qquad \text{ if } \quad  i+j-1\le n.
	\end{equation}
 
	Choosing $i = 2$ and $j=n$ in~\eqref{eq:weyl1} gives that 

    \begin{equation}
    	\lambda_{2}(\cL D) \le \lambda_{2}(\cL)d_{\max} \le 2h_{\cG}d_{\max}.
    \end{equation}
	With $i = 2$ and $j = 1$ in~\eqref{eq:weyl2}, we have that

    \begin{equation}
    	\lambda_{2}(\cL D) \ge \lambda_{2}(\cL)d_{\min} \ge \frac{1}{2}h_{\cG}^{2}d_{\min}.
    \end{equation}
	Combining the previous two displays gives~\eqref{eq:spectral_gap_bounds}.
	
	The inequality~\eqref{eq:largest_eigenvalue_bound} can be proved using the triangle inequality. The largest eigenvalue of $\cL$ is equal to the operator norm of $D-A$. Since the norm of both $D$ and $A$ are upper bounded by $d_{\max}$, we conclude that $\lambda_n(L) \le 2d_{\max}$.

    \change{Let $L_{K_n} = \bar w \left( nI - \vec{\mathbf{1}}\vec{\mathbf{1}}^T\right)$ denote the combinatorial Laplacian of the complete graph, where all edge weights are equal to $\bar w$. Since $L_{K_n} \succcurlyeq L$, for any $L$, we have that $\bar w n = \lambda_n(L_{K_n}) \ge \lambda_n(L)$.}
\end{proof}

\subsection{Proofs of Section~\ref{sec:contraction_and_pol}}

\begin{proof}[Proof of Proposition~\ref{prop:p_bounds}]
    We seek to write $P(\mathbf{z})$ in a way that $P(\mathbf{s})$ appears. Recall that $\mathbf{z} = (I+L)^{-1} \mathbf{s}$ and also $\widetilde{\mathbf{z}} = (I+L)^{-1} \widetilde{\mathbf{s}}$. Therefore 
    
    \begin{equation}
        \label{eq:PzPs}
        P(\mathbf{z}) = \widetilde{\mathbf{z}}^{T} \widetilde{\mathbf{z}} = \widetilde{\mathbf{s}}^{T}(I+L)^{-2}\widetilde{\mathbf{s}}. 
    \end{equation}
    Towards the lower bound in the claim, we may use an eigenvalue bound to obtain that 
    
    \begin{equation}
        P(\mathbf{z}) 
        \geq \lambda_{\min} ( (I+L)^{-2} ) \widetilde{\mathbf{s}}^{T} \widetilde{\mathbf{s}} 
        = (1 + \lambda_{\max} (L))^{-2} P(\mathbf{s}). 
    \end{equation}
    From~\eqref{eq:largest_eigenvalue_bound} we have that \change{$\lambda_{\max} (L) \leq (2 d_{\max}) \wedge  \bar w n$}; plugging this into the display above we obtain the claimed lower bound. 
    
    For the upper bound, first note that the eigenvector corresponding to the largest eigenvalue of $(I+L)^{-2}$ is $\vec{\mathbf{1}}$. Since $\widetilde{\mathbf{s}}$ is orthogonal to $\vec{\mathbf{1}}$, we have from~\eqref{eq:PzPs} that 

    \begin{equation}
        P(\mathbf{z}) 
        \leq \lambda_{n-1} ( (I+L)^{-2} ) \widetilde{\mathbf{s}}^{T} \widetilde{\mathbf{s}}
         = (1 + \lambda_{2} (L))^{-2} P(\mathbf{s}).
    \end{equation}
    From~\eqref{eq:spectral_gap_bounds} we have that $\lambda_{2}(L) \geq (1/2) d_{\min} h_{\cG}^{2}$; plugging this into the display above we obtain the desired upper bound. 
\end{proof}

\begin{proof}[Proof of Corollary~\ref{cor:global_min}]
    Take any graph $\cG$ and innate opinions $\mathbf{s}$. Proposition~\ref{prop:p_bounds} implies that 

    \begin{equation}
        P(\mathbf{z}_{\cG}) \ge P(\mathbf{s}) (1 + (2d_{\max}) \wedge (\bar w n))^{-2} \ge P(\mathbf{s}) (1 + \bar w n)^{-2}.
    \end{equation}
    
    Turning to the complete graph $K_{n}$, recall that the spectrum of its Laplacian has $0$ as an eigenvalue with eigenvector $\vec{\mathbf{1}}$. It also has eigenvalue \change{$\bar w n$} with multiplicity $n-1$ and eigenspace containing all vectors orthogonal to $\vec{\mathbf{1}}$. Since $\widetilde{\mathbf{s}}^T\vec{\mathbf{1}} = 0$, we have \change{$(I+L_{K_n})^{-1}\widetilde{\mathbf{s}} = (1+\bar w n)^{-1} \widetilde{\mathbf{s}}$}. Recalling the definition of polarization, we obtain \change{$P(\mathbf{z}_{K_n}) = \norm{(I+L_{K_n})^{-1}\widetilde{\mathbf{s}}}^{2} 
    = (1+\bar w n)^{-2} \norm{\widetilde{\mathbf{s}}}^{2} 
    = (1+\bar w n)^{-2} P(\mathbf{s})$}. 
    Comparing with the display above, we see that $K_{n}$ minimizes polarization over all graphs \change{with maximal weight $\bar w$}. 
\end{proof}

\subsection{Proofs of Section~\ref{sec:given_ops}}

\begin{proof}[Proof of Lemma~\ref{lem:p_diff_exact}]
    To obtain the claim, we expand $P(\mathbf{z}^+)$ in a way that $P(\mathbf{z})$ appears. First, note that \change{$L^{+} = L + \delta L_{ij}$ and $L_{ij}= \mathbf{v}_{ij}\mathbf{v}_{ij}^{T}$, and hence we have that $T = I + \delta(I+L)^{-1}\mathbf{v}_{ij}\mathbf{v}_{ij}^T$}. The Sherman-Morrison formula thus gives that $T^{-1} = I - \frac{\delta}{1+\delta\mathbf{v}_{ij}^T(I+L)^{-1}\mathbf{v}_{ij}} (I+L)^{-1}\mathbf{v}_{ij}\mathbf{v}_{ij}^T$. Plugging this into the formula for polarization, we obtain that 
    
    \change{
    \begin{align}
        P(\mathbf{z}^+) &= \norm{T^{-1}\widetilde{\mathbf{z}}}^2 
        = \norm{\left( I - \frac{\delta}{1+\delta\mathbf{v}_{ij}^T(I+L)^{-1}\mathbf{v}_{ij}} (I+L)^{-1}\mathbf{v}_{ij}\mathbf{v}_{ij}^T \right)\widetilde{\mathbf{z}}}^2\\
        &= \widetilde{\mathbf{z}}^T\widetilde{\mathbf{z}} -  \frac{2 \delta \widetilde{\mathbf{z}}^T(I+L)^{-1}\mathbf{v}_{ij}\mathbf{v}_{ij}^T \widetilde{\mathbf{z}} }{1+\delta \mathbf{v}_{ij}^T(I+L)^{-1}\mathbf{v}_{ij}} + (\widetilde{\mathbf{z}}^T\mathbf{v}_{ij})^2 \frac{\delta^2 \mathbf{v}_{ij}^T(I+L)^{-2}\mathbf{v}_{ij}}{\left(1+\delta \mathbf{v}_{ij}^T(I+L)^{-1}\mathbf{v}_{ij}\right)^2}.
    \end{align} 
    }
    Noting that $P(\mathbf{z}) = \widetilde{\mathbf{z}}^T\widetilde{\mathbf{z}}$, and $D_{ij}(\mathbf{z}) = (\widetilde{\mathbf{z}}^T\mathbf{v}_{ij})^2$ leads to the desired expression after rearranging.
\end{proof}

\begin{proof}[Proof of Proposition~\ref{prop:dP_dw}]
    For simplicity of notation, let $A = I+L$. Then, for any $t>0$, we have
    
    \begin{equation}
        \label{eq:pol_deriv_fixed_t}
        \frac{P(\mathbf{z}_{L+tL_{ij}}) - P(\mathbf{z}_L)}{t}  = \frac{\widetilde{\mathbf{s}}^T\left[ \left(A+t\mathbf{v}_{ij}\mathbf{v}_{ij}^T\right)^{-2} - A^{-2}\right]\widetilde{\mathbf{s}}}{t}  
    \end{equation}
    
    Using the Sherman-Morrison formula, we can compute that
    
    \begin{equation}
        \label{eq:square_inv_sherman_morrison}
        \begin{split}
            \left[\left(A+t\mathbf{v}_{ij}\mathbf{v}_{ij}^T\right)^{-1}\right]^2 &= \left[A^{-1} - \frac{tA^{-1}\mathbf{v}_{ij}\mathbf{v}_{ij}^TA^{-1}}{1+t\mathbf{v}_{ij}^T A^{-1}\mathbf{v}_{ij}}\right]^2 \\
            &= A^{-2} - 2t\frac{A^{-2}\mathbf{v}_{ij}\mathbf{v}_{ij}^TA^{-1}}{1+t\mathbf{v}_{ij}^T A^{-1}\mathbf{v}_{ij}}
            + o(t)
        \end{split}
    \end{equation}
    where $\frac{o(t)}{t} = o(1) \to_{t\to 0} 0$. Plugging~\eqref{eq:square_inv_sherman_morrison} into~\eqref{eq:pol_deriv_def} and taking the limit concludes.
\end{proof}

\begin{proof}[Proof of Corollary~\ref{cor:p_diff_special}]
    Recall that $\mathbf{v}_{ij} = \mathbf{e}_i - \mathbf{e}_j$. Since $N(i) = N(j)$, a direct computation gives \change{$L \mathbf{v}_{ij} = (d_{i} - w_{ij}) \mathbf{v}_{ij}$. Consequently, we have $(I+L)^{-1}\mathbf{v}_{ij} = \frac{1}{1 + d_i - w_{ij}} \mathbf{v}_{ij}$}. Plugging this into Lemma~\ref{lem:p_diff_exact} and simplifying yields the desired result. 
\end{proof}

\begin{proof}[Proof of Theorem~\ref{thm:p_diff_bounds}]
    The proof of this Theorem follows from bounding the terms in Lemma~\ref{lem:p_diff_exact}.
    
    First, we show the upper bound. Notice that \change{$\frac{\delta^2 \mathbf{v}_{ij}^T(I+L)^{-2}\mathbf{v}_{ij}}{1+\delta \mathbf{v}_{ij}^T(I+L)^{-1}\mathbf{v}_{ij}} \ge 0$,} so this term can be dropped. Through an eigenvalue bound we also find that 
    
    \change{
    \begin{equation}
        \frac{1}{1+\delta \mathbf{v}_{ij}^T(I+L)^{-1}\mathbf{v}_{ij}} \le \frac{1}{1+ 2\delta \lambda_{\min}((I+L)^{-1})} = \frac{1+\lambda_n(L)}{1+2\delta+\lambda_n(L)}.
    \end{equation}
    }
    Plugging these two observations into~\eqref{eq:p_diff_exact} and rearranging to find $\partial_{w_{ij}}P(L)$ concludes.
    
    For the lower bound, we have the following sequence of inequalities.
    
    \change{
    \begin{equation}
        \frac{\delta \mathbf{v}_{ij}^T(I+L)^{-2}\mathbf{v}_{ij}}{1+\delta \mathbf{v}_{ij}^T(I+L)^{-1}\mathbf{v}_{ij}} \le \frac{\delta \mathbf{v}_{ij}^T(I+L)^{-2}\mathbf{v}_{ij}}{1+\delta \mathbf{v}_{ij}^T(I+L)^{-2}\mathbf{v}_{ij}} \le \frac{2\delta }{2\delta  + (1+\lambda_{2}(L))^2}.
    \end{equation}
    }
    Therefore, by assumption and Lemma~\ref{lem:p_diff_exact}, we have that
    
    \change{
    \begin{align}
        P(\mathbf{z}) - P(\mathbf{z}^+) & \le \frac{\delta (z_i - z_j)^2}{1+\delta \mathbf{v}_{ij}^T(I+L)^{-1}\mathbf{v}_{ij}} [2\epsilon] \le \frac{2 \delta \epsilon (z_i - z_j)^2}{1+2\delta},
    \end{align}
    }
    where we used $\mathbf{v}_{ij}^T(I+L)^{-1}\mathbf{v}_{ij} \le 2$.
\end{proof}

\subsection{Proofs of Section~\ref{sec:adv_ops}}

\begin{proof}[Proof of Proposition~\ref{prop:robust_pol_min}]
    This proof requires only that we solve explicitly the adversary's optimization problem. 
    
    By construction, $\widetilde{\mathbf{s}}$ is orthogonal to $\vec{\mathbf{1}}$. As a result, the optimal solution for the adversary is $\sqrt{R}\mathbf{v}_2$, where $\mathbf{v}_2$ is the eigenvector corresponding to the spectral gap of $L'$, and their optimal value is:

    \begin{equation}
        \max_{\mathbf{s} \in \mathbb{R}^n: \norm{\mathbf{s}}_2^2 \le R} \ \mathbf{\widetilde{s}}^T\left( I+L'\right)^{-2}\mathbf{\widetilde{s}} = \frac{R}{\left(1+\lambda_2(L')\right)^2}.
    \end{equation}
    To minimize this quantity, it follows that the planner maximizes the spectral gap of~$L'$.
\end{proof}

\begin{proof}[Proof of Theorem~\ref{thm:spectral_gap_bds}]
    This proof uses a variation of a result by~\citet{Maas1987}. \change{The original result states that if a simple (unweighted, undirected) graph~$\cG^+_{\text{s}}$ is constructed by adding a non-edge $(i,j)$ to another simple graph $\cG_{\text{s}}$, we have
    
    \begin{equation}
        \min\left\{ \lambda_{2}(L_{\text{s}}) + \frac{\epsilon \alpha^2}{\epsilon + (2-\alpha^2)}, \lambda_{3}(L_{\text{s}}) - \epsilon \right\} \\
        \le \lambda_{2}(L^+_{\text{s}}) \le \min \{ \lambda_{2}(L_{\text{s}}) + \alpha^2,
        \lambda_{3}(L_{\text{s}}) \},
    \end{equation}
    }
    where $\alpha^2 = (v_i - v_j)^2$, and $\mathbf{v}$ is the eigenvector of $L$ corresponding to $\lambda_{2}(L)$.
    
    \change{
    It is possible to adapt the original proof to consider the case where weight $\delta$ is added to edge $(i,j)$. This result would yield:

    \begin{equation}
        \min\left\{ \lambda_{2}(L) + \frac{\epsilon \delta \alpha^2}{\epsilon + \delta (2-\alpha^2)}, \lambda_{3}(L) - \epsilon \right\}
        \le \lambda_{2}(L^+) \le \min \{ \lambda_{2}(L) + \delta\alpha^2, \lambda_{3}(L) \},
    \end{equation}
    The tightest lower bound is achieved by choosing

    \begin{equation}
        \epsilon^* = \frac{\beta - 2\delta}{2} + \left(  \left( \frac{\beta - 2\delta}{2}\right)^2 + \beta\delta(2-\alpha^2)\right)^{1/2},
    \end{equation}
    where $\beta = \lambda_{3}(L) - \lambda_{2}(L)$, so that both terms in the minimum are equal. First, we note that $\epsilon^* \ge \beta - 2\delta$, with equality when $\alpha = 2$. Additionally, the first term in the minimum is increasing in $\epsilon$, so therefore we have

    \begin{equation}
        \frac{\beta-2\delta}{\beta}\delta \alpha^2  \le \frac{(\beta-2\delta) \delta \alpha^2}{\beta-2\delta + \delta(2-\alpha^2)} \le \lambda_{2}(L^+) - \lambda_{2}(L) \le \alpha^2,
    \end{equation}
    as claimed since $\alpha^2 \ge 0$. Finally, note that $\lambda_2(L^+) \ge \lambda_2(L)$ as \change{$L^+ - L = \delta L_{ij} \succcurlyeq 0$}.
    }
\end{proof}

\begin{proof}[Proof of Corollary~\ref{cor:spectral_gap_red_bds}]
    Recall that we defined $P(L) = \frac{R}{\left(1+\lambda_2(L)\right)^2}$. We first prove the upper bound by using~\eqref{eq:spectral_gap_bounds}:
    
    \change{
    \begin{equation}
        \begin{split}
            \frac{1}{\left(1+\lambda_2(L)\right)^2} - \frac{1}{\left(1+\lambda_2(L^+)\right)^2} &\le \frac{1}{\left(1+\lambda_2(L)\right)^2} - \frac{1}{\left(1+\lambda_2(L)+\delta \alpha^2\right)^2} \\
            &\le \frac{\delta^2\alpha^4 + 2(1+\lambda_2(L))\delta \alpha^2 }{\left(1+\lambda_2(L)\right)^2\left(1+\lambda_2(L)+\delta\alpha^2\right)^2}.
        \end{split}
    \end{equation}
    Since $\delta\alpha^2 \ge 0$ and $\alpha^2 \le 2 \le 2(1+\lambda_2(L))$, we write $\alpha^4 \le 2(1+\lambda_2(L))\alpha^2$, and arrive at
    
    \begin{equation}
        \frac{1}{\left(1+\lambda_2(L)\right)^2} - \frac{1}{\left(1+\lambda_2(L^+)\right)^2} \le \frac{4\alpha^2 \left(\delta \vee \delta^2\right)}{\left(1+\lambda_2(L)\right)^3}
    \end{equation}
    as claimed. 
    
    The lower bound follows similarly by~\eqref{eq:spectral_gap_bounds}. For simplicity of notation, let $c = \max \left\{1 - \frac{2\delta}{\beta}, 0 \right\}$. Then,
    
    \begin{equation}
        \begin{split}
        \frac{1}{\left(1+\lambda_2(L)\right)^2} - \frac{1}{\left(1+\lambda_2(L^+)\right)^2} &\ge \frac{1}{\left(1+\lambda_2(L)\right)^2} - \frac{1}{\left(1+\lambda_2(L)+c\delta \alpha^2\right)^2} \\
        &\ge \frac{c^2\delta^2\alpha^4 + 2c\delta\alpha^2(1+\lambda_2(L)) }{\left(1+\lambda_2(L)\right)^2\left(1+\lambda_2(L)+c\delta\alpha^2\right)^2}.
        \end{split}
    \end{equation}
    
    Observe that $c^2\delta^2\alpha^4 \ge 0$, so this term can be dropped. Furthermore, $c\delta\alpha^2 \le 2\delta$, which gives us:
    
    \begin{equation}
        \frac{1}{\left(1+\lambda_2(L)\right)^2} - \frac{1}{\left(1+\lambda_2(L^+)\right)^2} \ge \frac{2 c\delta\alpha^2}{\left(1+2\delta+\lambda_2(L)\right)^3},
    \end{equation}
    as desired.
    }
\end{proof}

\newpage
\section{Additional Figures}
\label{app:figures}

\change{
In this short section, we present Figures showing how homophily (i.e., assortativity of innate opinions) and the spectral gap are affected by the planner's modifications. These provide greater detail than the initial and final values found in Table~\ref{tab:comp_results}.
}

\begin{figure}[ht]
    \centering
    \captionsetup{justification=centering}
    \begin{subfigure}{0.35\linewidth}
    	\includegraphics[width=\linewidth]{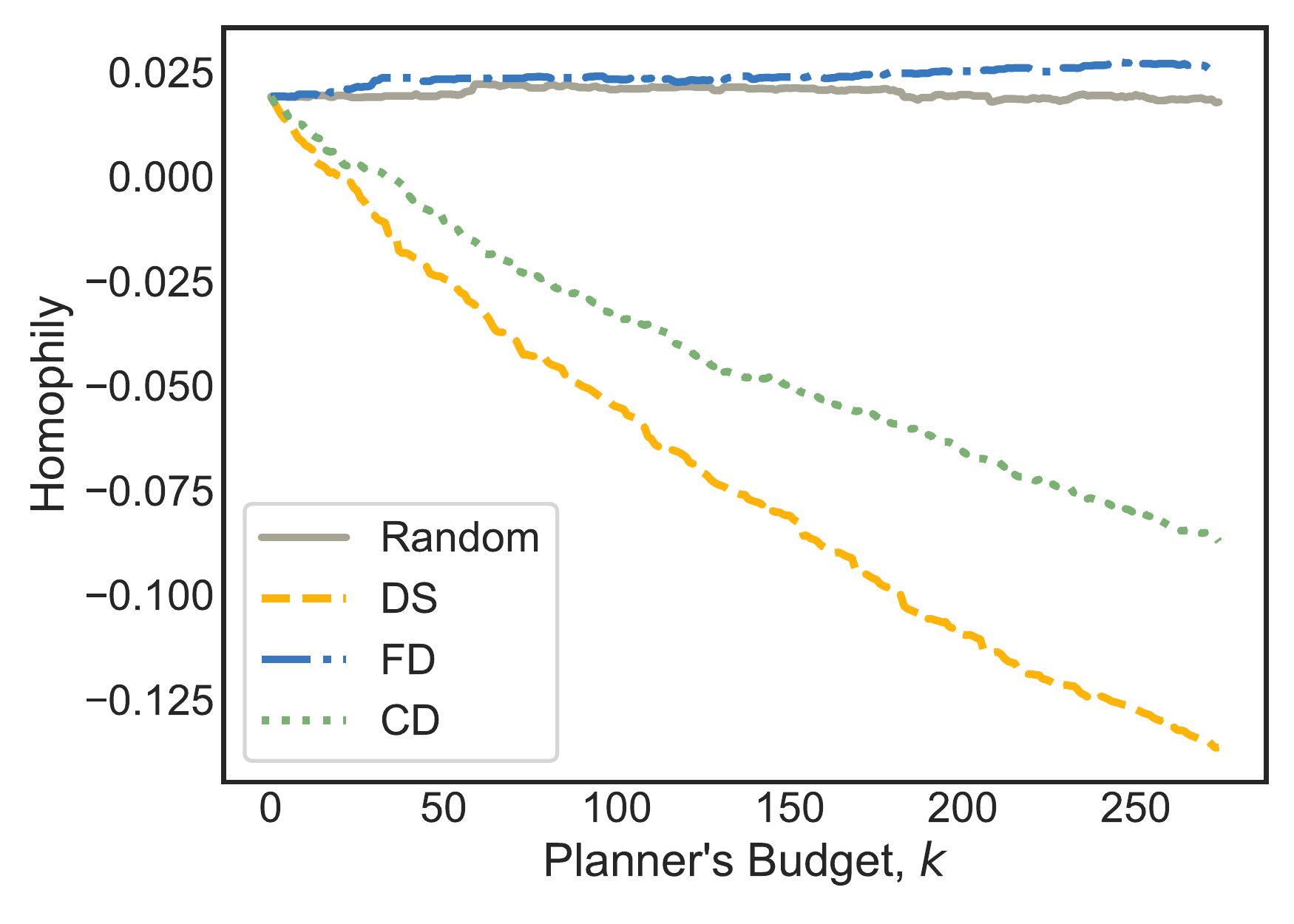}
    	\caption{Assortativity of innate opinions}
    	\label{fig:homophily_tw}	
    \end{subfigure} 
    \begin{subfigure}{0.35\linewidth}
        \centering
    	\includegraphics[width=\linewidth]{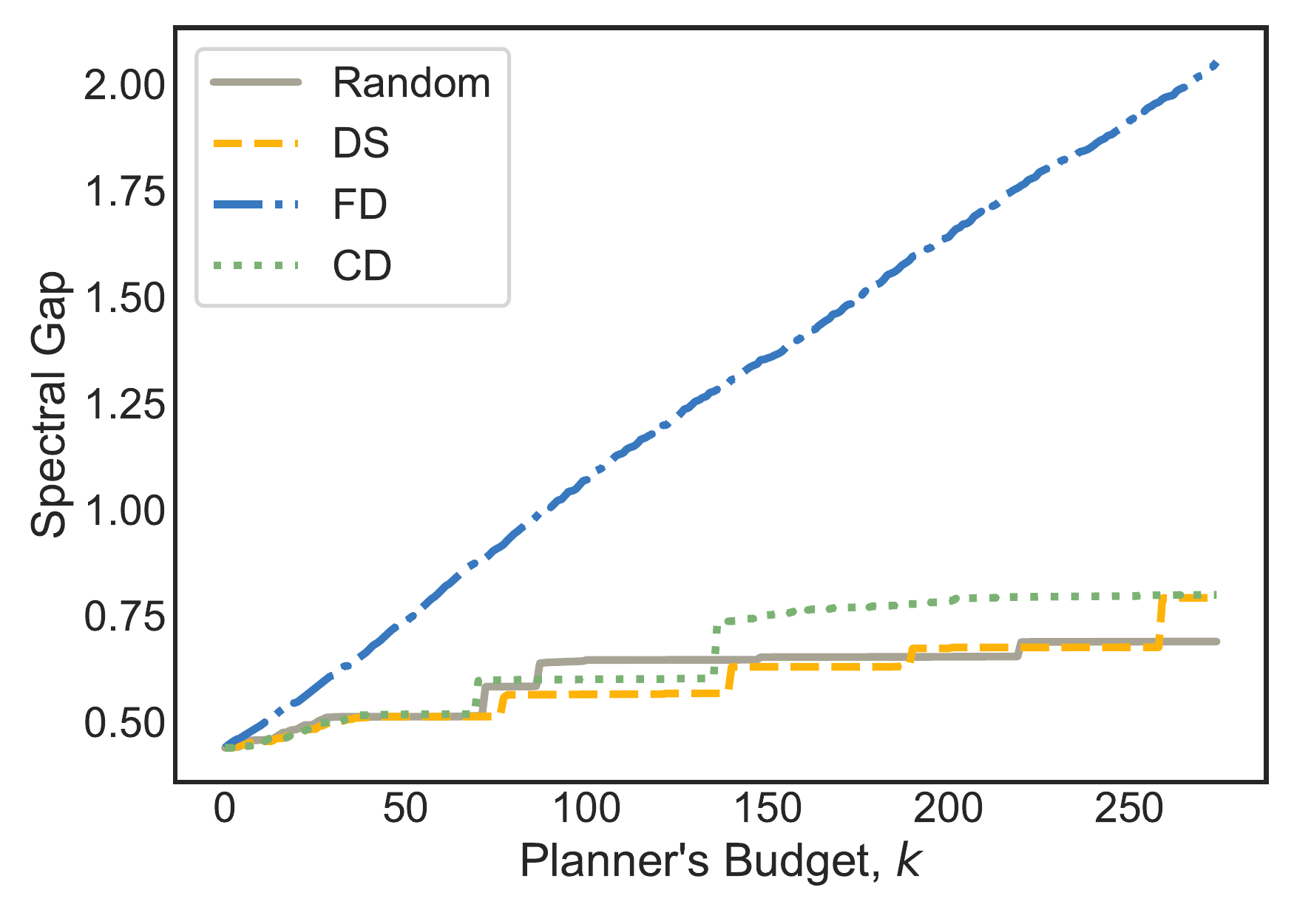}
    	\caption{Spectral gap} \label{fig:s_gap_tw}	
    \end{subfigure}
    \caption{Impacts of the planner's budget on the Twitter network.}
\end{figure}

\begin{figure}[ht]
    \centering
    \captionsetup{justification=centering}
    \begin{subfigure}{0.35\linewidth}
    	\includegraphics[width=\linewidth]{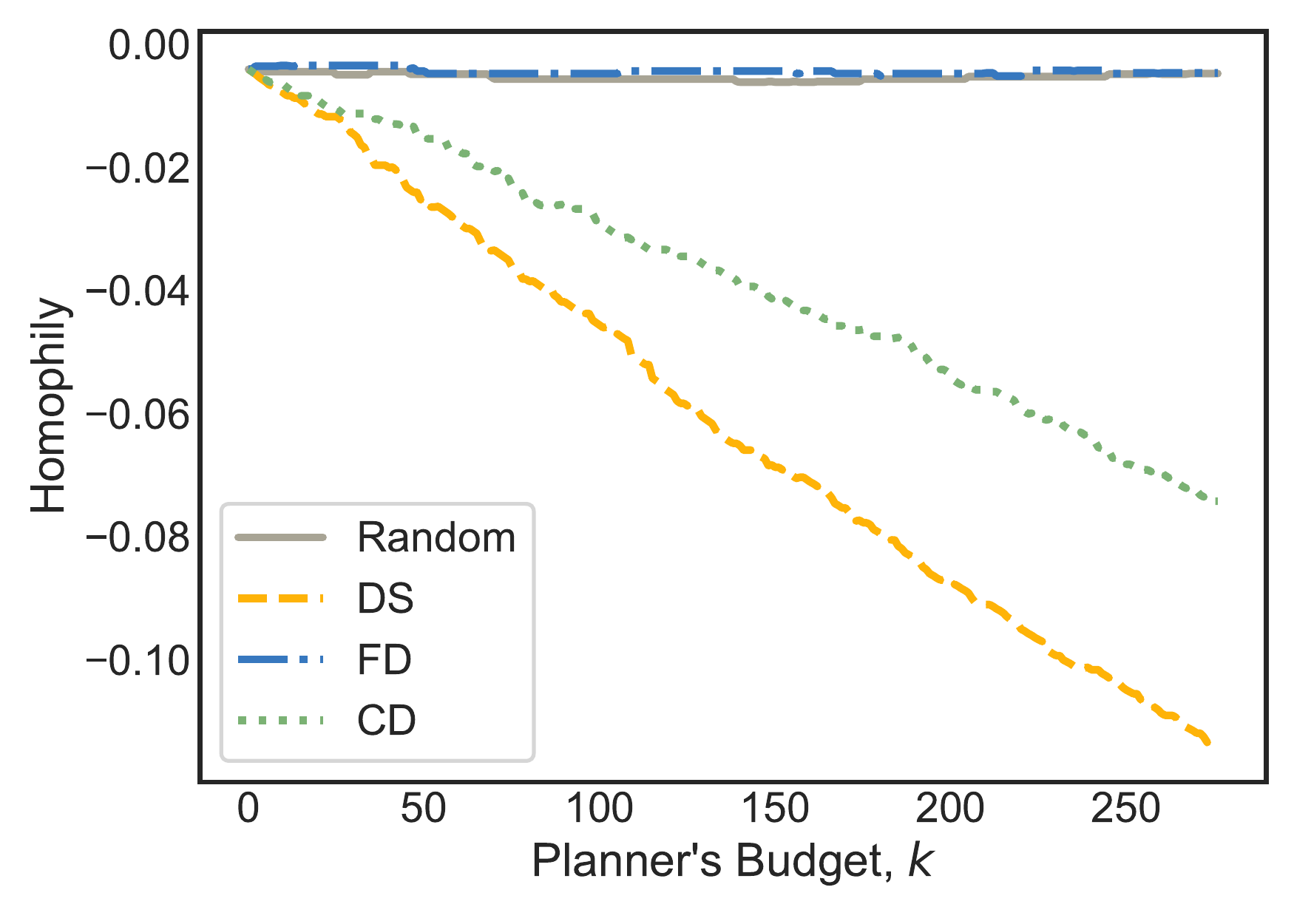}
    	\caption{Assortativity of innate opinions}
    	\label{fig:homophily_rd}	
    \end{subfigure} 
    \begin{subfigure}{0.35\linewidth}
        \centering
    	\includegraphics[width=\linewidth]{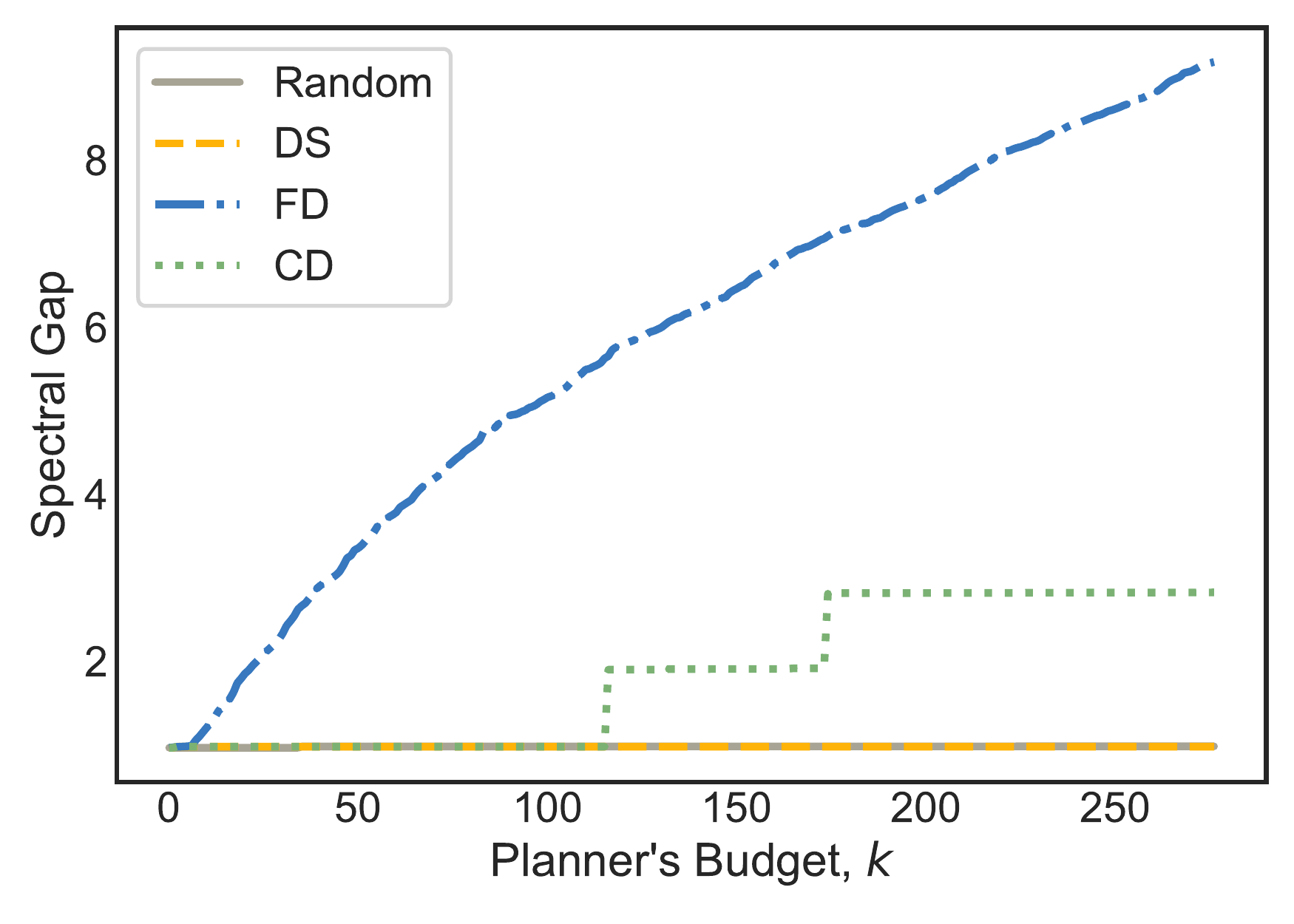}
    	\caption{Spectral gap} \label{fig:s_gap_rd}	
    \end{subfigure}
    \caption{Impacts of the planner's budget on the Reddit network.}
\end{figure}

\begin{figure}[ht]
    \centering
    \captionsetup{justification=centering}
    \begin{subfigure}{0.35\linewidth}
    	\includegraphics[width=\linewidth]{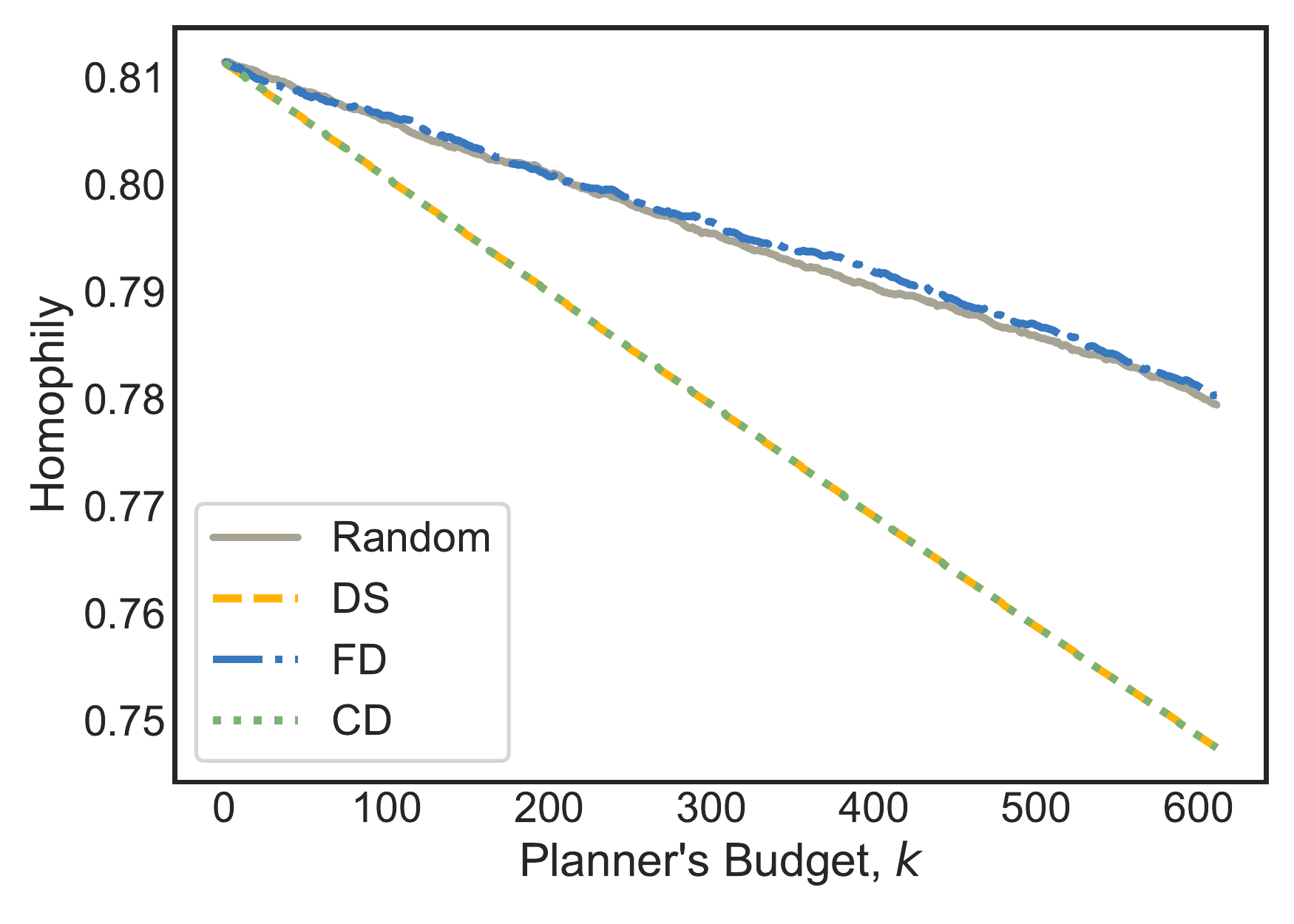}
    	\caption{Assortativity of innate opinions}
    	\label{fig:homophily_bg}	
    \end{subfigure} 
    \begin{subfigure}{0.35\linewidth}
        \centering
    	\includegraphics[width=\linewidth]{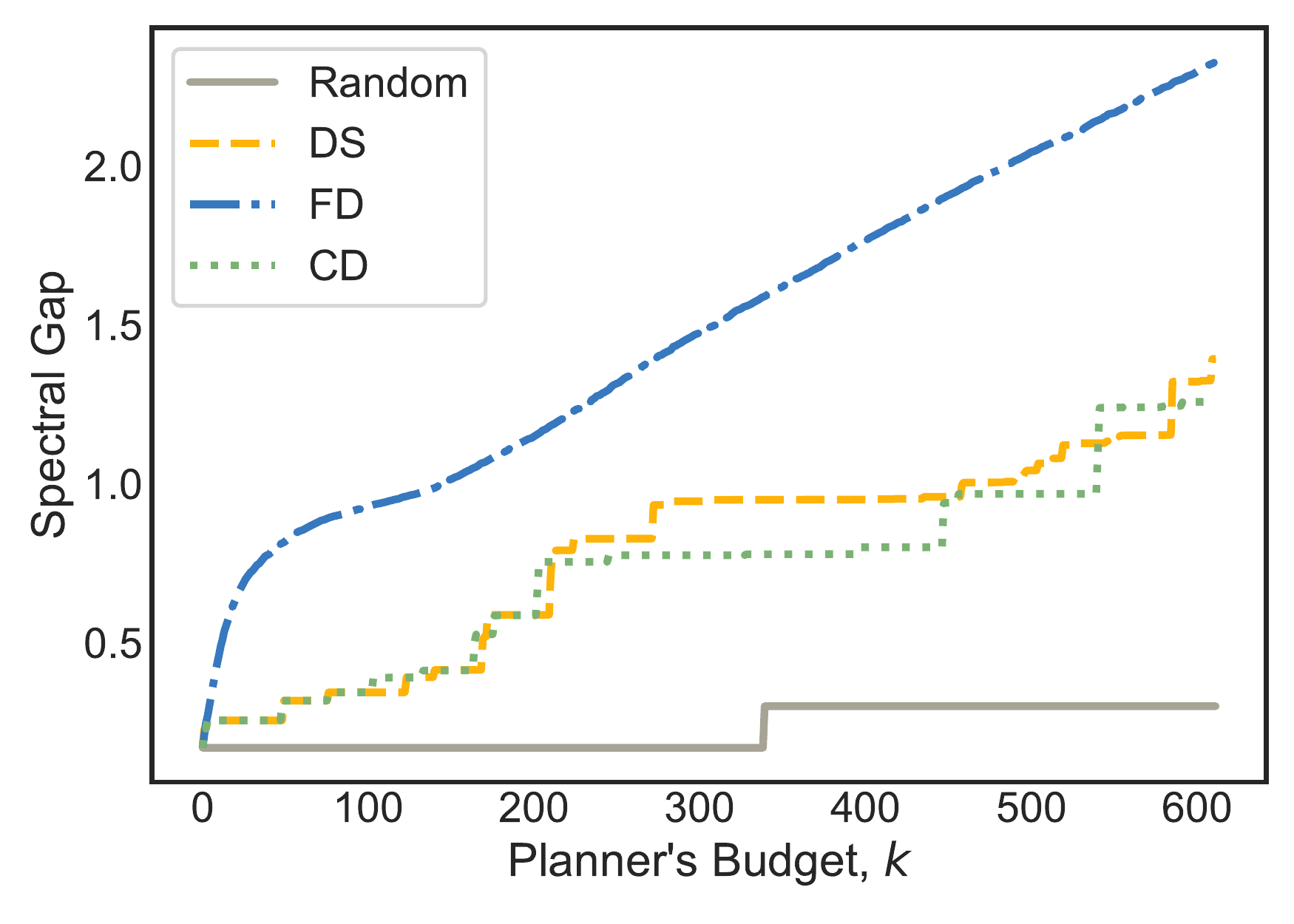}
    	\caption{Spectral gap} \label{fig:s_gap_bg}	
    \end{subfigure}
    \caption{Impacts of the planner's budget on the political blogs network.}
\end{figure}

\begin{figure}[ht]
    \centering
    \captionsetup{justification=centering}
    \begin{subfigure}{0.35\linewidth}
    	\includegraphics[width=\linewidth]{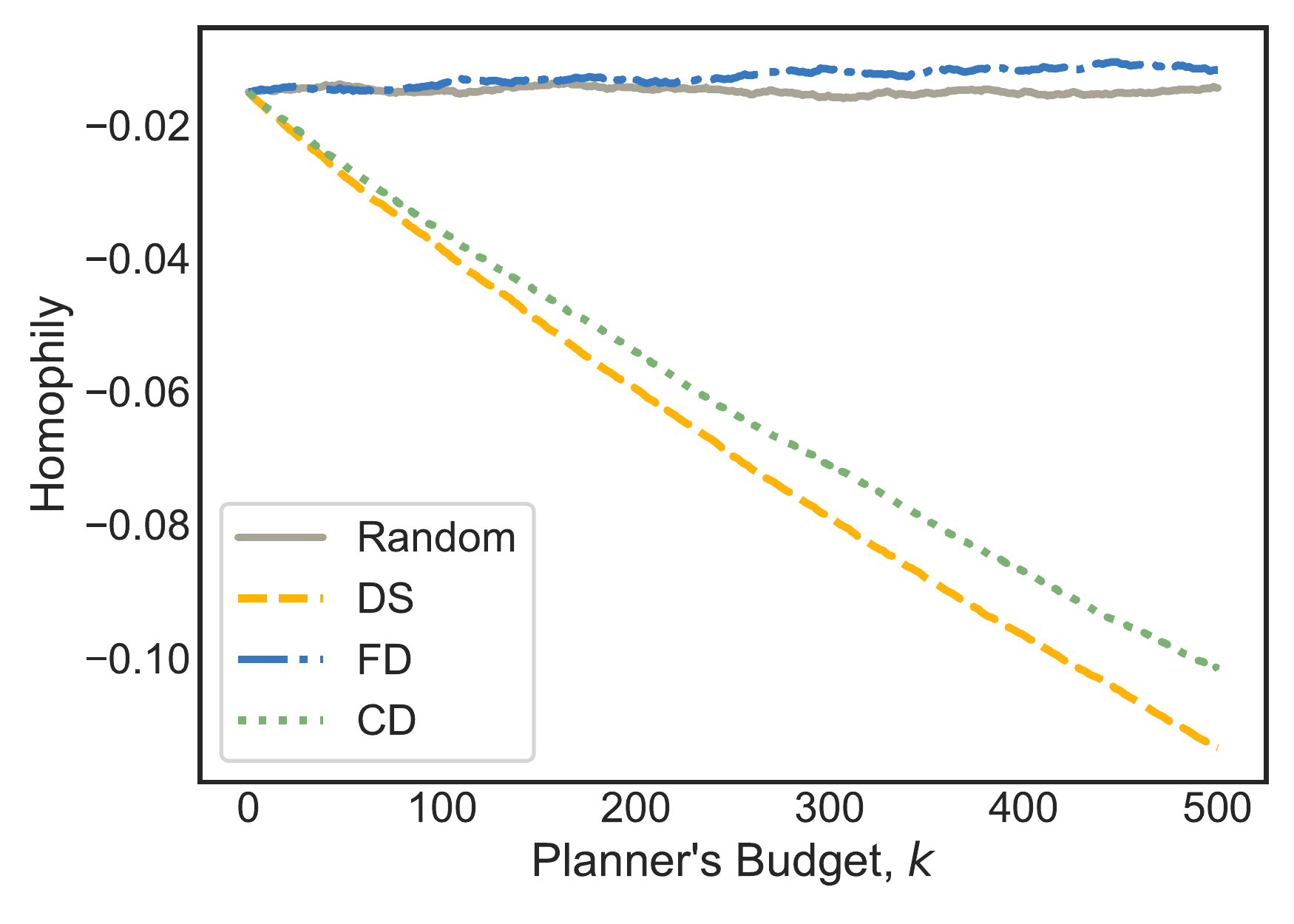}
    	\caption{Assortativity of innate opinions}
    	\label{fig:homophily_er}	
    \end{subfigure} 
    \begin{subfigure}{0.35\linewidth}
        \centering
    	\includegraphics[width=\linewidth]{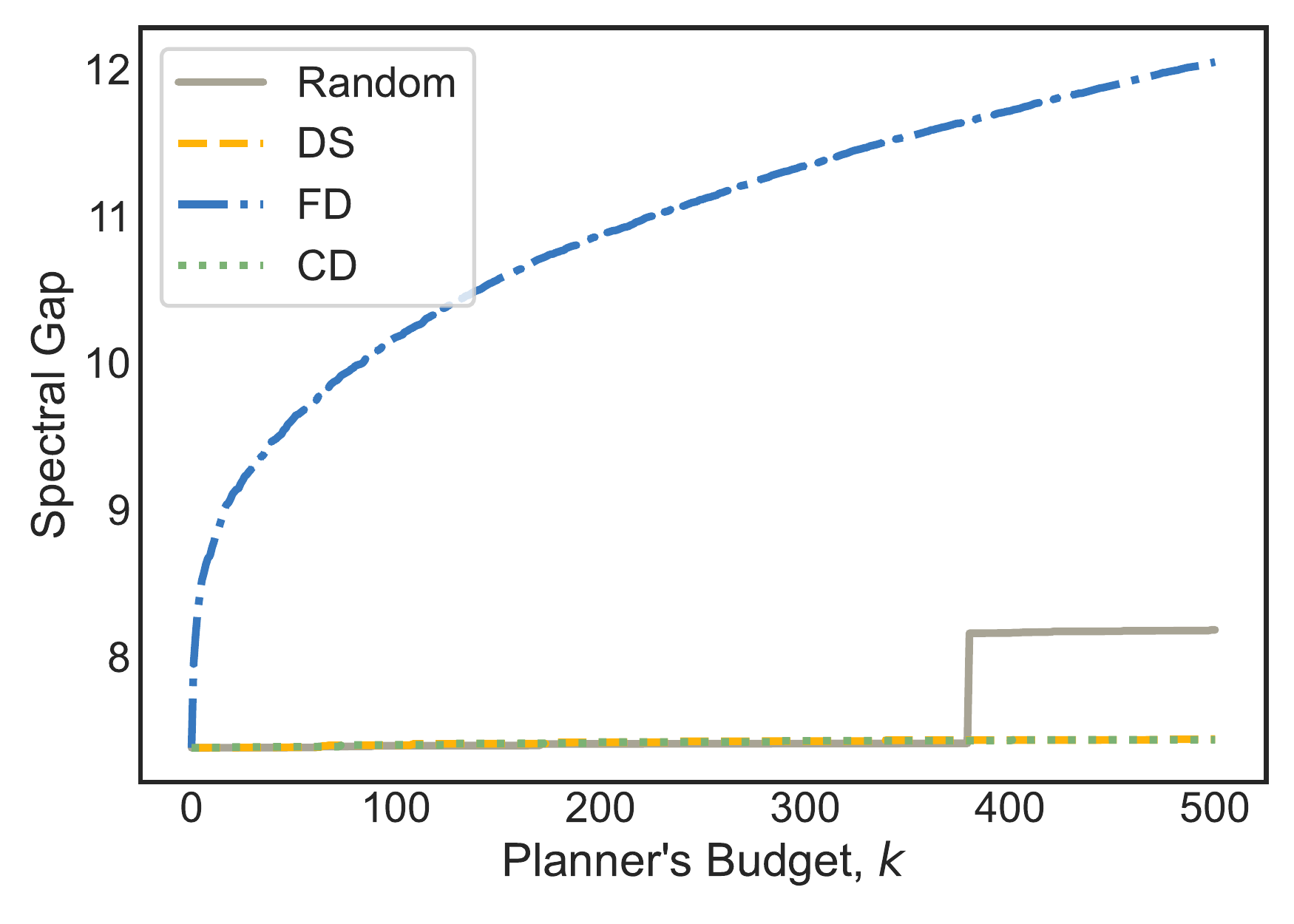}
    	\caption{Spectral gap} \label{fig:s_gap_er}	
    \end{subfigure}
    \caption{Impacts of the planner's budget on the Erd\H{o}s-R\'enyi graph.}
\end{figure}

\begin{figure}[ht]
    \centering
    \captionsetup{justification=centering}
    \begin{subfigure}{0.35\linewidth}
    	\includegraphics[width=\linewidth]{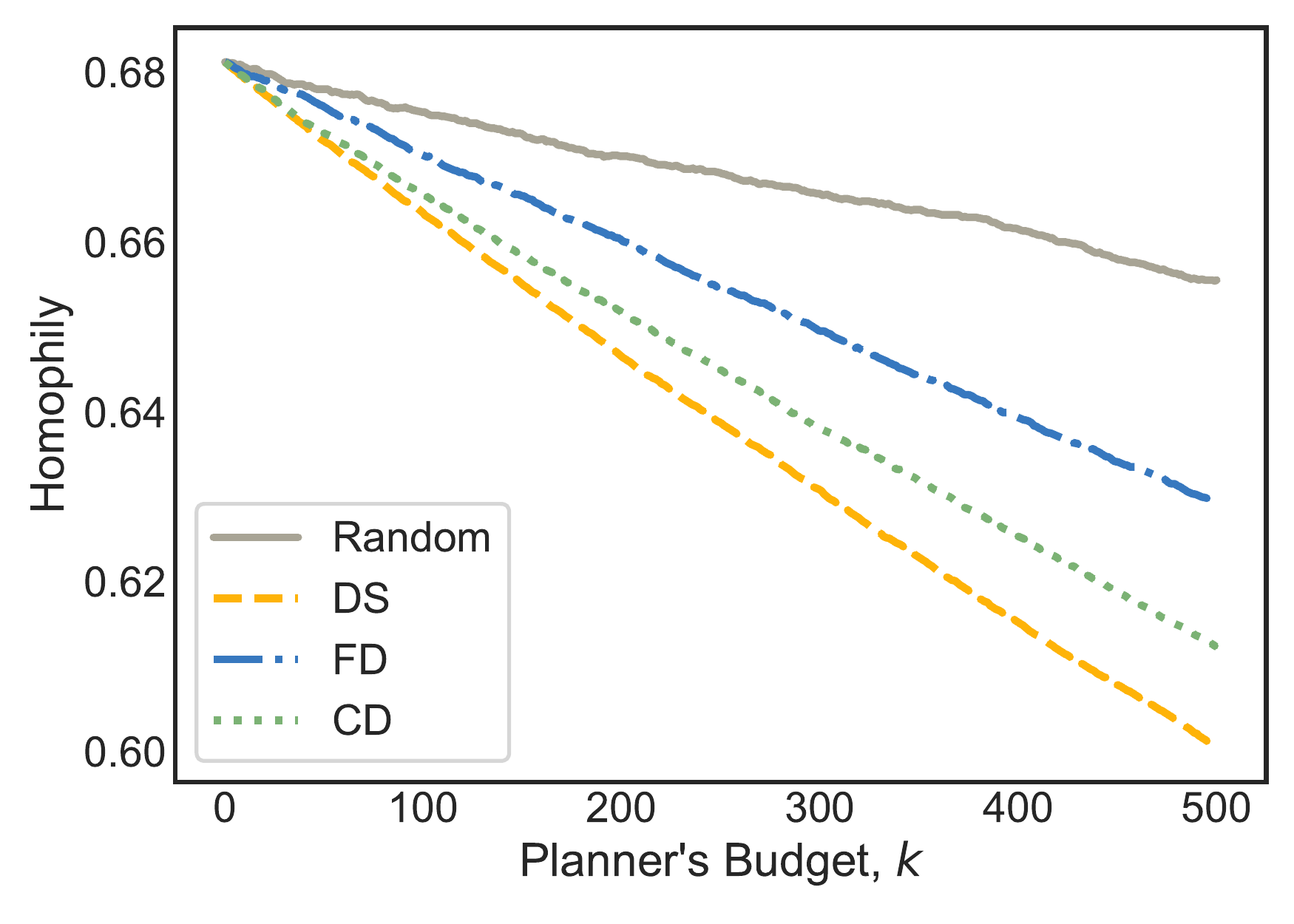}
    	\caption{Assortativity of innate opinions}
    	\label{fig:homophily_sbm}	
    \end{subfigure} 
    \begin{subfigure}{0.35\linewidth}
        \centering
    	\includegraphics[width=\linewidth]{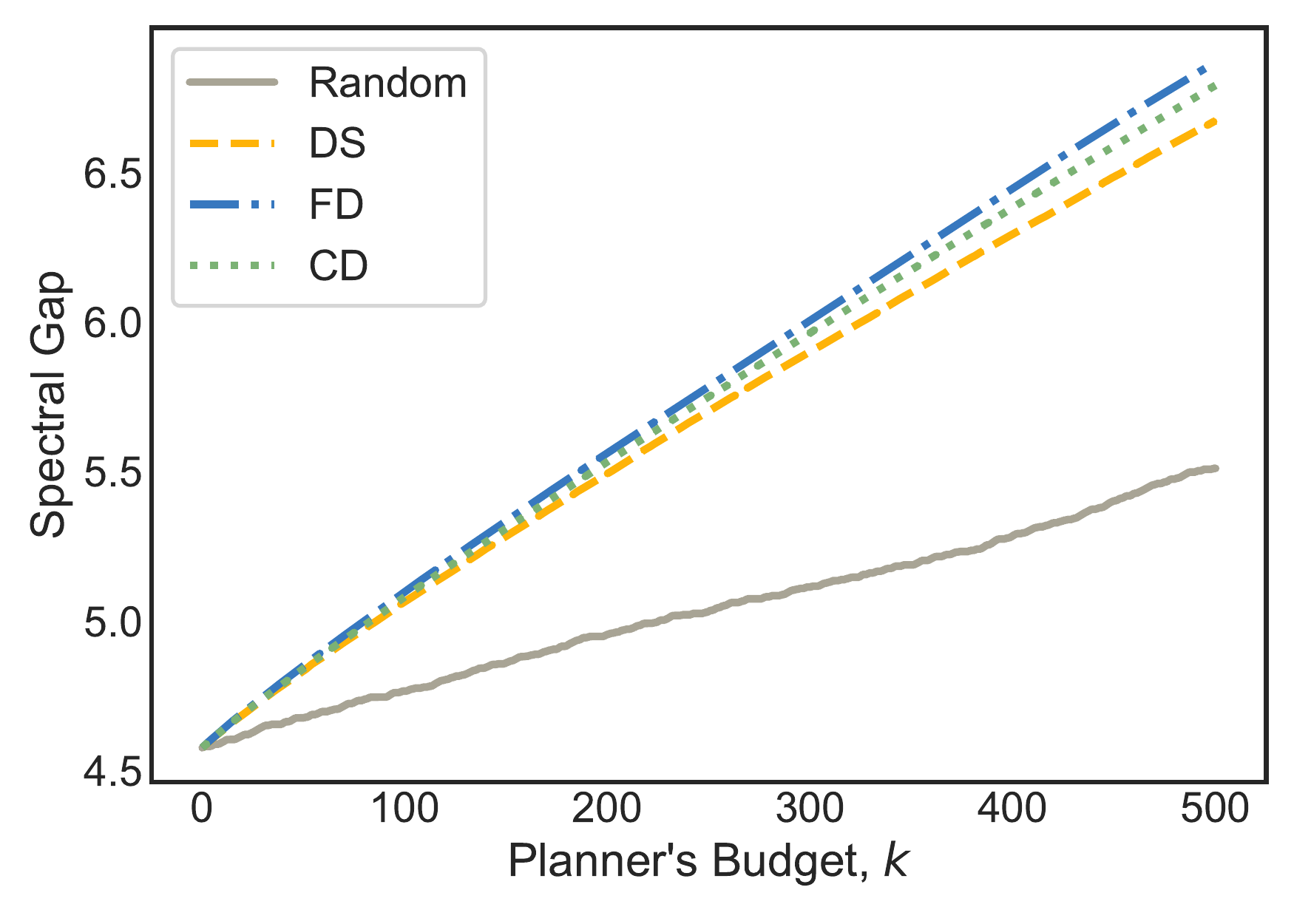}
    	\caption{Spectral gap} \label{fig:s_gap_sbm}	
    \end{subfigure}
    \caption{Impacts of the planner's budget on the stochastic block model graph.}
\end{figure}

\begin{figure}[ht]
    \centering
    \captionsetup{justification=centering}
    \begin{subfigure}{0.35\linewidth}
    	\includegraphics[width=\linewidth]{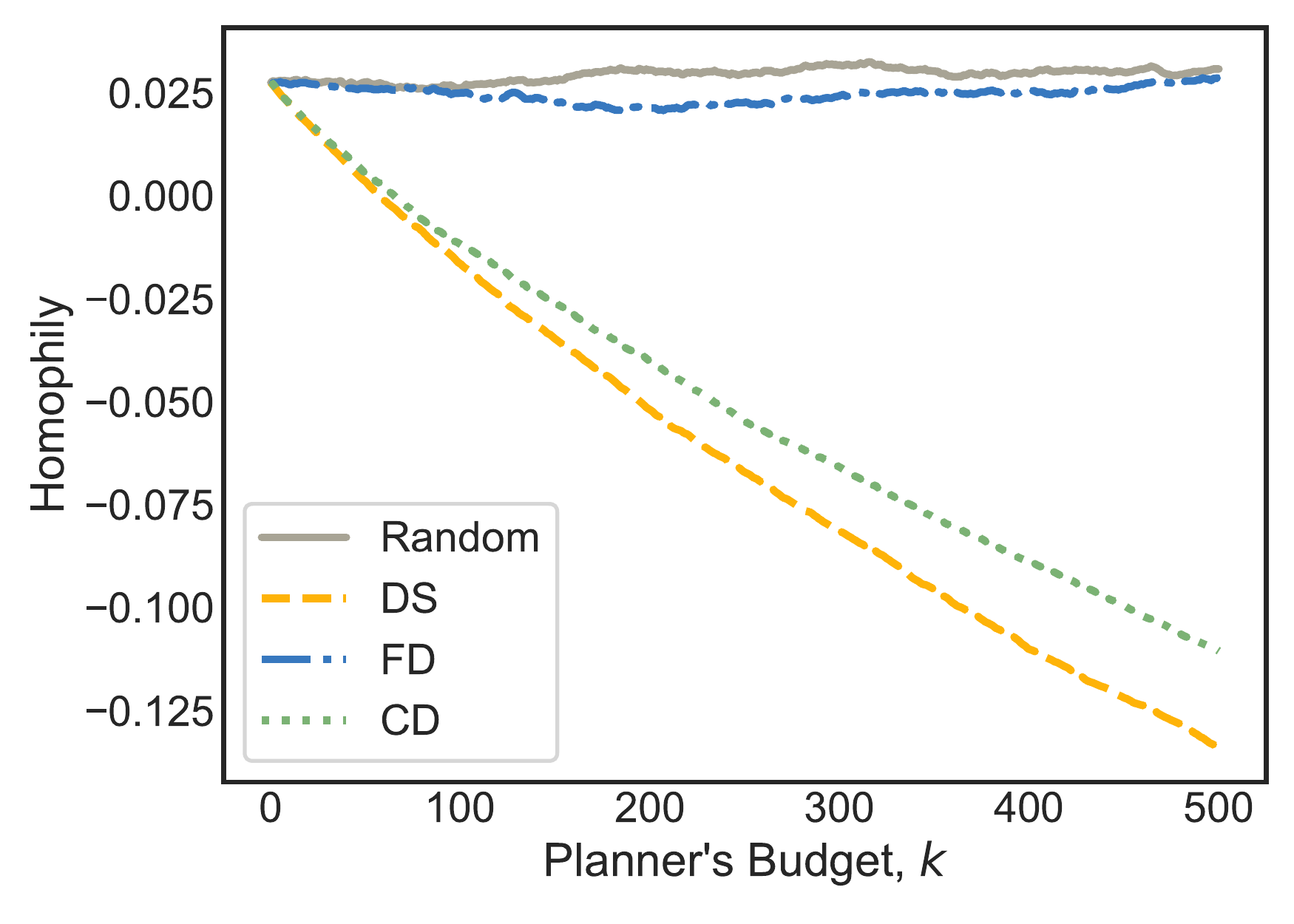}
    	\caption{Assortativity of innate opinions}
    	\label{fig:homophily_pa}	
    \end{subfigure} 
    \begin{subfigure}{0.35\linewidth}
        \centering
    	\includegraphics[width=\linewidth]{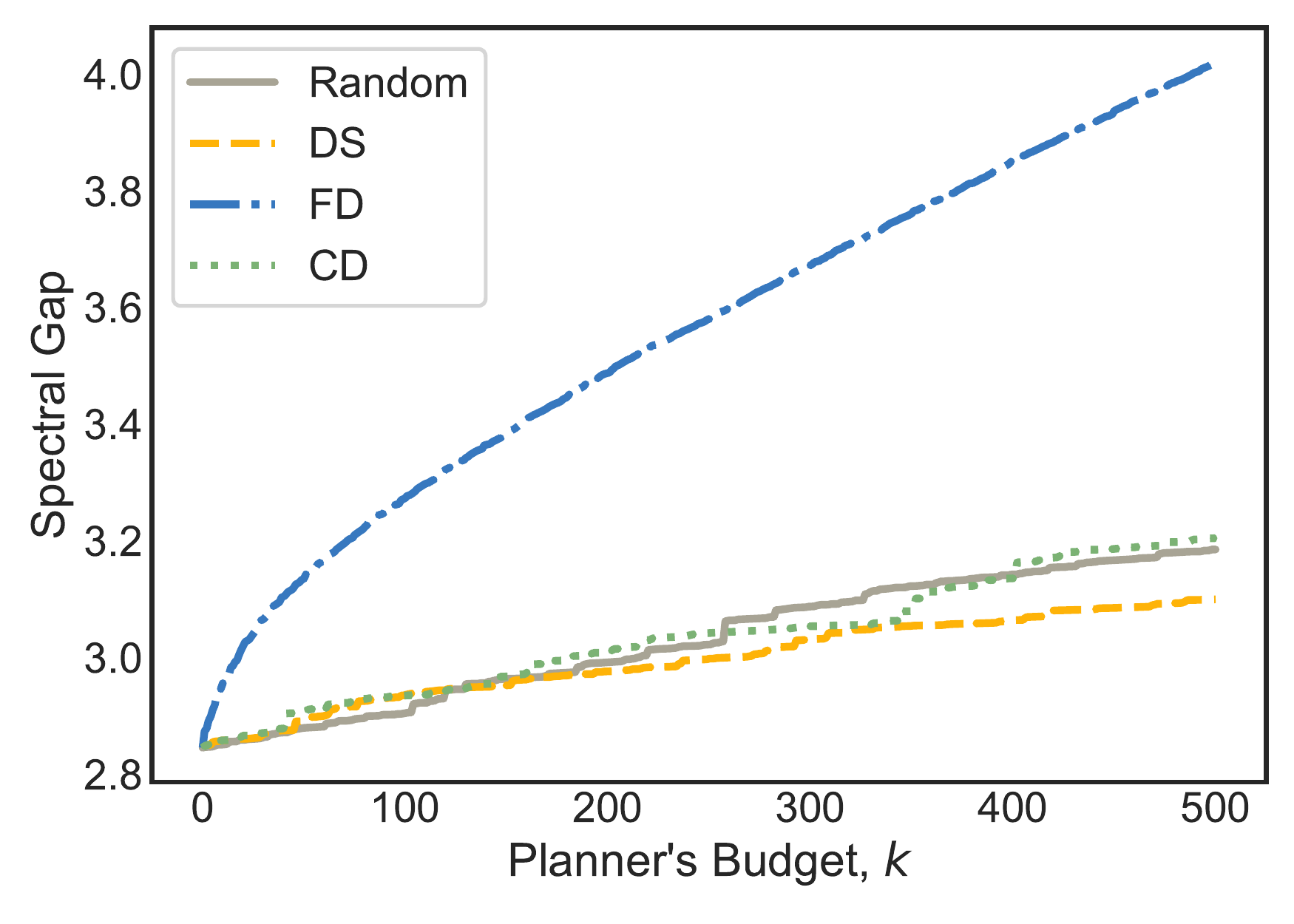}
    	\caption{Spectral gap} \label{fig:s_gap_pa}	
    \end{subfigure}
    \caption{Impacts of the planner's budget on the preferential attachment graph.}
\end{figure}

\end{document}